\documentclass[12pt]{article}
\usepackage[left=2.7cm,top=3.5cm,right=2cm,bottom=2.5cm]{geometry}

\newcommand{\del}{\nabla}
\newcommand{\bPsi}{\boldsymbol{\Psi}}
\newcommand{\bvPsi}{\boldsymbol{\varPsi}}

\newcommand{\br}{\boldsymbol{\textbf{r}}}

\newcommand{\bx}{\boldsymbol{\textbf{x}}}

\newcommand{\bH}{\boldsymbol{\textbf{H}}}

\newcommand{\bM}{\boldsymbol{\textbf{M}}}

\newcommand{\bR}{\boldsymbol{R}}

\newcommand{\phibar}{\bar{\phi}}
\newcommand{\psibar}{\bar{\psi}}
\newcommand{\Psibar}{\bar{\boldsymbol{\Psi}}}

\newcommand{\rhoin}[2]{\rho_{\text{in}}^{{#1}^{(#2)}}}
\newcommand{\rhoout}[2]{\rho_{\text{out}}^{{#1}^{(#2)}}}
\newcommand{\rhoinout}[2]{\rho_{\text{in(out)}}^{{#1}^{(#2)}}}
\newcommand{\epsilonbar}{\bar{\epsilon}_{i}}

\newcommand{\intomega}{\int_{\varOmega}}

\newcommand{\dx}{\,d\bx}
\newcommand{\dr}{\,d\br}
\newcommand{\btH}{\boldsymbol{\tilde{\textbf{H}}}}

\newcommand{\Rthree}{ {\mathbb{R}^{3}} }

\usepackage{amsmath}
\usepackage{graphicx}
\usepackage{graphics}
\usepackage{multirow}
\usepackage{fancyhdr, ifpdf}
\usepackage{amsxtra}
\usepackage{stmaryrd}
\usepackage{mathrsfs}
\usepackage{psfrag}
\usepackage{subfigure}
\usepackage{epsfig}
\usepackage{rotating}
\usepackage{setspace}
\usepackage{float}
\usepackage{amsfonts}
\usepackage{amsbsy}
\usepackage{amssymb}
\usepackage{amscd}
\usepackage{amsthm}
\usepackage{color}
\usepackage{algorithmic}
\usepackage{algorithm}
\usepackage{tabularx}
\usepackage{caption}
\usepackage{sidecap}
\newtheorem{prop}{Proposition}[section]

\definecolor{hellgruen}{rgb}{0.2,0.7,0.2}

\makeatletter
\renewcommand\subparagraph{\@startsection{subparagraph}{5}{\z@}{-3.25ex\@plus -1ex \@minus -.2ex}{0.0001pt \@plus .2ex}                                    {\normalfont\normalsize\bfseries}}
\makeatother
\date{}
\begin{document}
\setcounter{secnumdepth}{4}
\author{\emph{P. Motamarri\,\,$^a$,  M. R. Nowak\,\,$^b$,  K. Leiter\,\,$^c$, J. Knap\,\,$^c$, V. Gavini\,\,$^{a}\footnote{Corresponding author}$}\\
 \normalsize $^a$ Department of Mechanical Engineering, University of Michigan, Ann Arbor, MI 48109, USA \\
 \normalsize $^b$ Department of Electrical Engineering, University of Michigan, Ann Arbor, MI 48109, USA \\
 \normalsize $^c$ U.S. Army Research Labs, Aberdeen Proving Ground, Aberdeen, MD 21001, USA}
 \title{\LARGE{\textbf{Higher-order adaptive finite-element methods for Kohn-Sham density functional theory}}}
\maketitle
\abstract{We present an efficient computational approach to perform real-space electronic structure calculations using an adaptive higher-order finite-element discretization of Kohn-Sham density-functional theory (DFT). To this end, we develop an \emph{a priori} mesh adaption technique to construct a close to optimal finite-element discretization of the problem. We further propose an efficient solution strategy for solving the discrete eigenvalue problem by using spectral finite-elements in conjunction with Gauss-Lobatto quadrature, and a Chebyshev acceleration technique for computing the occupied eigenspace. The proposed approach has been observed to  provide a staggering $100-200$ fold computational advantage over the solution of a generalized eigenvalue problem. Using the proposed solution procedure, we investigate the computational efficiency afforded by higher-order finite-element discretizations of the Kohn-Sham DFT problem. Our studies suggest that staggering computational savings---of the order of $1000-$fold---relative to linear finite-elements can be realized, for both all-electron and local pseudopotential calculations, by using higher-order finite-element discretizations. On all the benchmark systems studied, we observe diminishing returns in computational savings beyond the sixth-order for accuracies commensurate with chemical accuracy, suggesting that the hexic spectral-element may be an optimal choice for the finite-element discretization of the Kohn-Sham DFT problem. A comparative study of the computational efficiency of the proposed higher-order finite-element discretizations suggests that the performance of finite-element basis is competing with the plane-wave discretization for non-periodic local pseudopotential calculations, and compares to the Gaussian basis for all-electron calculations to within an order of magnitude. Further, we demonstrate the capability of the proposed approach to compute the electronic structure of a metallic system containing 1688 atoms using modest computational resources, and good scalability of the present implementation up to 192 processors. }
\section{Introduction}
Electronic structure calculations have played a significant role in the investigation of materials properties over the past few decades. In particular, the Kohn-Sham approach to density functional theory (DFT)~\cite{kohn} has made quantum-mechanically informed calculations on ground-state materials properties computationally tractable, and has provided many important insights into a wide range of materials properties. The Kohn-Sham approach is based on the key result of Hohenberg \& Kohn~\cite{Hohenberg} that the ground-state properties of a materials system can be described by a functional of electron density. Though, the existence of an energy functional has been established by the Hohenberg-Kohn result, its functional form is not known to date. The work of Kohn \& Sham~\cite{kohn} addressed this challenge in an approximate sense, and has laid the foundations for the practical application of DFT to materials systems. The Kohn-Sham approach reduces the many-body problem of interacting electrons into an equivalent problem of non-interacting electrons in an effective mean field that is governed by the electron density. This effective single-electron description is exact in principle for ground-state properties, but is formulated in terms of an unknown exchange-correlation term that includes the quantum-mechanical interactions between electrons. This exchange-correlation term is approximated using various models---commonly modeled as an explicit functional of electron density---and these models have been shown to predict a wide range of materials properties across various materials systems. We note that the development of increasingly accurate and computationally tractable exchange-correlation functionals is an active research area in electronic structure calculations. Though the Kohn-Sham approach greatly reduces the computational complexity of the original many-body Schr\"{o}dinger problem, simulations of large-scale material systems with DFT are still computationally very demanding. Numerical algorithms which are robust, computationally efficient and scalable on parallel computing platforms are always desirable to enable DFT calculations at larger length and time scales, and on more complex systems, than possible heretofore.

The plane-wave basis has traditionally been one of the popular basis sets used for solving the Kohn-Sham problem~\cite{VASP,CASTEP,ABINIT}. The plane-wave basis allows for an efficient computation of the electrostatic interactions that are extended in real-space through Fourier transforms. However, the plane-wave basis also has some notable disadvantages. In particular, calculations are restricted to periodic geometries that are incompatible with most realistic systems containing defects (for e.g. dislocations). Further, the plane-wave basis provides a uniform spatial resolution which can be inefficient in the treatment of non-periodic systems like molecules, nano-clusters etc., or materials systems with defects, where higher basis resolution is often required in some spatial regions and a coarser resolution suffices elsewhere. Moreover, the plane-wave basis is non-local in real space, which significantly affects the scalability of computations on parallel computing platforms. On the other hand, atomic-orbital-type basis sets~\cite{Pople,WillisCooper,Jensen} have been been widely used for studying materials systems such as molecules and clusters. However, these basis sets are well suited only for isolated systems and cannot handle arbitrary boundary conditions. Furthermore, using these basis functions, it is difficult to achieve a systematic basis-set convergence for all materials systems. Due to the non-locality of these basis functions the efficiency of parallel scalability on a large number of processors is also affected.
Thus, the development of systematically improvable and scalable real-space techniques for electronic structure calculations has received significant attention over the past decade, and we refer to~\cite{Beck,SIESTA,PARSEC,OCTOPUS,BigDFT,CONQUEST,ONETEP,ACRES} and references therein for a comprehensive overview.  Among the real-space techniques, the finite-element basis presents some key advantages---it is amenable to unstructured coarse-graining, allows for consideration of complex geometries and boundary conditions, and is scalable on parallel computing platforms. We refer to~\cite{white,tsuchida1995,tsuchida1996,tsuchida1998,pask1999,pask2001,pask2005,Zhou2008,bylaska,lehtovaara,suryanarayana2010non,Lin,Zhou2012_JCP,Bao,Masud2012, Schauer_2013}, and references therein, for a comprehensive overview of the past efforts in developing real-space electronic structure calculations based on a finite-element discretization.

While the finite-element basis is more versatile than the plane-wave basis~\cite{white,pask1999}, it is not without its shortcomings. Prior investigations have shown that linear finite-elements require a large number of basis functions---of the order of $100,000$ basis functions per atom---to achieve chemical accuracy in electronic structure calculations~(cf.~e.g.~\cite{bylaska, Hermannson}), and this compares very poorly with plane-wave basis or other real-space basis functions. It has been demonstrated that higher-order finite-element discretizations can alleviate this degree of freedom disadvantage of linear finite-elements in electronic structure calculations~\cite{lehtovaara,Hermannson,batcho}. However, the use of higher-order elements increases the per basis-function computational cost due to the need for higher-order accurate numerical quadrature rules. Furthermore, the bandwidth of the matrix increases cubically with the order of the finite-element, which in turn increases the computational cost of matrix-vector products. In addition, since a finite-element basis is non-orthogonal, the discretization of the Kohn-Sham DFT problem results in a generalized eigenvalue problem, which is more expensive to solve in comparison to a standard eigenvalue problem resulting from using an orthogonal basis (for~e.g. plane-wave basis). Thus, the computational efficiency afforded by using a finite-element basis in electronic structure calculations, and its relative performance compared to plane-wave basis and atomic-orbital-type basis functions (for~e.g Gaussian basis), has remained an open question to date.

A recent investigation in the context of orbital-free DFT has indicated that the use of higher-order finite-elements can significantly improve the computational efficiency of the calculations~\cite{Motamarri2012}. For instance, a $100-1000$ fold computational advantage has been reported by using a fourth-order finite-element in comparison to a linear finite-element. In the present work, we extend this investigation to study the Kohn-Sham DFT problem and attempt to establish the computational efficiency afforded by higher-order finite-element discretizations in electronic structure calculations. To this end, we develop: (i) an \emph{a priori} mesh adaption technique to construct a close to optimal finite-element discretization of the problem; (ii) an efficient solution strategy for solving the discrete eigenvalue problem by using spectral finite-elements in conjunction with Gauss-Lobatto quadrature, and a Chebyshev acceleration technique for computing the occupied eigenspace. We subsequently study the numerical aspects of the finite-element discretization of the formulation, investigate the computational efficiency afforded by higher-order finite-elements, and compare the performance of the finite-element basis with plane-wave and Gaussian basis on benchmark problems.

The \emph{a priori} mesh adaption technique proposed in this work is based on the ideas in~\cite{levine1989,radio}, and closely follows the recent development of the mesh adaption technique for orbital-free DFT~\cite{Motamarri2012}. We refer to~\cite{Zhou2012_JCP,Bao} for recently proposed \emph{a posteriori} mesh adaption techniques in electronic structure calculations. The mesh adaption technique proposed in the present work is based on minimizing the discretization error in the ground-state energy, subject to a fixed number of elements in the finite-element mesh. To this end, we first develop an estimate for the finite-element discretization error in the Kohn-Sham ground-state energy as a function of the characteristic mesh-size distribution, $h(\br)$, and the exact ground-state electronic fields comprising of wavefunctions and electrostatic potential. We subsequently determine the optimal mesh distribution for the chosen representative solution by determining the $h(\br)$ that minimizes the discretization error. The resulting expressions for the optimal mesh distribution are in terms of the degree of the interpolating polynomial and the exact solution fields of the Kohn-Sham DFT problem. Since the exact solution fields are {\it a priori} unknown, we use the asymptotic behavior of the atomic wavefunctions~\cite{asymp} away from the nuclei to determine the coarse-graining rates for the finite-element meshes used in our numerical study. Though the resulting finite-element meshes are not necessarily optimal near the vicinity of the nuclei, the mesh coarsening rate away from the nuclei provides an efficient way of resolving the vacuum in non-periodic calculations.

We next implement an efficient solution strategy for solving the finite-element discretized eigenvalue problem, which is crucial before assessing the computational efficiency of the basis. We note that the non-orthogonality of the finite-element basis results in a discrete generalized eigenvalue problem, which is computationally more expensive than the standard eigenvalue problem that results from using an orthogonal basis like plane-waves. We address this issue by employing a spectral finite-element discretization and Gauss-Lobatto quadrature rules to evaluate the integrals which results in a diagonal overlap matrix, and allows for a trivial transformation to a standard eigenvalue problem. Further, we use the Chebyshev acceleration technique for standard eigenvalue problems to efficiently compute the occupied eigenspace (cf.~e.g.~\cite{cheby} in the context of electronic structure calculations). Our investigations suggest that the use of spectral finite-elements and Gauss-Lobatto rules in conjunction with Chebyshev acceleration techniques to compute the eigenspace gives a $10-20$ fold computational advantage, even for modest materials system sizes, in comparison to traditional methods of solving the standard eigenvalue problem where the eigenvectors are computed explicitly. Further, the proposed approach has been observed to provide a staggering $100-200$ fold computational advantage over the solution of a generalized eigenvalue problem that does not take advantage of the spectral finite-element discretization and Gauss-Lobatto quadrature rules. In our implementation, we use a self-consistent field (SCF) iteration with Anderson mixing~\cite{and}, and employ the finite-temperature Fermi-Dirac smearing~\cite{VASP} to suppress the charge sloshing associated with degenerate or close to degenerate eigenstates around the Fermi energy.

We next study various numerical aspects of the finite-element discretization of the Kohn-Sham DFT problem on benchmark problems involving both all-electron and pseudopotential calculation. Among pseudopotential calculations, in the present work, we restrict ourselves to local pseudopotentials as a convenient choice to demonstrate our ideas. We remark that the proposed methods are in no way restricted to local pseudopotentials, and an extension to non-local pseudopotentials is possible. We begin our investigation by examining the numerical rates of convergence of higher-order finite-element discretizations of Kohn-Sham DFT. We remark here that optimal rates of convergence have been demonstrated for quadratic hexahedral and cubic serendepity elements in pseudopotential Kohn-Sham DFT calculations~\cite{pask2005,Zhou2012_JCP}, and mathematically proved for Kohn-Sham DFT for the case of smooth pseudopotential external fields~\cite{zhou1}. We also note that there have been several works on the mathematical analysis of optimal rates of convergence for non-linear eigenvalue problems~\cite{zhou,Cances,zhou2}. However, the mathematical analysis of optimal rates of convergence of higher-order finite-element discretization of Kohn-Sham DFT problem involving Coulomb-singular potentials is an open question to date, to the best of our knowledge. In the present study, we compute the rates of convergence of the finite-element discretization of Kohn-Sham DFT for a range of finite-elements including linear tetrahedral element, hexahedral spectral-elements of order two, four and six. Two sets of benchmark problems are considered in this study: (i) all-electron calculations on boron atom and methane molecule; (ii) local pseudopotential calculations on a non-periodic barium cluster consisting of $2\times 2\times 2$ body-centered cubic (BCC) unit cells and a periodic face-centered cubic (FCC) calcium crystal. We note that our restriction in the present study to local pseudopotentials for pseudopotential calculations does not affect our conclusions on convergence rates, as demonstrated in~\cite{pask2005, Zhou2012_JCP} where non-local pseudopotentials were employed. In these benchmark studies, as well as those to follow, the proposed \emph{a priori} mesh adaption scheme is used to construct the meshes. These studies show rates of convergence in energy of $O(h^{2k})$ for a finite-element whose degree of interpolation is $k$, which denote optimal rates of convergence as demonstrated in~\cite{pask2005,zhou1}. An interesting aspect of this study is that optimal rates of convergence have been observed even for all-electron calculations involving Coulomb-singular potentials, which, to the best of our knowledge, have not been analyzed or reported heretofore for the Kohn-Sham problem. We note that, for Coulomb-singular potentials, in the context of orbital-free DFT optimal rates of convergence have been demonstrated in~\cite{zhou2} for $k=1,\,2$ and have been demonstrated numerically for up to $k=4$ in~\cite{Motamarri2012}. While the electrostatic interactions are common for both Kohn-Sham DFT and orbital-free DFT, the Kohn-Sham problem presents a more complex case as the approximation errors are governed by the entire occupied eigenspace of the Kohn-Sham problem as opposed to just the lowest eigenstate in the case of the orbital-free DFT problem.

We finally turn towards assessing the computational efficiency afforded by higher-order finite-element discretizations in Kohn-Sham DFT calculations. To this end, we use the same benchmark problems and measure the CPU-time for the solution of the Kohn-Sham DFT problem to various relative accuracies for all the finite-elements considered in the present study. We observe that higher-order elements can provide a significant computational advantage in the regime of chemical accuracy. The computational savings observed are about 1000-fold upon using higher-order elements in comparison with a linear finite-element for both all-electron as well as local pseudopotential calculations. We observe that a point of diminishing returns is reached at the sixth-order for the benchmark systems we studied and for accuracies commensurate with the chemical accuracy---i.e., no significant improvements in the computational efficiency was observed beyond this point. The degree of freedom advantage of higher-order finite-elements is nullified by the increasing per basis-function costs beyond this point. As demonstrated in Appendix B, the primary reason for the diminishing returns is the increase in the cost of computing the Hamiltonian matrix which also increasingly dominates the total time with increasing order of the element. To further assess the effectiveness of higher-order finite-elements, we conduct local pseudopotential calculations on large aluminium clusters ranging from $3\times 3\times 3$ to $5\times 5\times 5$ FCC unit cells using the hexic spectral finite-element, and compare the computational times with that of plane-wave basis discretization using the ABINIT package~\cite{ABINIT}. We note that similar relative accuracies in the ground-state energies are achieved using the hexic finite-element with a lower computational cost in comparison to the plane-wave basis. Furthermore, we computed the electronic structure of an aluminum cluster of $7\times 7\times 7$ FCC unit cells, containing 1688 atoms, with the finite-element basis using modest computational resources, which could not be simulated in ABINIT due to huge memory requirements. We note that, for isolated systems, superior computational efficiency of real-space techniques relative to plane-waves has also been previously demonstrated in~\cite{genovese} using a wavelet basis and in~\cite{PARSEC} using finite difference techniques.

We also investigate the efficiency of higher-order elements in the case of all-electron calculations on a larger materials systems and compare it with the Gaussian basis using the GAUSSIAN package~\cite{GAUSSIAN}. In this case, the benchmark systems considered are graphene sheet with 100 atoms and a tris (bipyridine) ruthenium complex with 61 atoms. We find that the solution times using the finite-element basis is larger by a factor of around $10$ in comparison to Gaussian basis, and we attribute this difference to the highly optimized Gaussian basis functions for specific atom types that resulted in the far fewer basis functions required to achieve chemical accuracy. While this difference in the performance can be offset by the better scalability of finite-element discretization on parallel computing platforms, there is also much room for further development and optimization in the finite-element discretization of the Kohn-Sham DFT problem. For instance, especially in the context of all-electron calculations, the partitions-of-unity finite-element method with atomic-orbital enrichment functions can significantly reduce the required number of finite-element basis functions as recently demonstrated in~\cite{Suku}, and presents an important direction for further investigations. Finally, we assess the parallel scalability of our numerical implementation. We demonstrate the strong scaling up to $192$ processors (limited by our access to computing resources) with an efficiency of $91.4\%$ using a $172$ atom aluminium cluster discretized with $3.91$ million degrees of freedom as our benchmark system.

The remainder of the paper is organized as follows. Section~\ref{formulation} describes the variational formulation of the Kohn-Sham DFT problem followed by a discussion on the discrete Kohn-Sham DFT eigenvalue problem. Section~\ref{meshadapt} develops the error estimates for the finite-element discretization of Kohn-Sham DFT, and uses these estimates to present an \emph{a priori} mesh adaption scheme. Section~\ref{implement} describes our numerical implementation of the self-consistent field iteration of the Kohn-Sham eigenvalue problem, and, in particular, discusses the efficient methodologies developed to solve the Kohn-Sham DFT problem using the finite-element basis. Section~\ref{results} presents a comprehensive numerical study demonstrating the computational efficiency afforded by higher-order finite-element discretizations in electronic structure calculations, and also provides a performance comparison of finite-element basis with plane-wave and Gaussian basis. We finally conclude with a short discussion and outlook in Section~\ref{concl}.

\section{Formulation}\label{formulation}
In this section, we describe the Kohn-Sham DFT energy functional and the associated variational formulation. We subsequently review the equivalent self-consistent formulation of the Kohn-Sham eigenvalue problem, and present the discretization of the formulation using a finite-element basis.
\subsection{Kohn-Sham variational problem}\label{KSVariationalFormulation}
We consider a material system consisting of $N$ electrons and $M$ nuclei. The spinless Kohn-Sham energy functional describing the $N$ electron system is given by~\cite{Parr1989,Finnis}
\begin{equation}\label{ksEnergy}
E(\bPsi,\bR) = T_s({\bPsi}) + E_{xc}(\rho)+  E_{H}(\rho) + E_{ext}(\rho,\bR) + E_{zz}(\bR),
\end{equation}
where
\begin{equation}\label{ksRho}
\rho(\br) = \sum_{i=1}^{N}|\psi_{i}(\bx)|^2
\end{equation}
represents the electron density. In the above expression, we denote the spatial coordinate by $\br$, whereas $\bx=(\br,s)$ includes both the spatial and spin degrees of freedom. We denote by $\bPsi = \{\psi_{1}(\bx),\psi_{2}(\bx),\cdots,\psi_{N}(\bx)\}$ the vector of orthonormal single electron wavefunctions, where each wavefunction $\psi_{i}\in \mathbb{X}\times\{\alpha,\beta\}$ can in general be complex-valued, and comprises of a spatial part belonging to a suitable function space $\mathbb{X}$ (elaborated subsequently) and a spin state denoted by $\alpha(s)$ or $\beta(s)$. We further denote by $\bR = \{\bR_1,\bR_2, \cdots \bR_M\}$ the collection of all nuclear positions. The first term in the Kohn-Sham energy functional in~\eqref{ksEnergy}, $T_s({\bPsi})$, denotes the kinetic energy of non-interacting electrons and is given by
\begin{equation} \label{ke}
T_{s}(\bPsi) = \sum_{i=1}^{N}\int \psi_{i}^{*}(\bx) \left(-\frac{1}{2} \del^{2}\right) \psi_{i}(\bx)\dx\,\,,
\end{equation}
where $\psi_i^*$ denotes the complex conjugate of $\psi_i$. $E_{xc}$ in the energy functional denotes the exchange-correlation energy which includes the quantum-mechanical many body interactions. In the present work, we model the exchange-correlation energy using the local density approximation (LDA)~\cite{alder,perdew} represented as
\begin{equation}\label{exc}
E_{xc}(\rho) = \int \varepsilon_{xc}(\rho)\rho(\br) \dr\,\,,
\end{equation}
where $\varepsilon_{xc}(\rho)=\varepsilon_x(\rho)+\varepsilon_c(\rho)$, and
\begin{equation}
\varepsilon_x(\rho) = -\frac{3}{4}\left(\frac{3}{\pi}\right)^{1/3}\rho^{1/3}(\br) \,\,,
\end{equation}
\begin{equation}
\varepsilon_c(\rho) = \begin{cases}
&\frac{\gamma}{(1 + \beta_1\sqrt(r_s) + \beta_2r_s)}\;\;\;\;\;\;\;\;\;\;\;\;\;\;\;\;\;\;\;\;\;\;\;r_s\geq1,\\
&A\,\log r_s + B + C\,r_s\log r_s + D\,r_s\;\;\;\;\;\;\;\;r_s\,<\,1,
\end{cases}
\end{equation}
and $r_s = (3/4\pi\rho)^{1/3}$. Specifically, we use the Ceperley and Alder constants as given in \cite{perdew}. We remark that we have restricted the present formulation and study to LDA exchange-correlation functionals solely for the sake of clarity, and the formulation can be trivially extended (cf. e.g.~\cite{suryanarayana2010non}) to local spin density approximation (LSDA) and generalized gradient approximation (GGA) exchange-correlation functionals.

The electrostatic interaction energies in the Kohn-Sham energy functional in~\eqref{ksEnergy} are given by
\begin{equation}\label{hartree}
E_{H}(\rho) = \frac{1}{2}\int\int\frac{\rho(\br)\rho(\br')}{|\br - \br'|} \,\dr\,\dr'\,,
\end{equation}
\begin{equation}\label{external}
E_{ext}(\rho,\bR) = \int \rho(\br) V_{ext}(\br,\bR) \dr = \sum_{J}\int \rho(\br) V_{J}(\br,\bR_J) \dr\,,
\end{equation}
\begin{equation}\label{repulsive}
E_{zz} = \frac{1}{2}\sum_{\substack{I,J \neq I}} \frac{Z_I Z_J}{|\bR_I-\bR_J|}\,,
\end{equation}
where $E_H$ is the Hartree energy representing the classical electrostatic interaction energy between electrons, $E_{ext}$ is the interaction energy between electrons and the external potential induced by the nuclear charges given by $V_{ext}=\sum_{J}V_{J}(\br,\bR_J)$ with $V_J$ denoting the potential (singular Coulomb potential or local pseudopotential) contribution from the $J^{th}$ nucleus, and $E_{zz}$ denotes the repulsive energy between nuclei with $Z_I$ denoting the charge on the $I^{th}$ nucleus. We note that in a non-periodic setting, representing a finite atomic system, all the integrals in equations~\eqref{ke}-\eqref{external} are over $\Rthree$ and the summations in~\eqref{external}-\eqref{repulsive} include all the atoms $I$ and $J$ in the system. In the case of an infinite periodic crystal, all the integrals over $\br$ in equations~\eqref{ke}-\eqref{external} extend over the unit cell, whereas the integrals over $\br'$ extend in $\Rthree$. Similarly, in~\eqref{external}-\eqref{repulsive} the summation over $I$ is on the atoms in the unit cell, and summation over $J$ extends over all lattice sites. We note that, in the context of periodic problems, the above expressions assume a single k-point ($\Gamma-$point) sampling. The computation of the electron density and kinetic energy in~\eqref{ksRho} and \eqref{ke} for multiple k-point sampling involves an additional quadrature over the k-points in the Brillouin zone (cf.~e.g \cite{mermin}).

The electrostatic interaction terms as expressed in equations~\eqref{hartree}-\eqref{repulsive} are nonlocal in real-space, and, for this reason, evaluation of electrostatic energy is the computationally expensive part of the calculation. Following the approach in~\cite{ofdft,suryanarayana2010non}, the electrostatic interaction energy can be reformulated as a local variational problem in electrostatic potential by observing that $\frac{1}{|\br|}$ is the Green's function of the Laplace operator. To this end, we represent the nuclear charge distribution by $b(\br,\bR) = -\sum \limits_{I=1}^{M} Z_{I}\tilde{\delta}_{{\bR_I}}(\br)$, where $Z_I\tilde{\delta}_{{\bR_I}}(\br)$ represents a bounded smooth charge distribution centered at $\bR_I$, either corresponding to a local pseudopotential, or, in the case of all-electron calculations, a regularization of the point charge having a support in a small ball around $\bR_I$ with charge $Z_I$. The nuclear repulsion energy can subsequently be represented as
 \begin{equation}
E_{zz}(\bR) = \frac{1}{2}\int\int\frac{b(\br,\bR)b(\br',\bR)}{|\br - \br'|}\dr\dr'\,.
\end{equation}
We remark that, while this differs from the expression in equation~\eqref{repulsive} by the self-energy of the nuclei, the self-energy is an inconsequential constant depending only on the nuclear charge distribution, and is explicitly evaluated and subtracted from the total energy in numerical computations (cf. Appendix A). Subsequently, the electrostatic interaction energy, up to a constant self-energy, is given by the following variational problem:
\begin{align}\label{elReformulation}
&\frac{1}{2}\int\int\frac{\rho(\br)\rho(\br')}{|\br - \br'|}\dr\dr' +
\int \rho(\br) V_{ext}(\br) \dr + \frac{1}{2}\int\int\frac{b(\br,\bR)b(\br',\bR)}{|\br - \br'|}\dr\dr' \notag\\
&= -\inf_{\phi \in \mathcal{Y}} \left\{\frac{1}{8\pi}\int |\del \phi(\br)|^2 \dr - \int (\rho(\br) + b(\br,\bR))\phi(\br)\dr\right\},
\end{align}
where $\phi(\br)$ denotes the trial function for the total electrostatic potential due to the electron density and the nuclear charge distribution and $\mathcal{Y}$ is a suitable function space discussed subsequently.

Using the local reformulation of electrostatic interactions, the Kohn-Sham energy functional~\eqref{ksEnergy} can be rewritten as
\begin{equation}\label{KS_El_reformulation}
E(\bPsi,\bR) = \sup_{\phi \in \mathcal{Y}} L(\phi,\bPsi,\bR) \,\,,
\end{equation}
where
\begin{equation}\label{KS_Lagrangian}
L(\phi,\bPsi,\bR) = T_{s}(\bPsi) + E_{xc}(\rho) - \frac{1}{8\pi}\int |\del \phi(\br)|^2 \dr + \int (\rho(\br) + b(\br,\bR))\phi(\br)\dr\,\,.
\end{equation}
Subsequently, the problem of determining the ground-state energy and electron density for given positions of nuclei can be expressed as the following variational problem:
\begin{equation}\label{inf_problem}
\inf_{\bPsi \in \mathcal{X}} E(\bPsi,\bR)\,\,,
\end{equation}
where $\mathcal{X}=\big\{\bPsi\,|\,\langle\psi_i,\psi_j\rangle_{\mathbb{X}\times\{\alpha,\beta\}}=\delta_{ij}\big\}$ with $\langle\;,\;\rangle_{\mathbb{X}\times\{\alpha,\beta\}}$ denoting the inner product defined on $\mathbb{X}\times\{\alpha,\beta\}$. $\mathbb{X}$ denotes a suitable function space that guarantees the existence of minimizers. We note that bounded domains are used in numerical computations, which in non-periodic calculations corresponds to a large enough domain containing the compact support of the wavefunctions and in periodic calculations correspond to the supercell. We denote such an appropriate bounded domain, subsequently, by $\Omega$. For formulations on bounded domains, $\mathbb{X} = \mathcal{Y} = H^1_0(\Omega)$ in the case of non-periodic problems and $\mathbb{X} = \mathcal{Y} = H^1_{per}(\Omega)$ in the case of periodic problems are appropriate function spaces which guarantee existence of solutions (cf. e.g.~\cite{suryanarayana2010non}). Mathematical analysis of the Kohn-Sham DFT problem proving the existence of solutions in the more general case of $\mathbb{R}^3$ ($\mathbb{X} = H^1(\mathbb{R}^3)$) has recently been reported~\cite{Cances1}.

\subsection{Kohn-Sham eigenvalue problem}
The stationarity condition corresponding to the Kohn-Sham variational problem is equivalent to the non-linear Kohn-Sham eigenvalue problem given by:
\begin{equation}\label{scfEigen}
\mathcal{H} {\psi}_{i} = \epsilon_{i} {\psi}_{i},
\end{equation}
where
\begin{equation}\label{scfEigen1}
\mathcal{H}  = \left(-\frac{1}{2} \del^{2} + V_{\text{eff}}(\rho,\bR)\right)
\end{equation}
is a Hermitian operator with eigenvalues $\epsilon_{i}$, and the corresponding orthonormal eigenfunctions $\psi_{i}$ for $i = 1,2,\cdots, N$ denote the canonical wavefunctions. The electron density in terms of the canonical wavefunctions is given by
\begin{equation}
\rho(\br) = \sum_{i=1}^{N}|\psi_i(\bx)|^2 \,\,,
\end{equation}
and the effective single-electron potential, $V_{\text{eff}}(\rho,\bR)$, in~\eqref{scfEigen1} is given by
\begin{equation}
 V_{\text{eff}}(\rho,\bR) = V_{ext}(\bR) + V_{H}(\rho) +V_{xc}(\rho) = V_{ext}(\bR) +\frac{\delta E_{H}}{\delta \rho} + \frac{\delta E_{xc}}{\delta \rho}\,\,.
\end{equation}
As discussed in Section~\ref{KSVariationalFormulation}, it is efficient to compute the total electrostatic potential, defined as the sum of the external potential ($V_{ext}(\bR)$) and the Hartree potential ($V_{H}(\rho)$), through the solution of a Poisson equation
\begin{equation*}
-\frac{1}{4\pi}\del^2\phi(\br,\bR) = \rho(\br) + b(\br,\bR) \,\,,
\end{equation*}
which is given by
\begin{equation}
\phi(\br,\bR) \equiv V_{H}(\rho) + V_{ext}(\br,\bR) = \int \frac{\rho(\br')}{|\br - \br'|} \dr' + \int \frac{b(\br',\bR)}{|\br - \br'|} \dr'\,\,.
\end{equation}
Finally, the system of equations corresponding to the Kohn-Sham eigenvalue problem are given by
\begin{subequations}\label{ksproblem}
\begin{gather}
\left(-\frac{1}{2} \del^{2} +  V_{\text{eff}}(\rho,\bR) \right){\psi}_{i} = \epsilon_{i} {\psi}_{i},\\
\rho(\br) = \sum_{i=1}^{N}|\psi_{i}(\bx)|^2,\\
-\frac{1}{4\pi}\del^2\phi(\br,\bR) = \rho(\br) + b(\br,\bR)\,\,;\qquad V_{\text{eff}}(\rho,\bR) = \phi(\br,\bR) + \frac{\delta E_{xc}}{\delta \rho}\,\,,
\end{gather}
\end{subequations}
which have to be solved with appropriate boundary conditions based on the problem under consideration. In the case of a periodic crystal, the effective potential $V_{\text{eff}}$ has the periodicity of the lattice and the solutions of the Kohn-Sham eigenvalue problem are given by the Bloch theorem~\cite{mermin}. Thus, for periodic systems, it is computationally efficient to compute the Bloch solutions directly. The formulation in~\eqref{ksproblem} represents a nonlinear eigenvalue problem which has to be solved self-consistently, and is subsequently discussed in Section~\ref{implement}. Next we turn to the discrete formulation of the above Kohn-Sham eigenvalue problem.

\subsection{Discrete Kohn-Sham eigenvalue problem}
If $X_h$ represents the finite-dimensional subspace with dimension $n_h$, the finite-element approximation of the various field variables (spatial part of the wavefunctions and the electrostatic potential) in the Kohn-Sham eigenvalue problem~\eqref{ksproblem} are given by
\begin{equation}\label{fem}
\psi^{h}_{i}(\br) = \sum_{j=1}^{n_h} N_j(\br) \psi^{j}_{i}\,\,,
\end{equation}
\begin{equation}\label{fem_phi}
\phi^{h}(\br) = \sum_{j=1}^{n_h} N_j(\br) \phi^{j}\,\,,
\end{equation}
where $N_{j}:1\leq j \leq n_h$ denote the basis of $X_h$. Subsequently, the discrete eigenvalue problem corresponding to~\eqref{ksproblem} is given by
\begin{equation}\label{ghep}
\bH \tilde{\bvPsi}_{i} = \epsilon^{h}_{i} \bM \tilde{\bvPsi}_{i}\,\,,
\end{equation}
where $\text{H}_{jk}$ denotes the discrete Hamiltonian matrix, $\text{M}_{jk}$ denotes the overlap matrix (or commonly referred to as the mass matrix in finite-element literature), and $\epsilon^{h}_{i}$ denotes $i^{th}$ eigenvalue corresponding to the eigenvector $\tilde{\bvPsi}_i$. The expression for the discrete Hamiltonian matrix $\text{H}_{jk}$ for a non-periodic problem with $\mathbb{X}=\mathcal{Y}=H^1_0(\Omega)$ as well as a periodic problem on a supercell with $\mathbb{X}=\mathcal{Y}=H^1_{per}(\Omega)$ is given by
\begin{equation}
\text{H}_{jk} = \frac{1}{2} \intomega \del N_{j} (\br)\, . \del N_k(\br) \dr + \intomega V^{h}_{\text{eff}}(\br,\bR) N_{j}(\br) N_{k}(\br) \dr \,\,.
\end{equation}
We refer to~\cite{pask2005} for the expression of $\text{H}_{jk}$ in the case of a periodic problem on a unit-cell using the Bloch theorem. The discrete Kohn-Sham eigenvalue problem~\eqref{ghep} is a generalized eigenvalue problem with an overlap matrix $\text{M}_{jk}=\intomega N_j(\br) N_k(\br) \dr$, which results from the non-orthogonality of the finite-element basis functions. However, the generalized eigenvalue problem~\eqref{ghep} can be transformed into a standard Hermitian eigenvalue problem as follows. Since the matrix $\bM$ is positive definite symmetric, there exists a unique positive definite symmetric square root of $\bM$, and is denoted by $\bM^{1/2}$. Hence, the following holds true
\begin{align}
\bH \tilde{\bvPsi}_{i} &= \epsilon^{h}_{i} \bM \tilde{\bvPsi}_{i} \nonumber \\
\Rightarrow\qquad \bH \tilde{\bvPsi}_{i} &= \epsilon^{h}_{i} \bM^{1/2} \bM^{1/2} \tilde{\bvPsi}_{i} \nonumber \\
\Rightarrow\qquad \btH \hat{\bvPsi}_{i} &= \epsilon^{h}_{i} \hat{\bvPsi}_{i} \label{hep}
\end{align}
where
\begin{align*}
\hat{\bvPsi}_{i} &= \bM^{1/2} \tilde{\bvPsi}_{i}\\
\btH &= \bM^{-1/2}\bH\bM^{-1/2}
\end{align*}
We note that $\btH$ is a Hermitian matrix, and~\eqref{hep} represents a standard Hermitian eigenvalue problem. The actual eigenvectors are recovered by the transformation $\tilde{\bvPsi}_{i} = \bM^{-1/2} \hat{\bvPsi}_{i}$. We remark that $\hat{\bvPsi}_{i}$ is a vector containing the expansion coefficients of the eigenfunction $\psi^{h}_i(\br)$ expressed in an orthonormal basis spanning the finite-element space. Furthermore, we note that the transformation to a standard eigenvalue problem~\eqref{hep} is computationally advantageous only if the matrix $\bM^{-1/2}$ can be evaluated with modest computational cost. This is readily possible by using spectral finite-elements rather than conventional finite-elements, and is discussed in detail in Section~\ref{implement}.

The convergence of finite-element approximation for the Kohn-Sham DFT model was shown in~\cite{suryanarayana2010non} using the notion of $\Gamma-$convergence. We also refer to the recent numerical analysis carried out on finite dimensional discretization of Kohn-Sham models~\cite{zhou1}, which also provides the rates of convergence of the approximation for pseudopotential calculations. We remark that in the present work we use the same finite-element discretization for both electronic wavefunctions and electrostatic potential, as is evident from equations~\eqref{fem}-\eqref{fem_phi}. Since the electrostatic potential has similar discretization errors as compared to the electronic wavefunctions and since the Kohn-Sham DFT problem is a saddle-point problem in electronic wavefunctions and electrostatic potential (see equations~\eqref{KS_El_reformulation}-\eqref{inf_problem}, also cf.~\cite{suryanarayana2010non}) the convergence of the finite-element discretization error is non-variational in general. We note, however, that by using a more refined discretization (h-refinement) or by using a higher-order polynomial (p-refinement) as in~\cite{Schauer_2013} for the discretization of electrostatic fields in comparison to the discretization of electronic wavefunctions, this drawback can be mitigated. Next, we derive the optimal coarse-graining rates for the finite-element meshes using the solution fields in the Kohn-Sham DFT problem.

\section{A-\emph{priori} mesh adaption}\label{meshadapt}
We propose an \emph{a priori} mesh adaption scheme in the spirit of \cite{radio,levine1989} by minimizing the error involved in the finite-element approximation of the Kohn-Sham DFT problem for a fixed number of elements in the mesh. The proposed approach closely follows the \emph{a priori} mesh adaption scheme developed in the context of orbital-free DFT~\cite{Motamarri2012}. In what follows, we first derive a formal bound on the energy error $|E-E_h|$ as a function of the characteristic mesh-size $h$,  and the distribution of electronic fields (wavefunctions and electrostatic potential). We note that, in a recent study, error estimates for a generic finite dimensional approximation of the Kohn-Sham model have been derived~\cite{zhou1}. However, the forms of these estimates are not useful for developing mesh-adaption schemes as the study primarily focused on proving the convergence of the finite-dimensional approximation and determining the convergence rates. We first present the derivation of an error bound in terms of the canonical wavefunctions and the electrostatic potential, and subsequently develop an \emph{a priori} mesh adaption scheme based on this error bound.

\subsection{Estimate of energy error}
In the present section and those to follow, we demonstrate our ideas on a system consisting of $2N$ electrons for the sake of simplicity and notational clarity. Let ($\Psibar^{h} = \{\psibar_{1}^{h}\,,\,\psibar_2^{h} \,\cdots\,\psibar_{N}^h\}\,, \phibar^{h}\,, \mathbf{\bar{\epsilon}}^h=\{\bar{\epsilon}_1^{h}\,,\bar{\epsilon}_2^{h}\,\cdots\,\bar{\epsilon}_N^{h}\}$) and ($\Psibar = \{\psibar_{1}\,,\,\psibar_2 \,\cdots\,\psibar_{N}\} , \phibar\,, \mathbf{\bar{\epsilon}}=\{\bar{\epsilon}_1\,,\bar{\epsilon}_2\,\cdots\,\bar{\epsilon}_N\}$) represent the solutions (spatial part of canonical wavefunctions, electrostatic potential, eigenvalues) of the discrete finite-element problem \eqref{ghep} and the continuous problem \eqref{ksproblem} respectively. In the following derivation and henceforth in this article, we consider all wavefunctions to be real-valued and orthonormal. We note that it is always possible to construct real-valued orthonormal wavefunctions for both non-periodic problems as well as periodic problems on the supercell. The wavefunctions are complex-valued for periodic problems on a unit-cell (with multiple k-points using the Bloch theorem), and the following approach is still valid, but results in more elaborate expressions for the error bounds. Using the local reformulation of electrostatic interactions in the Kohn-Sham energy functional (equations \eqref{KS_El_reformulation}-\eqref{KS_Lagrangian}), the ground-state energy in the discrete and the continuous problem can be expressed as:
\begin{equation}
E_h(\Psibar^{h},\phibar^{h}) = 2\sum_{i=1}^{N} \intomega \frac{1}{2} |\del \psibar_{i}^{h}|^{2} \dr + \intomega F(\rho(\Psibar^{h})) \dr - \frac{1}{8\pi} \intomega |\del \phibar^{h}| ^{2} \dr + \intomega (\rho(\Psibar^{h}) \; +\; b)\phibar^h \dr\,,
\end{equation}
\begin{equation}
E(\Psibar,\phibar) = 2\sum_{i=1}^{N} \intomega \frac{1}{2} |\del \psibar_{i}|^{2} \dr + \intomega F(\rho(\Psibar)) \dr - \frac{1}{8\pi} \intomega |\del \phibar| ^{2} \dr + \intomega (\rho(\Psibar) \; +\; b)\phibar \dr\,,
\end{equation}
where
\begin{gather*}
F(\rho) = \epsilon_{xc}(\rho)\rho \,\,.
\end{gather*}

\begin{prop}\label{prop1}
In the neighborhood of ($\Psibar, \phibar\,, \mathbf{\bar{\epsilon}}$), the finite-element approximation error in the ground-state energy can be bounded as follows:
\end{prop}
\begin{equation}
 \begin{split}
 |E_h - E| &\leq 2\sum_{i=1}^{N} \Bigl[\frac{1}{2}\intomega |\del \delta \psi_{i}|^{2}  \dr + \left|\; \epsilonbar \intomega (\delta \psi_{i})^2 \dr \right| + \left|\intomega F'(\rho(\Psibar)) (\delta \psi_{i} )^{2} \dr \right| \\
 &+ \left|\intomega (\delta \psi_{i})^2\phibar \dr \right|+ 2 \left|\intomega \psibar_{i}\; \delta \psi_{i} \;\delta \phi \dr\right|\Bigr]   + \frac{1}{8\pi} \intomega |\del \delta \phi|^2 \dr\\
  &+ 8 \left|\intomega F''(\rho(\Psibar)) \left(\sum_{i} \psibar_{i} \delta \psi_{i}\right)^{2}\dr\right|\,\,\,.
 \end{split}
 \end{equation}

\begin{proof}
We first expand $E_h(\Psibar^{h},\phibar^{h}) $ about the solution of the continuous problem, i.e $\Psibar^{h} = \Psibar + \delta \bPsi$ and $\phibar^{h} = \phibar + \delta \phi $, and we get
\begin{equation}
\begin{split}
E_h(\Psibar + \delta \bPsi, \phibar + \delta \phi) =  2\sum_{i=1}^N &\intomega \frac{1}{2} |\del (\psibar_{i} + \delta \psi_{i})|^{2} \dr + \intomega F\left(\rho(\Psibar + \delta \bPsi)\right) \dr\\
 &- \frac{1}{8\pi} \intomega |\del (\phibar + \delta \phi)|^{2}\dr + \intomega \left(\rho(\Psibar + \delta \bPsi)  \; +\; b\right)(\phibar + \delta \phi) \dr\,,
\end{split}
\end{equation}
which can then be simplified, using the Taylor series expansion, to
\begin{equation}\label{taylor}
\begin{split}
&E_h(\Psibar^{h}, \phibar^{h}) = 2\sum_{i=1}^{N} \intomega \frac{1}{2} (|\del \psibar_{i} |^2  + |\del \delta \psi_{i}|^{2}  + 2 \del \psibar_{i} \cdot \del \delta \psi_{i})\dr + \intomega F(\rho(\Psibar)) \dr \\
 &+ 4\sum_{i=1}^{N}\intomega F'(\rho(\Psibar)) \psibar_{i} \;\delta \psi_{i} \dr
+ 8\intomega F''(\rho(\Psibar)) \left(\sum_{i=1}^{N} \psibar_{i}\; \delta \psi_{i}\right)^{2} \dr  + 2\sum_{i=1}^{N}\intomega F'(\rho(\Psibar)) (\delta \psi_{i} )^{2} \dr \\
 &-\frac{1}{8\pi} \intomega \left( |\del \phibar|^2 + |\del \delta \phi|^2 + 2\del \phibar\cdot  \del \delta \phi\right)\dr + \intomega (\rho(\Psibar) + b)\phibar\; \dr  + 4\sum_{i=1}^{N} \intomega \psibar_{i}\; \delta \psi_{i}\; \phibar\; \dr \\
 &+ \intomega (\rho(\Psibar) + b)\delta \phi\; \dr  +  2\sum_{i=1}^{N} \intomega (\delta \psi_{i})^2\phibar\; \dr +  4 \sum_{i=1}^{N} \intomega \psibar_{i}\; \delta \psi_{i} \;\delta \phi\; \dr + O(\delta \psi_{i}^{3}, \delta \psi_i^{2} \delta \phi)\,.
\end{split}
\end{equation}
We note that ($\Psibar, \phibar\,, \mathbf{\bar{\epsilon}}$) satisfy the following Euler-Lagrange equations for each $i=1,\ldots,N$.
\begin{subequations}\label{euler}
\begin{gather}
 \frac{1}{2}\intomega \del \psibar_{i}\cdot  \del \delta \psi_{i} \dr + \intomega F'(\rho(\Psibar)) \psibar_{i} \;\delta \psi_{i} \dr + \intomega \psibar_{i}\; \delta \psi_{i}\; \phibar\; \dr =  \epsilonbar\intomega \; \psibar_{i}\; \delta \psi_{i} \dr\,\,,  \\
-\frac{1}{4\pi} \intomega \del \phibar \cdot \del \delta \phi\; \dr + \intomega (\rho(\Psibar) + b)\delta \phi\; \dr  = 0\,\,.
\end{gather}
\end{subequations}
Using~\eqref{taylor} and the Euler-Lagrange equations~\eqref{euler}, we get
\begin{equation}\label{error1}
\begin{split}
E_h - E &= 2\sum_{i=1}^{N} \intomega \left[ \frac{1}{2} |\del \delta \psi_{i}|^{2}  + 2 \; \epsilonbar  \; \psibar_{i}\; \delta \psi_{i} + F'(\rho(\Psibar)) (\delta \psi_{i} )^{2}\right] \dr + 8\intomega F''(\rho(\Psibar)) \left(\sum_{i=1}^{N} \psibar_{i} \delta \psi_{i}\right)^{2} \dr \\
&- \frac{1}{8\pi} \intomega |\del \delta \phi|^2 \dr + 2\sum_{i=1}^{N} \left[\intomega (\delta \psi_{i})^2\phibar\; \dr + 2\intomega \psibar_{i}\; \delta \psi_{i} \;\delta \phi\; \dr\right] + O(\delta \psi_{i}^{3}, \delta \psi_i^{2} \delta \phi) \,\,.
\end{split}
\end{equation}
The orthonormality constraint functional in the discrete form is given by
\begin{equation}
c(\bPsi^{h}) = \intomega \psi^{h}_{i} \psi^{h}_{j} \dr - \delta_{ij}\,\,,
\end{equation}
and upon expanding about the solution $\Psibar$, we get
\begin{align}
c(\Psibar^{h}) &= \intomega (\psibar_{i} + \delta \psi_{i}) (\psibar_{j} + \delta \psi_{j}) \dr - \delta_{ij}\,\\
& = \intomega \left[\psibar_{i} \psibar_{j} + \delta \psi_{i} \psibar_{j}  + \delta \psi_{j} \psibar_{i} + \delta \psi_{i} \delta \psi_{j}\right]\dr - \delta_{ij} \label{orthonormality_constraint}\,\,.
\end{align}
Using
\begin{equation}
\intomega\psibar_{i} \psibar_{j}\dr = \delta_{ij}\,\,,
\end{equation}
and $c(\Psibar^{h}) = 0$ in \eqref{orthonormality_constraint}, we get for $i=j$
\begin{equation}\label{constraint}
2 \intomega \psibar_{i} \delta \psi_{i} \dr = - \intomega (\delta \psi_{i})^2 \dr \qquad i=1,2,\ldots,N\,\,.
\end{equation}
Using equations \eqref{error1} and \eqref{constraint}, we arrive at the following error bound in energy
 \begin{equation}
 \notag
 \begin{split}
 |E_h - E| &\leq 2\sum_{i=1}^{N} \Bigl[\frac{1}{2}\intomega |\del \delta \psi_{i}|^{2}  \dr + \left|\; \epsilonbar \intomega (\delta \psi_{i})^2 \dr \right| +\left| \intomega F'(\rho(\Psibar)) (\delta \psi_{i} )^{2} \dr \right| \\
 & + \left|\intomega (\delta \psi_{i})^2\phibar\; \dr \right|+ 2 \left|\intomega \psibar_{i}\; \delta \psi_{i} \;\delta \phi\; \dr\right|\Bigr]  +  \frac{1}{8\pi} \intomega |\del \delta \phi|^2 \dr \\
  &+ 8\left|\intomega F''(\rho(\Psibar)) \left(\sum_{i} \psibar_{i} \delta \psi_{i}\right)^{2}\dr\right|\,.
 \end{split}
 \end{equation}
 \end{proof}
 \begin{prop}\label{prop2}
The finite-element approximation error in proposition~\ref{prop1} expressed in terms of the approximation errors in electronic wave-functions and electrostatic potential is given by
\begin{equation}
|E_h - E| \leq C \left(\sum_{i}\parallel\psibar_{i} - \psibar^{h}_{i}\parallel^{2}_{1,\Omega} +|\phibar - \phibar^h|^{2}_{1,\Omega} + \sum_{i}\parallel\psibar_{i} - \psibar^{h}_{i}\parallel_{0,\Omega} \parallel\phibar - \phibar^h\parallel_{1,\Omega}\right)
\end{equation}
\end{prop}
\begin{proof}
We use the following norms: $| \cdot |_{1,\Omega}$ represents the semi-norm in $H^{1}$ space, $\parallel\cdot \parallel_{1,\Omega}$  denotes the $H^{1}$ norm, $\parallel\cdot \parallel_{0,\Omega}$ and  $\parallel\cdot \parallel_{0,p,\Omega}$  denote the  standard $L^{2}$ and  $L^{p}$ norms respectively. All the constants to appear in the following estimates are positive and bounded. Firstly, we note that
\begin{equation}
\sum_{i}\frac{1}{2} \intomega |\del \delta \psi_{i}|^2 \dr  \leq C_1 \sum_{i} |\psibar_{i} - \psibar_{i}^h|^{2}_{1,\Omega}\,,
\end{equation}
\begin{equation}
\sum_{i} |\epsilonbar| \intomega (\delta \psi_{i})^2 \dr = \sum_{i} |\epsilonbar| \intomega (\psibar_{i}-\psibar_{i}^h)^2 \dr \leq C_2 \sum_{i} \parallel\psibar_{i} - \psibar_{i}^{h}\parallel ^{2}_{0,\Omega}\,.
\end{equation}
Using Cauchy-Schwartz and Sobolev inequalities, we arrive at the following estimate
\begin{align}
\sum_{i} \left|\intomega F'(\rho(\Psibar)) (\delta \psi_{i})^2 \dr \right| &\leq \sum_{i} \intomega  \left|F'(\rho(\Psibar) )(\psibar_{i} - \psibar_{i}^{h})^2\right| \dr \nonumber\\[0.1in]
&\leq C_3 \sum_{i}  \parallel  F'(\rho(\Psibar)) \parallel_{0,\Omega} \parallel(\psibar_{i} - \psibar_{i}^{h})^2\parallel_{0,\Omega}\nonumber\\[0.1in]
&= C_3 \sum_{i} \parallel   F'(\rho(\Psibar)) \parallel_{0,\Omega} \parallel\psibar_{i} - \psibar_{i}^{h}\parallel^{2}_{0,4,\Omega}\nonumber\\[0.1in]
&\leq \bar{C}_3 \sum_{i} \parallel\psibar_{i} - \psibar_{i}^h\parallel^{2}_{1,\Omega}\,.
\end{align}
Further, we note
\begin{equation}
\frac{1}{8\pi} \intomega |\del (\phibar - \phibar^h)|^2 \dr \leq C_4 |\phibar - \phibar^h|^{2}_{1,\Omega}\,.
\end{equation}
Using Cauchy-Schwartz and Sobolev inequalities we arrive at
\begin{align}
\sum_{i} \left|\intomega (\delta \psi_{i})^2\phibar\;\dr\right| \leq \sum_{i} \intomega \left|(\psibar_{i} - \psibar_{i}^h)^2\;\phibar\right| \dr &\leq \sum_{i}  \parallel \phibar \parallel_{0,\Omega} \parallel (\psibar_{i} - \psibar_{i}^h)^2\parallel_{0,\Omega}\nonumber\\
&\leq C_5 \sum_{i}  \parallel\psibar_{i} - \psibar_{i}^{h}\parallel^{2}_{0,4,\Omega}\nonumber\\
& \leq \bar{C}_5 \sum_{i}  \parallel\psibar_{i} - \psibar_{i}^h\parallel^{2}_{1,\Omega}\,.
\end{align}
Also, we note that
\begin{align}
\sum_{i}  \left|\intomega \psibar_{i}\; \delta \psi_{i} \;\delta \phi\; \dr\right| &\leq \sum_{i}  \intomega \left|\psibar_{i} (\psibar_{i} - \psibar_{i}^h)(\phibar - \phibar^h)\right|\dr \nonumber\\
&\leq \; \sum_{i} \parallel \psibar_{i} \parallel_{0,6,\Omega} \parallel\psibar_{i} - \psibar_{i}^h\parallel_{0,\Omega} \parallel\phibar - \phibar^h\parallel_{0,3,\Omega}\nonumber\\
&\leq \sum_{i}  C_6 \parallel\psibar_{i} - \psibar_{i}^h\parallel_{0,\Omega} \parallel\phibar - \phibar^h\parallel_{1,\Omega}\,\,\,,
\end{align}
where we made use of the generalized H\"older inequality in the first step and Sobolev inequality in the next.
Finally, we use Cauchy-Schwartz inequality to arrive at
\begin{align}
\left|\intomega F''(\rho(\Psibar)) \left(\sum_{i} \psibar_{i} \delta \psi_{i}\right)^{2}\dr\right| & \leq \intomega \left|F''(\rho(\Psibar))\right| \left(\sum_{i} \left|\psibar_{i}\right|^{2}\right) \left(\sum_{i} \left|\delta \psi_{i}\right|^{2}\right) \dr\\
&= \sum_{i} \intomega \left|F''(\rho(\Psibar)) \rho(\Psibar) (\delta \psi_{i})^2 \right| \dr\\
&\leq C_7 \sum_{i} \parallel\psibar_{i} - \psibar_{i}^{h}\parallel ^{2}_{0,\Omega}\,.
\end{align}
Using the bounds derived above, it follows that
\begin{equation}\label{error3}
|E_h - E| \leq C \left(\sum_{i}\parallel\psibar_{i} - \psibar^{h}_{i}\parallel^{2}_{1,\Omega} +|\phibar - \phibar^h|^{2}_{1,\Omega} + \sum_{i}\parallel\psibar_{i} - \psibar^{h}_{i}\parallel_{0,\Omega} \parallel\phibar - \phibar^h\parallel_{1,\Omega}\right)
\end{equation}
\end{proof}
We now bound the finite-element discretization error with interpolation errors, which in turn can be bounded with the finite-element mesh size $h$. This requires a careful analysis in the case of Kohn-Sham DFT and has been discussed in~\cite{zhou1}. Using the results from the proof of Theorem 4.3 in~\cite{zhou1}, we bound the estimates in equation~\eqref{error3} using the following inequalities (cf.~\cite{Ciarlet})
\begin{subequations}
\begin{gather}
|\psibar_{i} - \psibar_{i}^{h}|_{1,\Omega} \leq \bar{C}_0 |\psibar_{i} - \psi_{i}^{I}|_{1,\Omega}\leq \tilde{C}_0\sum_{e}h_e^{k}|\psibar_{i}|_{k+1,\Omega_e}\,\,,\\
\parallel\psibar_{i} - \psibar_{i}^{h}\parallel_{0,\Omega} \leq \bar{C}_1\parallel\psibar_{i} - \psi_{i}^{I}\parallel_{0,\Omega} \leq \tilde{C}_1\sum_{e}h_e^{k+1}|\psibar_{i}|_{k+1,\Omega_e}\,\,,\\
|\phibar - \phibar^h|_{1,\Omega} \leq \bar{C}_2 |\phibar - \phi^I|_{1,\Omega}\leq \tilde{C}_2\sum_{e}h_e^{k}|\phibar|_{k+1,\Omega_e}\,\,,\\
\parallel\phibar - \phibar^h\parallel_{0,\Omega} \leq \bar{C}_2 \parallel\phibar - \phi^I\parallel_{0,\Omega}\leq \tilde{C}_3\sum_{e}h_e^{k+1}|\phibar|_{k+1,\Omega_e}\,\,,
\end{gather}
\end{subequations}
where $k$ is the order of the polynomial interpolation, and $e$ denotes an element in the regular family of finite-elements~\cite{Ciarlet} with mesh-size $h_e$ covering a domain $\Omega_e$. Using the above estimates, the error estimate to $O(h^{2k+1})$ is given by
\begin{equation}
|E_h - E| \leq \mathcal{C}\sum_{e}h_e^{2k}\left[\sum_{i}|\psibar_{i}|_{k+1,\Omega_e}^{2} + |\phibar|_{k+1,\Omega_e}^{2}\right]\,.\label{errorfinal1}
\end{equation}
\subsection{Optimal coarse-graining rate}\label{sec:optimal_mesh}
Following the approach in~\cite{radio}, we seek to determine the optimal mesh-size distribution by minimizing the approximation error in energy for a fixed number of elements. Using the definition of the semi-norms, we rewrite equation~\eqref{errorfinal1} as
\begin{equation}
|E_h - E| \leq \mathcal{C} \sum_{e=1}^{N_{e}} \Bigl[ h_{e}^{2k}\int_{\Omega_{e}} \Bigl[\sum_{i} |D^{k+1}\psibar_{i}(\br)|^{2} + |D^{k+1}\phibar(\br)|^{2}\Bigr] \dr\Bigr]\,,
\end{equation}
where $N_e$ denotes the total number of elements in the finite-element triangulation, and $D^{k+1}$ denotes the $(k+1)^{th}$ derivative of any function. An element size distribution function $h(\br)$ is introduced so that the target element size is defined at all points $\br$ in $\Omega$, and we get
\begin{align}\label{errorN}
|E_h - E| &\leq \mathcal{C} \sum_{e=1}^{N_{e}} \int_{\Omega_{e}} \Bigl[ h_{e}^{2k}\Bigl[\sum_{i} |D^{k+1}\psibar_{i}(\br)|^{2} + |D^{k+1}\phibar(\br)|^{2}\Bigr] \dr\Bigr]\,\\
& \leq \mathcal{C'} \intomega h^{2k}(\br)\Bigl[\sum_{i}|D^{k+1}\psibar_{i}(\br)|^{2} + |D^{k+1}\phibar(\br)|^{2}\Bigr] \dr\,.
\end{align}
Further, the number of elements in the mesh is in the order of
\begin{equation}
N_{e}  \propto \intomega \frac{\dr}{h^{3}(\br)} \label{elem}\,.
\end{equation}
The optimal mesh-size distribution is then determined by the following variational problem which minimizes the approximation error in energy subject to a fixed number of elements:
\begin{equation}
\min_{h} \intomega \Bigl\{ h^{2k}(\br)\Bigl[\sum_{i}|D^{k+1}\psibar_{i}(\br)|^{2} + |D^{k+1}\phibar(\br)|^{2}\Bigr] \Bigr\}\dr \quad \mbox{subject to}:\intomega\frac{\dr}{h^{3}(\br)}=N_{e}\,.
\end{equation}
The Euler-Lagrange equation associated with the above problem is given by
\begin{equation}
2kh^{2k-1}(\br)\Bigl[\sum_{i}|D^{k+1}\psibar_{i}(\br)|^{2} + |D^{k+1}\phibar(\br)|^{2}\Bigr] - \frac{3\eta}{h^{4}(\br)} = 0\,,
\end{equation}
where $\eta$ is the Lagrange multiplier associated with the constraint. Thus, we obtain the following distribution
\begin{equation}\label{optimmesh}
h(\br) = A \Bigl(\sum_{i}|D^{k+1}\psibar_{i}(\br)|^{2} + |D^{k+1}\phibar(\br)|^{2}\Bigr)^{-1/(2k+3)}\,,
\end{equation}
where the constant $A$ is computed from the constraint that the total number of elements in the finite-element discretization is $N_e$.

The coarse-graining rate derived in equation~\eqref{optimmesh} has been employed to construct the finite-element meshes by using the \emph{a priori} knowledge of the asymptotic solutions of $\psibar_{i}(\br)$ and $\phibar(\br)$ for different kinds of problems we study in the subsequent sections.

\section{Numerical implementation}\label{implement}
We now turn to the numerical implementation of the discrete formulation of the Kohn-Sham eigenvalue problem described in Section~\ref{formulation}. We first discuss the higher-order finite-elements used in our study with specific focus on spectral finite-elements, which are important in developing an efficient numerical solution procedure.

\subsection{Higher-order finite-element discretizations}
Linear finite-element basis has been extensively employed for a wide variety of applications in engineering involving complex geometries and moderate levels of accuracy. On the other hand, much higher levels of accuracy (chemical accuracy) is desired in electronic structure computations of materials properties. To achieve the desired chemical accuracy, a linear finite-element basis is computationally inefficient since it requires a large number of basis functions per atom~\cite{Hermannson,bylaska}. Hence, we investigate if higher-order finite-element basis functions can possibly be used to efficiently achieve the desired chemical accuracy. To this end, we employ in our study $C^0$ basis functions comprising of linear tetrahedral element (TET4) and spectral hexahedral elements up to degree eight (HEX27, HEX125SPECT, HEX343SPECT, HEX729SPECT). The numbers following the words `TET' and `HEX' denote the number of nodes in the element, and the suffix `SPECT' denotes that the element is a spectral finite-element. We note that spectral finite-elements~\cite{patera1984spectral,boyd2001chebyshev} have been employed in a previous work in electronic structure calculations~\cite{batcho}, but the computational efficiency afforded by these elements has not been thoroughly studied. We first briefly discuss spectral finite-elements (also referred to as spectral-elements) employed in the present work and the role they play in improving the computational efficiency of the Kohn-Sham DFT eigenvalue problem.

The spectral-element basis functions employed in the present work are constructed as Lagrange polynomials interpolated through an optimal distribution of nodes corresponding to the roots of derivatives of Legendre polynomials, unlike conventional finite-elements which use equispaced nodes in an element. Such a distribution does not have nodes on the boundaries of an element, and hence it is common to append nodes on the element boundaries which guarantees $C^{0}$ basis functions. These set of nodes are usually referred to as Gauss-Lobatto-Legendre points. Furthermore, we note that conventional finite-elements result in a poorly conditioned discretized problem for a high order of interpolation, whereas spectral-elements provide better conditioning~\cite{boyd2001chebyshev}. The improved conditioning of the spectral-element basis was observed to provide a 2-3 fold computational advantage over conventional finite-elements in a recent benchmark study~\cite{Motamarri2012} conducted to assess the computational efficiency of higher-order elements in the solution of the orbital-free DFT problem.

A significant advantage of the aforementioned spectral-elements is realized when we conjoin their use with specialized Gaussian quadrature rules that have quadrature points coincident with the nodes of the spectral-element, which in the present case corresponds to the Gauss-Lobatto-Legendre (GLL) quadrature rule~\cite{GLL}. Importantly, the use of such a quadrature rule will result in a diagonal overlap matrix (mass matrix) $\bM$. To elaborate, consider the elemental mass matrix $\bM^{e}$ given by
\begin{align}
\int_{\Omega_{e}} N_{i}(\br) N_{j}(\br) \dr &= \int_{-1}^{1}\int_{-1}^{1}\int_{-1}^{1} N_{i}(\xi,\eta,\zeta)N_{j}(\xi,\eta,\zeta) \;\; det (J_{e}) \;d\xi \; d\eta\; d\zeta\\
& =\sum_{p,q,r = 0}^{n_{q}} w_{p,q,r} N_{i}(\xi_{p}, \eta_{q}, \zeta_{r}) N_{j}(\xi_{p}, \eta_{q}, \zeta_{r})\;\; det(J_{e})
\end{align}
where $(\xi,\eta,\zeta)$ represents the barycentric coordinates, $J_{e}$ represents the elemental jacobian matrix of an element $\Omega_e$, and $n_q$ denotes the number of quadrature points in each dimension in a tensor product quadrature rule. Since the quadrature points are coincident with nodal points, the above expression is non-zero only if $i = j$, thus resulting in a diagonal elemental mass matrix and subsequently a diagonal global mass matrix. A diagonal mass matrix makes the transformation of the generalized Kohn-Sham eigenvalue problem~\eqref{ghep} to a symmetric standard eigenvalue problem~\eqref{hep} trivial. As discussed and demonstrated subsequently, the transformation to a standard eigenvalue problem allows us to use efficient solution procedures to compute the eigenspace in the self-consistent field iteration. We note that, while the use of the GLL quadrature rule is important in efficiently transforming the generalized eigenvalue problem to a standard eigenvalue problem, this quadrature rule is less accurate in comparison to Gauss quadrature rules. An $n$ point Gauss-Lobatto rule can integrate polynomials exactly up to degree $2n-3$, while an $n$ point Gauss quadrature rule can integrate polynomials exactly up to degree $2n -1$. Thus, in the present work, we use the GLL quadrature rule only in the evaluation of the overlap matrix, while using the more accurate Gauss quadrature rule to evaluate the discrete Hamiltonian matrix $\bH$. The accuracy and sufficiency of this reduced-order GLL quadrature for the evaluation of overlap matrix is demonstrated in Appendix C.

\subsection{Self-consistent field iteration}
As noted in Section~\ref{formulation}, the Kohn-Sham eigenvalue problem represents a nonlinear eigenvalue problem and must be solved self-consistently to compute the ground-state electron density and energy. We use computationally efficient schemes to evaluate the occupied eigenspace of the Kohn-Sham Hamiltonian (discussed below) in conjunction with finite temperature Fermi-Dirac distribution and charge density mixing to develop an efficient and robust solution scheme for the self-consistent field iteration of Kohn-Sham problem.
\begin{algorithm}
\caption{Self Consistent Field Iteration}
\label{scfEigenAlgo}
\begin{algorithmic}
\STATE 1. Provide initial guess for electron density $\rho^{h}_{0}(\br)$ on the finite-element mesh. This will be the input electron density for the first self-consistent iteration ($\rho^{h}_{\text{in}}(\br) = \rho^{h}_{0}(\br)$).  \\[0.1in]
\STATE 2. Compute the total electrostatic potential $\phi^{h}(\br,\bR) = V_H(\rho^{h}_{\text{in}}(\br)) + V_{ext}(b(\br,\bR))$ by solving the discrete Poisson equation. \\[0.1in]
\STATE 3. Compute the effective potential, $V_{\text{eff}}(\rho^{h}_{\text{in}},\bR) = V_{xc}(\rho^{h}_{\text{in}}) + \phi^{h}(\br,\bR)$\,\,.\\
\STATE 4. Solve for the occupied subspace spanned by the eigenfunctions $\psi^{h}_{i}(\br)$,  $i = 1,2\cdots\tilde{N}$, corresponding to $\tilde{N}$ ($\tilde{N}>N/2$) smallest eigenvalues of the Kohn-Sham eigenvalue problem~\eqref{ghep}.\\[0.1in]
\STATE 5. Calculate the fractional occupancy factors ($f_{i}$) using the Fermi-Dirac distribution (Section~\eqref{Fermi-Dirac}) \\[0.1in]
\STATE 6. Compute the new output charge densities $\rho^{h}_{\text{out}}$ from the eigenfunctions:
\begin{equation}\label{rhocomp}
\rho^{h}_{\text{out}}(\br) = 2\sum_{i} f(\epsilon_{i},\epsilon_{F})|\psi^{h}_{i}(\br)|^2,
\end{equation}
\STATE 7. If $||\rho^{h}_{\text{out}}(\br) - \rho^{h}_{\text{in}}(\br)|| \leq \;\text{tolerance}$, \emph{stop}; Else, compute new $\rho^{h}_{\text{in}}$ using a mixing scheme (Section~\ref{MixingScheme}) and go to step 2.
\end{algorithmic}
\end{algorithm}

Algorithm~\ref{scfEigenAlgo} depicts the typical steps involved in the self-consistent field (SCF) iteration. An initial guess of the electron density field is used to start the computation. A reasonable choice of such an initial guess is the superposition of atomic charge densities, and is used in the present study unless otherwise mentioned. The input charge density ($\rho^{h}_{\text{in}}(\br)$) to a self-consistent iteration is used to compute the total electrostatic potential $\phi(\br,\bR)$  by solving the following discrete Poisson equation using a preconditioned conjugate gradient method provided by the PETSc \cite{petsc} package using a Jacobi preconditioner.
\begin{equation}\label{forcephi}
\sum_{k=1}^{n_h} \Bigl[\frac{1}{4\pi} \intomega \del N_j(\br)\, . \del N_k(\br) \dr\Bigr]\phi^{k} = \intomega \left(\rho^{h}_{\text{in}}(\br) + b(\br,\bR)\right) N_j(\br) \dr\,\,.
\end{equation}
Subsequently, the effective potential $V_{\text{eff}}$ is evaluated to set up the discrete Kohn-Sham eigenvalue problem~\eqref{ghep}. We now discuss the different strategies we have investigated to compute the occupied eigenspace of the Kohn-Sham Hamiltonian $\bH$, and their relative merits.
\subsubsection{Solver strategies for finding the occupied eigenspace}
We examined two different solution strategies to compute the occupied subspace: (i) explicit computation of eigenvectors at every self-consistent field iteration; (ii) A Chebyshev filtering approach.

\paragraph{Explicit computation of eigenvectors:}
We first discuss the methods examined in the present work that involve an explicit computation of eigenvectors at a given self-consistent iteration. We recall that the discrete Kohn-Sham eigenvalue problem is a generalized Hermitian eigenvalue problem (GHEP)~\eqref{ghep}. As mentioned previously, by using the GLL quadrature rules for the evaluation of the overlap matrix $\bM$, which results in a diagonal overlap matrix, the generalized eigenvalue problem can be trivially transformed into a standard Hermitian eigenvalue problem (SHEP). We have explored both approaches in the present work, i.e. (i) solving the generalized eigenvalue problem employing conventional Gauss quadrature rules; (ii) solving the transformed standard eigenvalue problem by using GLL quadrature rules in the computation of overlap matrix.

We have employed the Jacobi-Davidson~(JD) method~\cite{jd} to solve the GHEP. The JD method falls into the category of iterative orthogonal projection methods where the matrix is orthogonally projected into a lower dimensional subspace and one seeks an approximate eigenpair of the original problem in the subspace. The basic idea in JD method is to arrive at better approximations to eigenpairs by a systematic expansion of the subspace realized by solving a ``Jacobi-Davidson correction equation" that involves the solution of a linear system. In the present work, a Jacobi preconditioner has been employed in the solution of the correction equation. The correction equation is solved only approximately, and this approximate solution is used for the expansion of the subspace. Though the JD method has significant advantages in computing the interior eigenvalues and closely spaced eigenvalues, we found the JD method to be computationally expensive for systems involving the computation of eigenvectors greater than 50, due to the increase in the number of times the correction equation is solved.

On the other hand, we employed the Krylov-Schur (KS) method~\cite{krylovschur2001} for solving the SHEP. In practice, one could also use the JD method to solve the SHEP, but, as previously mentioned, the JD method is expensive to solve systems involving few hundreds of electrons and beyond. The KS method can be viewed as an improvement over traditional Krylov subspace methods such as Arnoldi and Lanczos methods~\cite{Trefethen,Bai}. The KS method is based on Krylov-Schur decomposition where the Hessenberg matrix has the Schur form. The key idea of the KS method is to iteratively construct the Krylov-subspace using Arnoldi iteration and subsequently filter the unwanted spectrum from the Krylov-Schur decomposition. This results in a robust restarting scheme with faster convergence in most cases.

We now demonstrate the computational efficiency realized by solving the discrete Kohn-Sham eigenvalue problem as a transformed SHEP in comparison to GHEP. To this end, we consider an all-electron simulation of a graphene sheet containing 16 atoms with 96 electrons ($N = 96$) and a local pseudopotential simulation (cf. section~\ref{sec:pseudopotential_calculations} for details on the pseudopotentials employed) of $3\times 3\times 3$ face-centered-cubic aluminum nano-cluster containing 172 atoms with 516 electrons ($N = 516$) as benchmark systems. The relative error in the ground-state energy for the finite-element mesh used in the case of graphene is around $1.2 \times 10^{-5}$ (~0.0004 $Ha/atom$) while it is around $3.6 \times 10^{-6}$ (~0.0002 $eV/atom$) in the case of aluminium cluster. The reference ground-state energy is obtained using the commercial code GAUSSIAN in the case of the all-electron simulation of the graphene system, while it is obtained using the convergence study presented in Section~\ref{results} for the aluminium cluster. Table~\ref{tab:GHEP_SHEP} shows the computational time taken for the first SCF iteration in each of the above cases. All the times reported in the present work represent the total CPU times.
 \begin{table}[h]
  \begin{center}
  \caption{\small{Comparison of Generalized vs Standard eigenvalue problems.}}
  \label{tab:GHEP_SHEP}
 \begin{tabular}{|c|c|c|c|c|c|}
   \hline
  Element Type & DOFs & Problem Type & $N$ & Time (GHEP) & Time (SHEP) \\ \hline\hline
 HEX125SPECT &  1,368,801 & graphene & 96 & 1786 CPU-hrs & 150 CPU-hrs\\ \hline
 HEX343SPECT &  2,808,385 & Al $3\times 3\times 3$ cluster &516  & 2084 CPU-hrs & 80 CPU-hrs\\
 \hline
\end{tabular}
\end{center}
\end{table}
The Jacobi-Davidson method for GHEP and Krylov-Schur method for SHEP provided by the SLEPc package~\cite{slepc} have been employed in the present study. We remark that, in employing the Jacobi-Davidson method, eigenvectors from the previous SCF iterations have been supplied as input approximations for the subsequent SCF iteration. The Krylov-Schur method, on the other hand, allows for one only vector to be supplied as the input approximation to a given SCF iteration. Hence, the eigenvector corresponding to smallest eigenvalue from the previous SCF iteration has been supplied as the input approximation for the subsequent SCF iteration. It is interesting to note that a $10$-fold speedup is realized by transforming the Kohn-Sham eigenvalue problem to a SHEP in the case of graphene, while a $25$-fold speedup was obtained in the case of aluminium cluster. We note that a similar observation was recently reported in~\cite{Zhou2012_JCP} where the GHEP was transformed to SHEP via the mass-lumping approximation. Further, other simulations conducted as part of the present study suggest that this speedup increases with increasing system size.

\paragraph{Chebyshev filtering:} We now examine the alternate approach of Chebyshev filtering proposed in~\cite{cheby}, which is designed to iteratively compute the occupied eigenspace at every SCF iteration. We note that the Chebyshev filtering approach is only valid for standard eigenvalue problems. To this end, we use the aforementioned approach to convert the GHEP to a SHEP by employing the GLL quadrature rules in computing the overlap matrix, and remark that the use of spectral elements in conjunction with the GLL quadrature is crucial in using the Chebyshev filtering technique to solve the Kohn-Sham eigenvalue problem in a finite-element basis. The Chebyshev filtering approach is based on a subspace iteration technique, where an initial subspace is acted upon by a Chebyshev filter constructed from the Kohn-Sham Hamiltonian that transforms the subspace to the occupied eigenspace.

In the present work, at any given SCF iteration, we begin with the initial subspace $V$ formed from the eigenvectors of the previous SCF iteration. We note that, as is the case with all subspace iteration techniques, we choose the dimension of the subspace $V$, $\tilde{N}$, to be larger than the number of filled ground-state orbitals. Typically, we choose $\tilde{N}\sim\frac{N}{2}+20$. This is also necessary to employ the finite temperature Fermi-Dirac smearing, discussed in Section~\ref{Fermi-Dirac}, to stabilize the SCF iterations in materials systems that have very small band-gaps or have degenerate states at the Fermi energy. As proposed in~\cite{cheby}, the Chebyshev filter is constructed from a shifted and scaled Hamiltonian, $\overline{\bH}=c_1\btH+c_2$, where $\btH$ is the transformed Hamiltonian in the SHEP (cf. equation~\eqref{hep}). The constants $c_1$ and $c_2$ which correspond to the scaling and shifting are determined such that the unwanted eigen-spectrum is mapped into $[-1,1]$ and the wanted spectrum into $(-\infty,-1)$. In order to compute these constants, we need estimates of the upper bounds of the wanted and unwanted spectrums. The upper bound of the unwanted spectrum, which corresponds to the largest eigenvalue of $\btH$, can be obtained inexpensively by using a small number of iterations of the Lanczos algorithm. The upper bound of the wanted spectrum is chosen as largest Rayleigh quotient of $\btH$ in the space $V$ from the previous SCF iteration. Subsequently, the degree-$m$ Chebyshev filter, $p_m(\overline{\bH})$, which magnifies the spectrum of $\overline{\bH}$ in $(-\infty,-1)$---the wanted eigen-spectrum of $\btH$---transforms the initial subspace $V$ to the occupied eigenspace of $\btH$. The degree of the Chebyshev filter is chosen such that the obtained space is a close approximation of the occupied space. We note that the action of the Chebyshev filter on $V$ can be performed recursively, similar to the recursive construction of the Chebyshev polynomials~\cite{ChebyPolynomial}. After obtaining the occupied eigenspace, we orthogonalize the basis functions, and subsequently project $\btH$ into the eigenspace to compute the eigenvalues that are used in the Fermi-Dirac smearing discussed in the next subsection.

We remark that the degree of the polynomial required for the Chebyshev filter depends on the separation between eigenvalues of $\overline{\bH}$ in $(-\infty,-1)$, which in turn depends on: (i) the ratio between the wanted and unwanted eigenspectrums of $\btH$; (ii) the separation between the eigenvalues in the wanted spectrum of $\btH$. The size of the unwanted spectrum is primarily governed by the largest eigenvalue of $\btH$, which, in turn, is related to the finite-element discretization---increases with decrease in the element-size of the finite-element mesh. In general, all-electron calculations require locally refined meshes near the nuclei as they involve Coulomb-singular potential fields and highly oscillatory core wavefunctions. Hence, a very high degree of Chebyshev polynomial---of the order of $10^{2}-10^{3}$ for the problems studied in this work---needs to be employed to effectively filter the unwanted spectrum. On the other hand, simulations performed on systems with smooth pseudopotential required Chebyshev polynomial degrees between 10 to 50 for the range of problems studied in the present work. Further, qualitatively speaking, a larger degree Chebyshev filter is required for larger systems as the separation between eigenvalues in the wanted spectrum of $\btH$ reduces with increasing number of electrons.

We now compare the computational times (cf. table~\ref{tab:SHEP_ChFSI}) taken for a single SCF iteration solved using an eigenvalue solver based on Krylov-Schur method and the Chebyshev filter using the aforementioned benchmark problems comprising of a 16-atom graphene sheet and 172-atom aluminium cluster. We use a Chebyshev polynomial of degree $800$ for the graphene all-electron calculation and a polynomial degree of $12$ for aluminum cluster pseudopotential calculation respectively. As is evident from the results, we clearly see a factor of $12$ speedup that is obtained in the case of graphene, and a factor of around $6$ speedup that is obtained in the case of aluminium cluster. The speedup obtained was even greater for larger materials systems.
\begin{table}[h]
 \begin{center}
  \caption{\small{Comparison of Standard eigenvalue problem vs Chebyshev filtered subspace iteration (ChFSI).}}
  \label{tab:SHEP_ChFSI}
 \begin{tabular}{|c|c|c|c|c|c|}
   \hline
  Element Type & DOFs & Problem Type & $N$ & Time (SHEP) & Time (ChFSI) \\ \hline\hline
 HEX125SPECT &  1,368,801 & graphene & 96 & 150 CPU-hrs & 12.5 CPU-hrs\\ \hline
 HEX343SPECT &  2,808,385 & Al $3\times 3\times 3$ cluster &512  &  80 CPU-hrs & 13 CPU-hrs\\
 \hline
\end{tabular}
\end{center}
\end{table}

The use of spectral finite-elements in conjunction with Chebyshev filtered subspace iteration presents an efficient and robust approach to solve the Kohn-Sham problem using the finite-element basis. Thus, for subsequent simulations reported in the present work that use hexahedral elements, we employ the Krylov-Schur method for the first SCF iteration to generate a good initial subspace and use the Chebyshev filtering approach for all subsequent iterations to compute the occupied eigenspace. However, for simulations that use tetrahedral elements, we solve the GHEP using Jacobi-Davidson method as a transformation to SHEP is non-trivial and involves the inversion of overlap matrix using iterative techniques.

\subsubsection{Finite temperature smearing: Fermi-Dirac distribution}\label{Fermi-Dirac}
For materials systems with very small band gaps or those with degenerate energy levels at the Fermi energy, the SCF iteration may exhibit \emph{charge sloshing}---a phenomenon where large deviations in spatial charge distribution are observed between SCF iterations with different degenerate (or close to degenerate) levels being occupied in different SCF iterations. In such a scenario, the SCF exhibits convergence in the ground-state energy, but not in the spatial electron density. It is common in electronic structure calculations to introduce an orbital occupancy factor~\cite{VASP} based on the energy levels and a smearing function to remove charge sloshing in SCF iterations. A common choice for the smearing function is the finite temperature Fermi-Dirac distribution, and the orbital occupancy factor $f_i$ corresponding to an energy level $\epsilon_i$ is given by
\begin{equation}\label{fermi}
f_i\equiv f(\epsilon_i,\epsilon_F) = \frac{1}{1 + \exp(\frac{\epsilon_i - \epsilon_{F}}{\sigma})}\,\,,
\end{equation}
where the smearing factor $\sigma = k_{B} T$ with $k_{B}$ denoting the Boltzmann constant and $T$ denoting the temperature in Kelvin. In the above expression, $\epsilon_{F}$ denotes the Fermi energy, which is computed from the constraint on the total number of electrons given by $\sum_{i}2f_i=N$. We note that the convergence of ground-state energy is quadratic in the smearing parameter $\sigma$~\cite{VASP}.

\subsubsection{Mixing scheme:}\label{MixingScheme}
The convergence of the SCF iteration is crucially dependent on the mixing scheme, and many past works in the development of electronic structure methods have focussed on this aspect~\cite{Pulay,and,Broyden,EDIIS,eyert}. In the present work, we employ an $n$-stage Anderson mixing scheme~\cite{and}, which is briefly described below for the sake of completeness. Let $\rhoin {h} {n}(\br)$ and $\rhoout {h} {n}(\br)$ represent the input and output electron densities of the $n^{\text{th}}$ self-consistent iteration. The input to the $(n+1)^{\text{th}}$ self-consistent iteration, $\rhoin {h} {n+1}(\br)$, is computed as follows
\begin{equation}\label{rhoin}
\rhoin {h} {n+1} = \gamma_{\text{mix}}\; {\bar{\rho}}_{\text{out}}^{h} + (1-\gamma_{\text{mix}}) \;{\bar{\rho}}_{\text{in}}^{h}
\end{equation}
where
\begin{equation}\label{rhooptim}
{\bar{\rho}}^{h}_{\text{in(out)}}  =   c_{n}\; \rhoinout {h} {n} + \sum_{k=1}^{n-1} c_{k} \;\rhoinout {h} {n-k}
\end{equation}
and the sum of all the constants $c_{i}$ is equal to one, i.e.,
\begin{equation}
c_1 + c_2 + c_3 + \cdots + c_n = 1\,\,.
\end{equation}
Using the above constraint, equation~\eqref{rhooptim} can be written as
\begin{equation}
{\bar{\rho}}_{\text{in(out)}}^{h} = \rhoinout {h} {n} + \sum_{k=1}^{n-1}c_k \left(\rhoinout {h} {n-k} - \rhoinout {h} {n} \right) \,\,.
\end{equation}
Denoting  $F = \rho_{\text{out}}^{h} - \rho_{\text{in}}^{h}$, the above equation can be written as
\begin{equation}
\bar{F} = F^{\;(n) }+ \sum_{k=1}^{n-1}c_k \left(F^{\;(n-k)} - F^{\;(n)}\right)\,\,.
\end{equation}
The unknown constants $c_1$ to $c_{n-1}$ are determined by minimizing $R = ||\bar{F}||^{2}_{2} = ||{\bar{\rho}}_{\text{in}}^{h} - {\bar{\rho}}_{\;\text{out}}^{h}||^{2}_{2} $, which amounts to solving the following system of $(n-1)$ linear equations given by:
\begin{equation}\label{mix}
\sum_{k=1}^{n-1} \left(F^{(n)} - F^{(n-m)},F^{(n)} - F^{(n-k)}\right) c_{k} = \left(F^{(n)} - F^{(n-m)}, F^{(n)} \right) \quad m = 1 \cdots n-1
\end{equation}
where the notation $(F,G) $ stands for the $L_{2}$ inner product between functions $F(\br)$ and $G(\br)$ and is given by
\begin{equation}
(F,G) = \int F(\br) G(\br) \dr\,\,.
\end{equation}
The value of the parameter $\gamma_{mix}$ in equation~\eqref{rhoin} is chosen to be $0.5$ in the present work. All the integrals involved in the linear system~\eqref{mix} are evaluated using Gauss quadrature rules, and the values of $\rhoinout {h}{n} (\br)$ are stored as quadrature point values after every $n^{th}$ self consistent iteration. In all the simulations conducted in the present work, the Anderson mixing scheme is used with full history.

\section{Numerical results}\label{results}
\subsection{Rates of convergence}\label{RatesOfConv}
We begin with the examination of convergence rates of the finite-element approximation using a sequence of meshes with decreasing mesh sizes for various polynomial orders of interpolation. The benchmark problems used in this study, include: (i) all-electron calculations performed on boron atom and methane molecule, which represent non-periodic problems with a Coulomb-singular nuclear potential; (ii) local pseudopotential calculations performed on a barium cluster that represents a non-periodic problem with a smooth external potential, and a bulk calculation of face-centered-cubic (FCC) calcium crystal. In the case of all-electron calculations, the nuclear charges are treated as point charges on the nodes of the finite-element triangulation and the discretization provides a regularization for the electrostatic potential. We note that the self-energy of the nuclei in this case is mesh-dependent and diverges upon mesh refinement. Thus, the self energy is also computed on the same mesh that is used to compute the total electrostatic potential, which ensures that the divergent components of the variational problem on the right hand side of equation~\eqref{elReformulation} and the self energy exactly cancel owing to the linearity of the Poisson equation (cf. Appendix A  for a detailed discussion).

We conduct the convergence study by adopting the following procedure. Using the \emph{a priori} knowledge of the asymptotic solutions of the atomic wavefunctions~\cite{asymp}, we determine the coarsening rate from equation~\eqref{optimmesh} which is used to construct the coarsest mesh. Though the computed coarsening rates use the far-field asymptotic solutions instead of the exact ground-state wavefunctions that are \emph{a priori} unknown, the obtained meshes nevertheless provide a systematic way for the discretization of vacuum in non-periodic calculations as opposed to using an arbitrary coarse-graining rate or uniform discretization. In the case of periodic pseudopotential calculations, a finite-element discretization with a uniform mesh-size is used. A uniform subdivision of the initial coarse-mesh is carried out to generate a sequence of refined meshes, which represents a systematic refinement of the finite-element approximation space. The ground-state energies from the discrete formulation, $E_h$, obtained from the sequence of meshes constructed using the HEX125SPECT element and containing $N_e$ elements are used to obtain a least squares fit of the form
\begin{equation}\label{convfit}
|E_h - E_0| =  \mathcal{C} (1/N_e)^{2k/3}\,,
\end{equation}
to determine the constants $E_0$, $\mathcal{C}$ and $k$. The obtained value of $E_0$, which represents the extrapolated continuum ground-state energy computed using the HEX125SPECT element, is used as the reference energy to compute the relative error $\frac{|E_h - E_0|}{|E_0|}$ in the convergence study of various orders of finite-elements reported in subsequent subsections.

\subsubsection{All-electron calculations}
We first begin with all-electron calculations by studying two examples: (i) boron atom (ii) methane molecule.
\paragraph{Boron atom:}   This is one of the simplest systems displaying the full complexity of an all-electron calculation. For the present case, we use a Chebyshev filter of order $500$ to compute the occupied eigenspace. As discussed in Section~\ref{Fermi-Dirac}, we use a finite-temperature smearing to avoid instability in the SCF iteration due to charge sloshing from the degenerate states at the Fermi energy. A smearing factor $\sigma = 0.0003168~Ha$ (T=100K) is used in the present study. The simulation domain used is a spherical domain of radius $20$ $a.u.$ with Dirichlet boundary conditions employed on electronic wavefunctions and total electrostatic potential. We first determine the mesh coarse-graining rate by noting that the asymptotic decay of atomic wavefunctions is exponential, and an upper bound to this decay under the Hartree-Fock approximation is given by~\cite{asymp}
\begin{equation}
\psi(r) \sim \exp\Bigl[-\sqrt{2\;\tilde{\epsilon}} \;r\Bigr]\;\;\;\;\text{for}\;\;\; r\rightarrow\infty\,,
\end{equation}
where $-\tilde{\epsilon}$ denotes the energy of the highest occupied atomic/molecular orbital. While the above estimate has been derived for the Hartree-Fock formulation, it nevertheless provides a good approximation to the asymptotic decay of wavefunctions computed using the Kohn-Sham formulation. We use the aforementioned estimate, though not optimal, for all the wavefunctions in the atomic system, and adopt this approach for all systems considered subsequently. Hence, in equation ~\eqref{optimmesh}, we consider $\psibar_{i}$ to be
\begin{equation}
\psibar(r) = \sqrt{\frac{\xi^3}{\pi}} \exp\Bigl[-\xi \;r\Bigr] \;\;\;\;\text{where}\;\;\;\;\xi = \sqrt{2\,\tilde{\epsilon}}\,.
\end{equation}
The electrostatic potential governed by the Poisson equation with a total charge density being equal to the sum of $5\bar{\psi}^{2}(r)$ and $-5\delta(r)$ is given by
\begin{equation}
\bar{\phi}(r) = -5 \exp{(-2\xi \,r)}\left(\xi + \frac{1}{r}\right)\,.
\end{equation}
Using the above equations, the  mesh coarse-graining rate from equation~\eqref{optimmesh}  is given by
{\footnotesize
\begin{equation}\label{hel_optim}
h(r) = A\left[\frac{5}{\pi}\xi^{2k+5} \exp{(-2\,\xi\,r)} + 25 \exp{(-4\,\xi\,r)}\left[\xi^{k+2} 2^{k+1} + \sum_{n=0}^{k+1} \dbinom{k+1}{n} \frac{2^{n} \xi^{n} (k+1-n)!}{r^{k-n+2}}\right]^{2}\right]^{ -1/(2k+3)}\,.
\end{equation}
}
Since $\tilde{\epsilon}$ in the above equation is unknown \emph{a priori}, the value of $\tilde{\epsilon}^{h}$ determined on a coarse mesh is used in the above equation to obtain $h(r)$ away from the atom. The finite-elements around the boron atom has been subdivided to get local refinement near the boron atom. We now perform the numerical convergence study with tetrahedral and hexahedral spectral elements up to eighth order using this coarse-graining rate, and the results are shown in figure~\ref{fig:BoronConvgRate}. The value of $E_0$ computed from equation~\eqref{convfit} is $-24.3431910234~Ha$, which is used to compute the relative errors in the energies. The ground-state energy computed by performing an all-electron calculation using APE (Atomic Pseudopotential Engine) software~\cite{ape} is found to be $-24.34319112~Ha$.

We observe that all the elements studied show close to optimal rates of convergence, $O(h^{2k})$, where $k$ is the degree of the polynomial. An interesting point to note is that, although the governing equations are non-linear in nature and the nuclear potential approaches a Coulomb-singular solution upon mesh refinement, optimal rates of convergence are obtained. Recent mathematical analysis~\cite{zhou1} shows that the finite-element approximation for the Kohn-Sham DFT problem does provide optimal rates of convergence for pseudopotential calculations. To the best of our knowledge,  mathematical analysis of higher-order finite-element approximations of the Kohn-Sham DFT problem with Coulomb-singular nuclear potentials is still an open problem.

We note that, in the case of linear finite-elements, a large number of elements are required to even achieve modest relative errors. In fact, close to five million linear TET4 elements are required for a single boron atom to obtain a relative error of $10^{-2}$, while relative errors up to $10^{-4}$ are achieved with just few hundreds of HEX125SPECT and HEX343SPECT elements, and even higher accuracies are achieved with a few thousands of these elements.

\paragraph{Methane molecule:}
The next example we study is methane with a C-H bond length of 2.07846 a.u. and a C-H-C tetrahedral angle of $109.4712^{\circ}$. For the present case, we use a Chebyshev filter of order $500$ to compute the occupied eigenspace, and a smearing factor $\sigma = 0.0003168~Ha$ (T=100K) for the Fermi-Dirac smearing. The simulation domain used is a cubical domain of side $50$ $a.u.$ with Dirichlet boundary conditions employed on electronic wavefunctions and total electrostatic potential. As in the case of boron atom, the finite-element mesh for this molecule is constructed to be locally refined around the atomic sites, while coarse-graining away. A uniform mesh is first constructed near the methane molecule and the finite-elements around each nuclei are then subdivided to obtain local refinements around each nuclei. The mesh coarsening rate in the outer region is determined numerically by employing the asymptotic solution of the far-field electronic fields, estimated as a superposition of single atom far-field asymptotic fields, in equation~\eqref{optimmesh}. To this end, asymptotic behavior of the atomic wavefunctions in carbon atom ($\psibar^{C}(r)$) is chosen to be
\begin{equation}\label{carbon}
\psibar^{C}(r) = \sqrt{\frac{\xi^3}{\pi}} \exp\Bigl[-\xi \;r\Bigr] \;\;\;\;\text{where}\;\;\;\;\xi = \sqrt{2\,\tilde{\epsilon}}\,,
\end{equation}
where $\tilde{\epsilon}$ (negative of the eigenvalue of the highest occupied eigenstate) is determined from a coarse mesh calculation of single carbon atom. The corresponding electrostatic potential is governed by the Poisson equation, with total charge density being equal to the sum of $6|\bar{\psi}^{C}(r)|^{2}$ and $-6\delta(r)$, and is given by
\begin{equation}
\bar{\phi}(r) = -6 \exp{(-2\xi \,r)}\left(\xi + \frac{1}{r}\right)\,.
\end{equation}
In the case of hydrogen atom, the analytical solution is given by
\begin{equation}\label{hydrogen}
\psibar^{H}(r) = \sqrt{\frac{1}{\pi}} \exp\Bigl[- \;r\Bigr] \,,
\end{equation}
and the corresponding electrostatic potential is given by
\begin{equation}
\bar{\phi}(r) = - \exp{(-2 \,r)}\left(1 + \frac{1}{r}\right)\,.
\end{equation}
We now perform the numerical convergence study with both tetrahedral and hexahedral elements with the meshes constructed as explained before. Figure~\ref{fig:methaneConvgRate} shows the convergence results for the various elements, and figure~\ref{fig:methanecontour} shows the isocontours of electron density for methane molecule. The value of $E_0$ computed from equation~\eqref{convfit}, the reference ground-state energy per atom of the methane molecule which is used to compute the relative errors in the energies, is found to be $-8.023988150~Ha$. The ground-state energy per atom computed using the GAUSSIAN package~\cite{GAUSSIAN} with polarization consistent $4$ DFT basis set~(pc-4) is found to be $-8.0239855633~Ha$. As in the case of boron atom, we obtain close to optimal convergence rates, and significantly higher relative accuracies in ground-state energies are observed by using higher-order elements.
\begin{figure}[h]
\hfill
\begin{minipage}[t]{.48\textwidth}
\centering
\includegraphics[width=\textwidth]{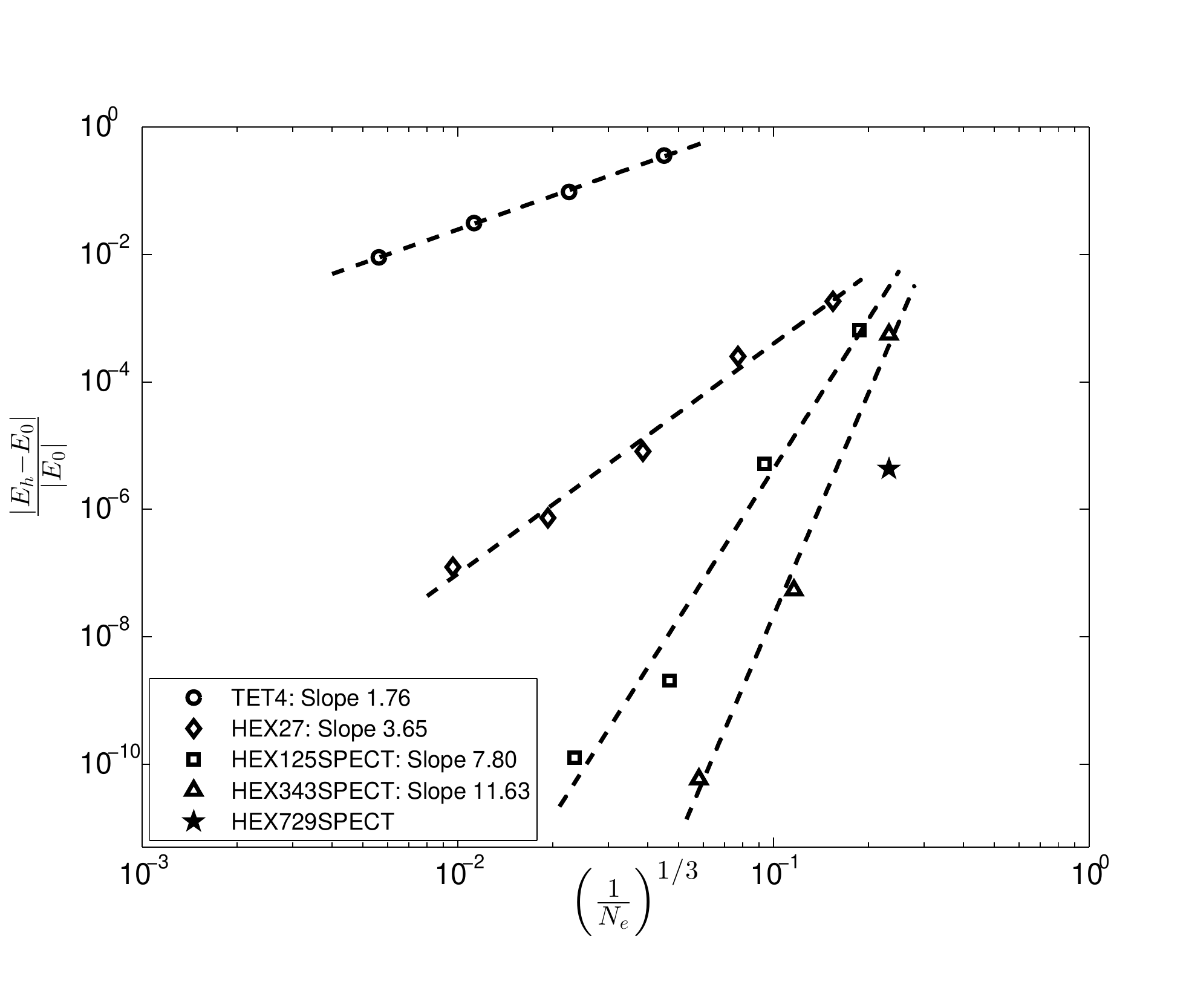}
\caption{\small{Convergence rates for the finite-element approximation of boron atom.}}
\label{fig:BoronConvgRate}
\end{minipage}
\hfill
\begin{minipage}[t]{.48\textwidth}
\centering
\includegraphics[width=\textwidth]{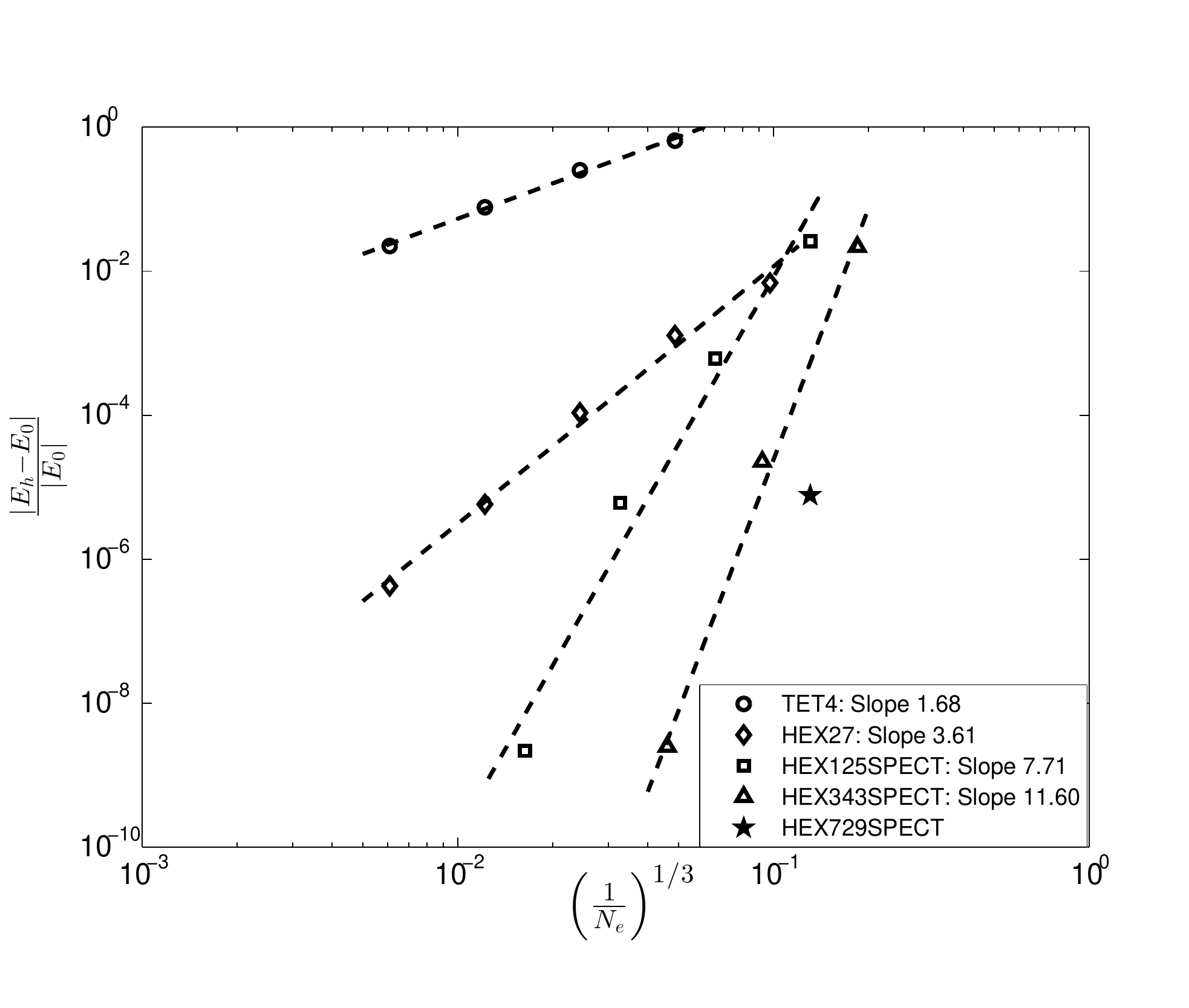}
\caption{\small{Convergence rates for the finite-element approximation of methane molecule.}}
\label{fig:methaneConvgRate}
\end{minipage}
\end{figure}

\begin{figure}[htbp]
\hfill
\begin{minipage}[t]{0.45 \textwidth}
\centering
\includegraphics[width = 0.7\textwidth]{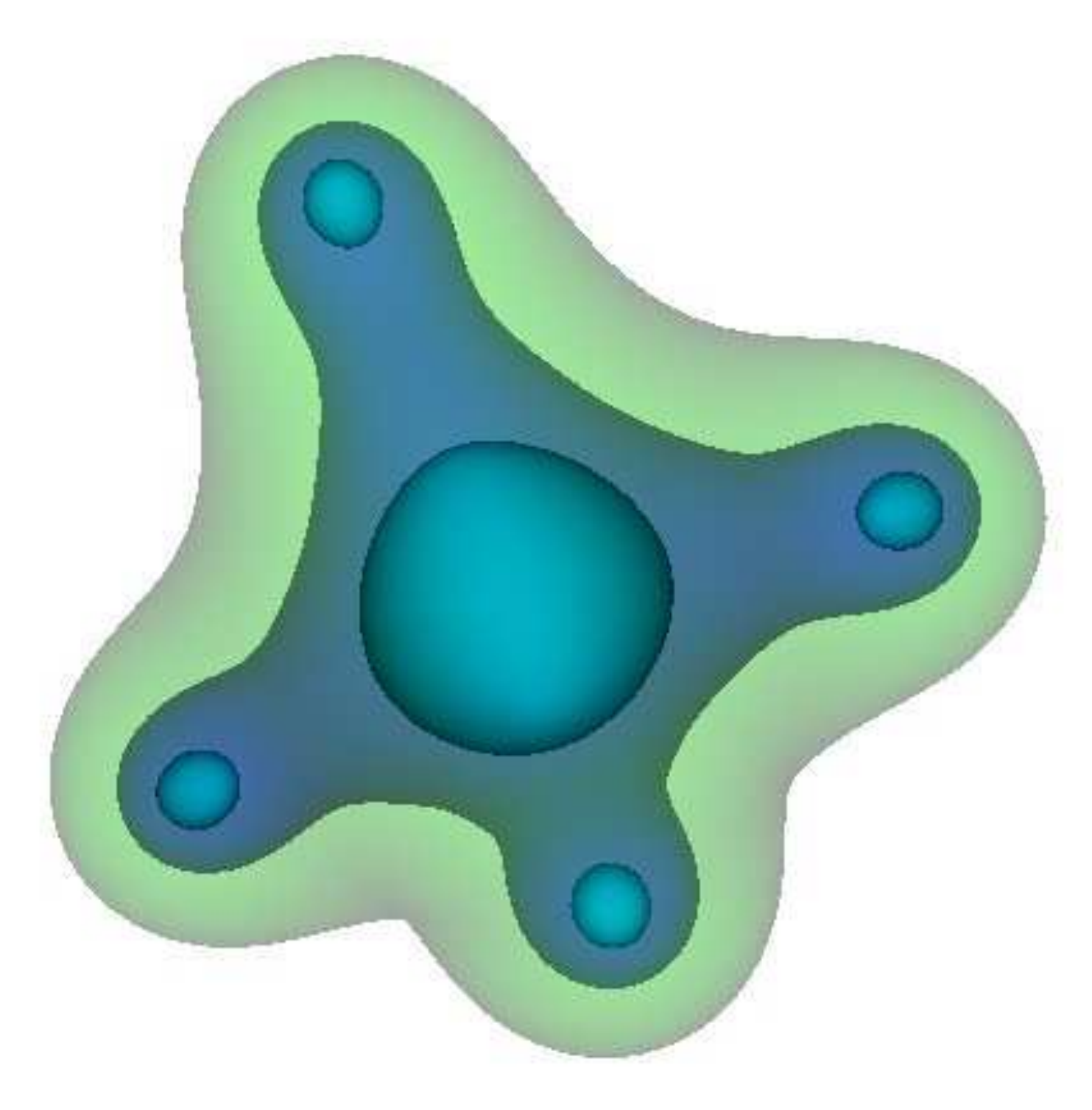}
\caption{\small{Electron density contours of methane molecule.}}
\label{fig:methanecontour}
\end{minipage}
\hfill
\begin{minipage}[t]{0.45 \textwidth}
\centering
\includegraphics[width = \textwidth]{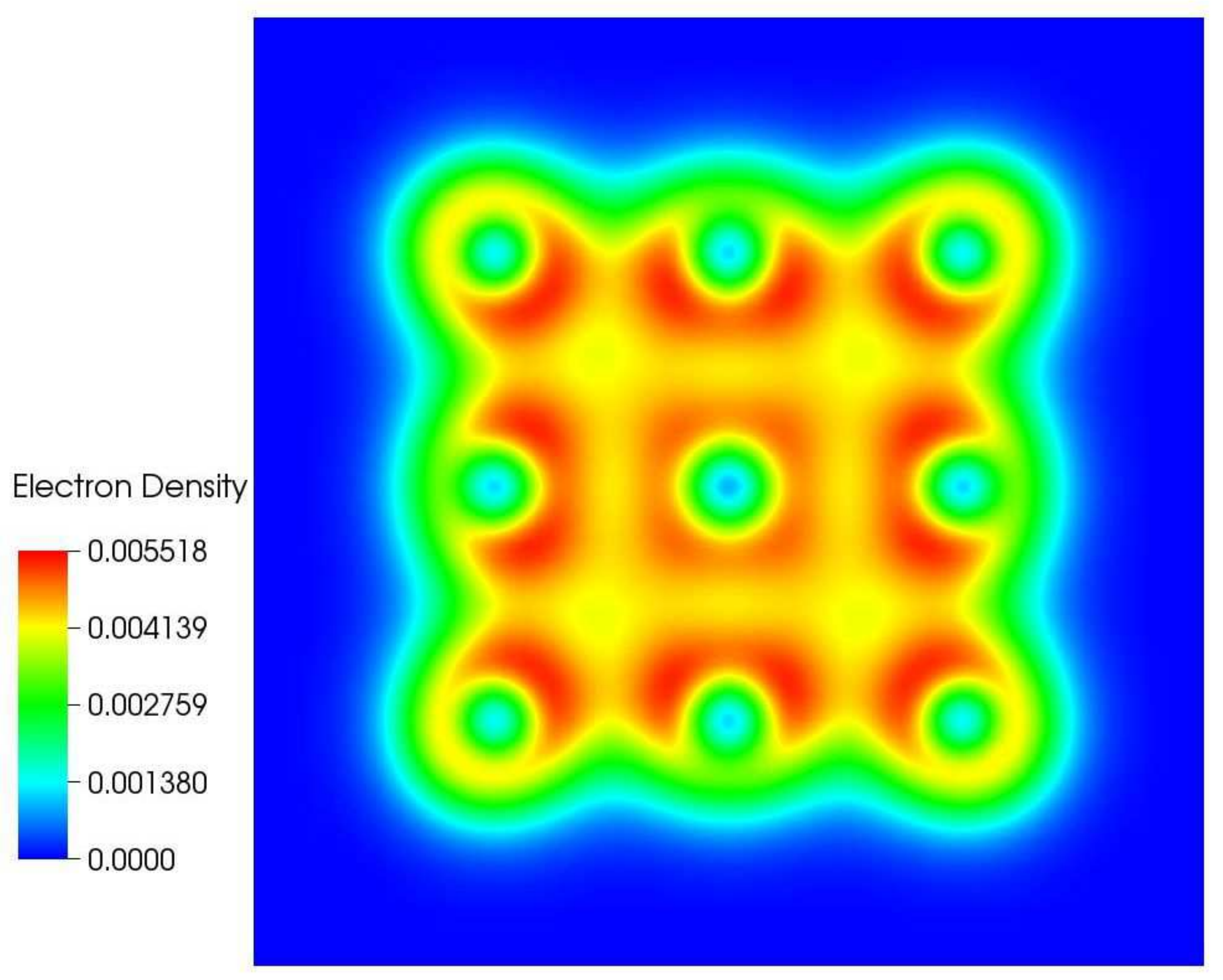}
\caption{\small{Electron density contours of barium $2\times 2\times 2$ BCC cluster.}}
\label{fig:bariumcontour}
\end{minipage}
\end{figure}

\subsubsection{Local Pseudopotential calculations}\label{sec:pseudopotential_calculations}
We now turn to  pseudopotential calculations in multi-electron systems. A pseudopotential constitutes the effective potential of the nucleus and core electrons experienced by the valence electrons. Pseudopotentials are constructed such that the wavefunctions of valence electrons outside the core and their corresponding eigenvalues are close to those computed using all-electron calculations. We note that, in the present work, we have restricted our investigation to local pseudopotential calculations as the present focus of this work is the demonstration of the computational efficiency afforded by adaptive higher-order finite-element discretizations. We note that the use of non-local pseudopotentials---for instance, the Troullier Martins pseudopotential in the Kleinman-Bylander form~\cite{nlp}---results in an additional sparse matrix in the discrete Hamiltonian whose sparsity is dependent on the extent of the non-local projectors. We expect that the consideration of non-local pseudopotentials will only have a marginal effect on the demonstrated performance of the algorithms and the scalability results, and a careful study of this aspect will be undertaken in our future investigations. In the present work, we use the local evanescent core pseudopotential~\cite{fiolhais} as a model pseudopotential to demonstrate our ideas. This pseudopotential has the following form
\begin{equation}
V^{I}_{ion} = -\frac{Z}{R_c}\left(\frac{1}{y}(1 - (1+\beta y)e^{-\alpha y}) - Ae^{-y}\right)\,\,,
\end{equation}
where $Z$ denotes the number of valence electrons and $y = |\br - \bR_{I}|/R_{c}$. The core decay length $R_c$ and $\alpha \geq 0$ are atom-dependent constants~\cite{fiolhais}. The constants $\beta$ and $A$ are evaluated by the following relations:
\begin{equation}
\beta = \frac{\alpha^{3} - 2\,\alpha}{4(\alpha^2 - 1)}\,,\;\;\;\;A = \frac{1}{2}\alpha^2 - \alpha \beta.
\end{equation}
\paragraph{Barium cluster:}
The first local pseudopotential calculation we present is a barium $2\times 2\times 2$ body-centered cubic (BCC) cluster with a lattice parameter of $9.5~a.u.$. A Chebyshev filter of order $16$ is employed to compute the occupied eigenspace, and a smearing factor $\sigma = 0.000634~Ha$ (T=200K) is used for the Fermi-Dirac smearing. The simulation domain used is a cubical domain of side $100$ $a.u.$ with Dirichlet boundary conditions employed on electronic wavefunctions and total electrostatic potential. The finite-element mesh for this molecule is constructed to be uniform in the cluster region where barium atoms are present, while coarse-graining away. The mesh coarsening rate in the vacuum is determined numerically by employing the asymptotic solution of the far-field electronic fields, estimated as a superposition of single atom far-field asymptotic fields, in equation~\eqref{optimmesh}. To this end, asymptotic behavior of the atomic wavefunctions in barium atom ($\psibar(r)$)  is chosen to be
\begin{equation}
\psibar(r) = \sqrt{\frac{\xi^3}{\pi}} \exp\Bigl[-\xi \;r\Bigr] \;\;\;\;\text{where}\;\;\;\;\xi = \sqrt{2\,\tilde{\epsilon}}\,,
\end{equation}
where $\tilde{\epsilon}$ (negative of the eigenvalue of the highest occupied eigenstate) is estimated from a coarse mesh calculation. The corresponding electrostatic potential is determined by the Poisson equation, with total charge density being equal to the sum of $2\bar{\psi}^{2}(r)$ and $-2\delta(r)$, and is given by
\begin{equation}
\bar{\phi}(r) = -2 \exp{(-2\xi \,r)}\left(\xi + \frac{1}{r}\right)\,.
\end{equation}
The numerical convergence study is conducted with both tetrahedral and hexahedral elements. Figure~\ref{fig:BariumClusterConvgRate} shows the rates of convergence for the various elements considered that are close to optimal rates of convergence and figure~\ref{fig:bariumcontour} show the relevant electron-density contours. The value of $E_0$ computed from equation~\eqref{convfit}, the reference ground-state energy per atom which is used to compute the relative errors in the energies, is found to be $-0.6386307998~Ha$. The energy per atom obtained with plane-wave basis using ABINIT with a cutoff energy $30~ Ha$ and cell-size $80~a.u.$ is $-0.638627743~Ha$. The main observation that distinguishes this study from the all-electron study is that all orders of interpolation provide much greater accuracies for the local pseudopotential calculations in comparison to all-electron calculations. Linear basis functions are able to approximate the ground-state energies up to relative errors of $10^{-3}$, whereas relative errors of $10^{-6}$ can be achieved with higher-order elements with polynomial degrees of four and above.

\begin{figure}[htbp]
\hfill
\begin{minipage}[t]{.48 \textwidth}
\centering \includegraphics[width = \textwidth]{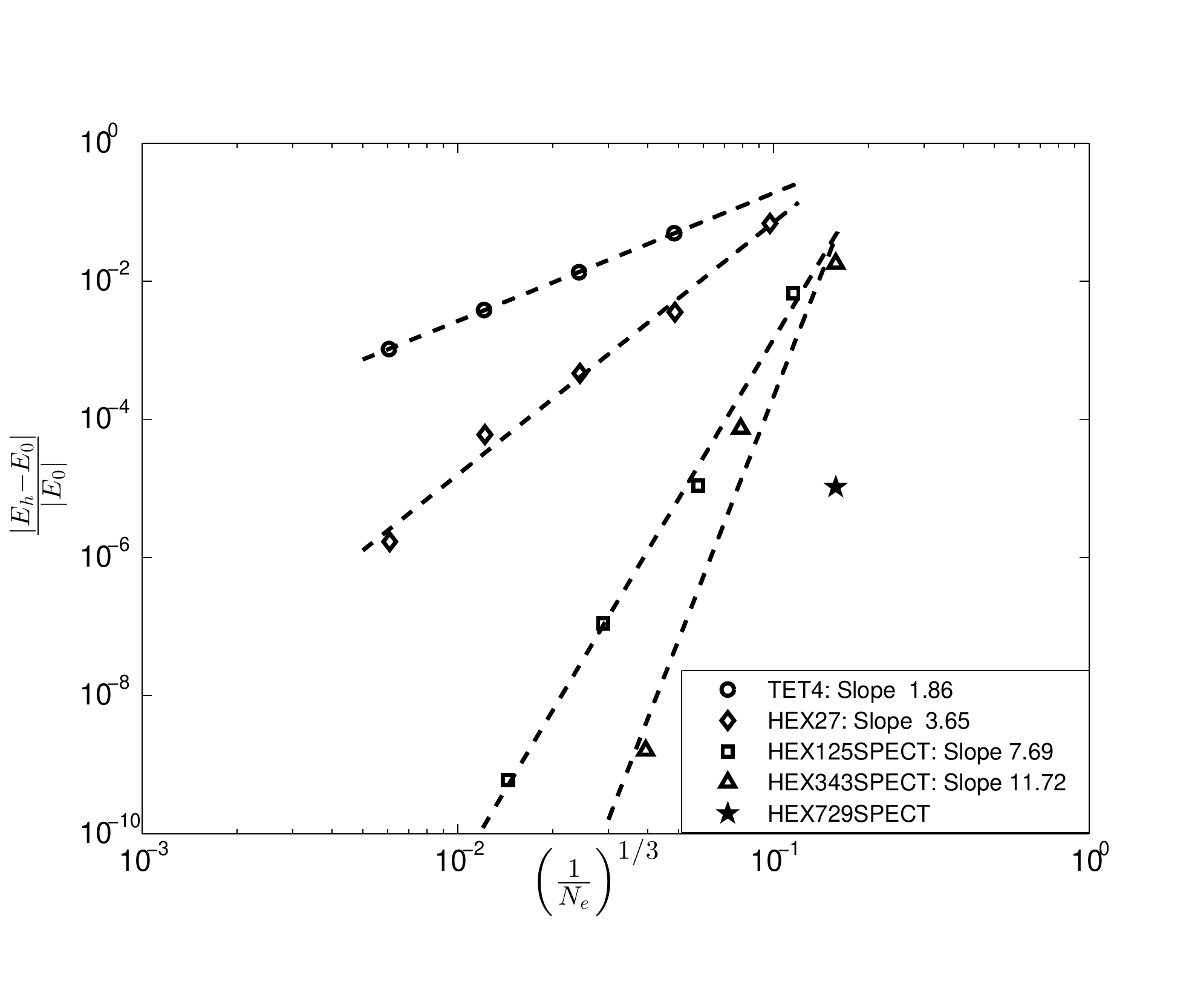}
\caption{\small{Convergence rates for the finite-element approximation of barium cluster.}}
\label{fig:BariumClusterConvgRate}
\end{minipage}
\hfill
\begin{minipage}[t]{.48 \textwidth}
\centering \includegraphics[width = \textwidth]{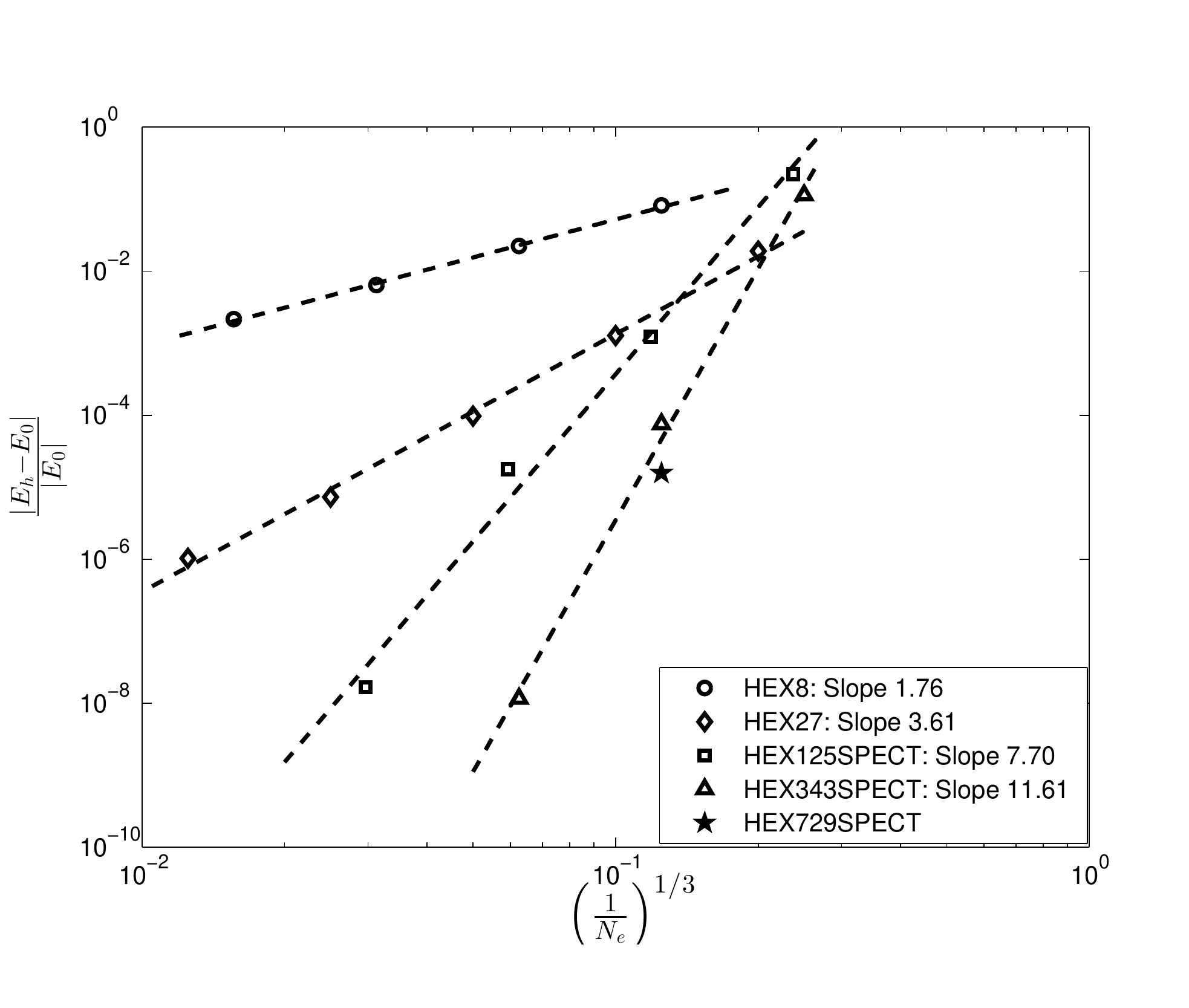}
\captionof{figure}{\small{Convergence rates for the finite-element approximation of bulk FCC calcium.}}
\label{fig:periodiccalc}
\end{minipage}
\end{figure}

\paragraph{Perfect crystal with periodic boundary conditions:}
The next example considered is that of a perfect calcium face-centered cubic (FCC) crystal with lattice constant $10.55~a.u.$. Bloch theorem~\cite{mermin} is used in the simulation with $10$ k-points (high symmetry) to sample the first Brillouin zone, which represents a quadrature rule of order 2~\cite{data}. The eigenspace is computed using the Krylov-Schur method, and a smearing parameter of $0.003168~Ha$ (T=1000K) is used in these simulations. Figure~\ref{fig:periodiccalc} shows the rates of convergence for the various higher-order finite-elements considered in the present work. The value of $E_0$ computed from equation~\eqref{convfit}, the reference bulk energy per atom, is computed to be $-0.729027041~Ha$. The bulk energy per atom obtained using ABINIT with a cutoff energy of $40~Ha$ is found to be $-0.72902775~Ha$. We note that the results are qualitatively similar to the local pseudopotential calculations carried on barium cluster.

\subsection{Computational cost}
We now examine the key aspect of computational efficiency afforded by the use of higher-order finite-element approximations in the Kohn-Sham DFT problem. As seen from the results in Section~\ref{RatesOfConv}, higher-order finite-element discretizations provide significantly higher accuracies with far fewer elements in comparison to linear finite-elements. However, the use of higher-order elements increases the per-element computational cost due to an increase in the number of nodes per element, which also results in an increase in the bandwidth of the Hamiltonian matrix. Further, higher-order elements require a higher-order accurate quadrature rule, which again increases the per-element computational cost. Thus, in order to unambiguously determine the computational efficiency afforded by higher-order finite-element discretizations, we measured the CPU-time taken for the simulations conducted on the aforementioned benchmark problems for a wide range of meshes providing different relative accuracies. All the simulations are conducted using meshes with the coarse-graining rates determined by the approach outlined in Section~\ref{sec:optimal_mesh}. All the numerical simulations reported in this work are conducted using a parallel implementation of the code based on MPI, and are executed on a parallel computing cluster with the following specifications: dual-socket six-core Intel Core I7 CPU nodes with 12 total processors (cores) per node, 48~GB memory per node, and 40~Gbps Infiniband networking between all nodes for fast MPI communication. The various benchmark calculations were executed using 1 to 12 cores, while the results for the larger problems discussed subsequently were executed on 48 to 96 cores. It was verified (see Section~\ref{sec:scalability}) that our implementation scales linearly on this parallel computing platform for the range of processors used, and hence the total CPU-times reported for the calculation are close to the wall-clock time on a single processor. The number of processors used to conduct ABINIT and GAUSSIAN simulations for the comparative studies, discussed subsequently, are carefully chosen to ensure scalability of these codes, and are typically less than 20 cores.

\subsubsection{Benchmark systems}\label{sec:comp_atom}
We first consider the benchmark systems comprising of boron atom, methane molecule, barium cluster and bulk calcium crystal. The mesh coarsening rates for these benchmark systems derived in Section~\ref{RatesOfConv} are employed in the present study. The number of elements are varied to obtain finite-element approximations with varying accuracies that target relative energy errors in the range of $10^{-1}-10^{-7}$. We employ the same numerical algorithms and algorithmic parameters---order of Chebyshev filter, finite-temperature smearing parameter---as discussed in Section~\ref{RatesOfConv} for the present study. The total CPU-time is measured for each of these simulations on the series of meshes constructed for varying finite-element interpolations and normalized with the longest time in the series of simulations for a given material system. The relative error in ground-state energy is then plotted against this normalized CPU-time. Figures~\ref{fig:SingleBoronTime}, \ref{fig:MethaneTime}, \ref{fig:BariumClusterTime} and \ref{fig:CalciumClusterTime} show these results for boron, methane molecule, barium cluster and bulk calcium crystal, respectively.

Our results show that the computational efficiency of higher-order interpolations improves as the desired accuracy of the computations increases, in particular for errors commensurate with chemical accuracy---order of 1 $meV$ per atom error for pseudopotential calculations and 1 $mHa$ per atom error for all-electron calculations. We note that a thousand-fold computational advantage is obtained with higher-order elements over linear TET4 element even for modest accuracies corresponding to relative errors of $10^{-2}$. For relative errors of $10^{-3}$, quadratic HEX27 element performance is similar to other finite-elements with quartic interpolation and beyond, and sometimes marginally better. However, all higher-order elements significantly outperform linear TET4 element. Considering relative errors of $10^{-5}$, quartic HEX125SPECT element performs better in comparison to quadratic HEX27 element almost by a factor of 10, while hexic HEX343SPECT element is computationally more efficient than HEX125SPECT element by a factor greater than three and this factor improves further for lower relative errors. The octic HEX729SPECT element performs only marginally better than the hexic element for relative errors lower than $10^{-5}$. Comparing the results across different materials systems, we observe that the performance of lower-order elements is inferior in the case of all-electron systems in comparison to  systems with smooth local pseudopotentials. For instance, at a relative error of $10^{-2}$, the solution time using TET4 is more than three orders of magnitude larger than HEX343SPECT for the case of methane molecule. However, the solution time is three orders of magnitude larger for TET4 over HEX343SPECT for the case barium cluster at a relative error of $10^{-3}$.

In summary, for chemical accuracies, the computational efficiency improves significantly with the order of the element up to sixth-order, with diminishing returns beyond. Further, the relative performance of higher-order elements with respect to linear TET4 element in the case of all-electron calculations is significantly better in comparison to local pseudopotential calculations. Lastly, qualitatively speaking, the sequence of graphs of relative error vs. normalized CPU-time for the various elements tend towards increasing accuracy and computational efficiency with increasing order of finite-element interpolation. However, we note that, for the systems studied, the point of diminishing returns in terms of computational efficiency of higher-order elements for relative errors commensurate with chemical accuracy is around sixth-order. As demonstrated in Appendix B, the primary reason for the diminishing returns is the increase in the cost of computing the Hamiltonian matrix which also increasingly dominates the total time with increasing order of the element.

\begin{figure}[!]
\hfill
\begin{minipage}[t]{.48 \textwidth}
\centering \includegraphics[width = \textwidth]{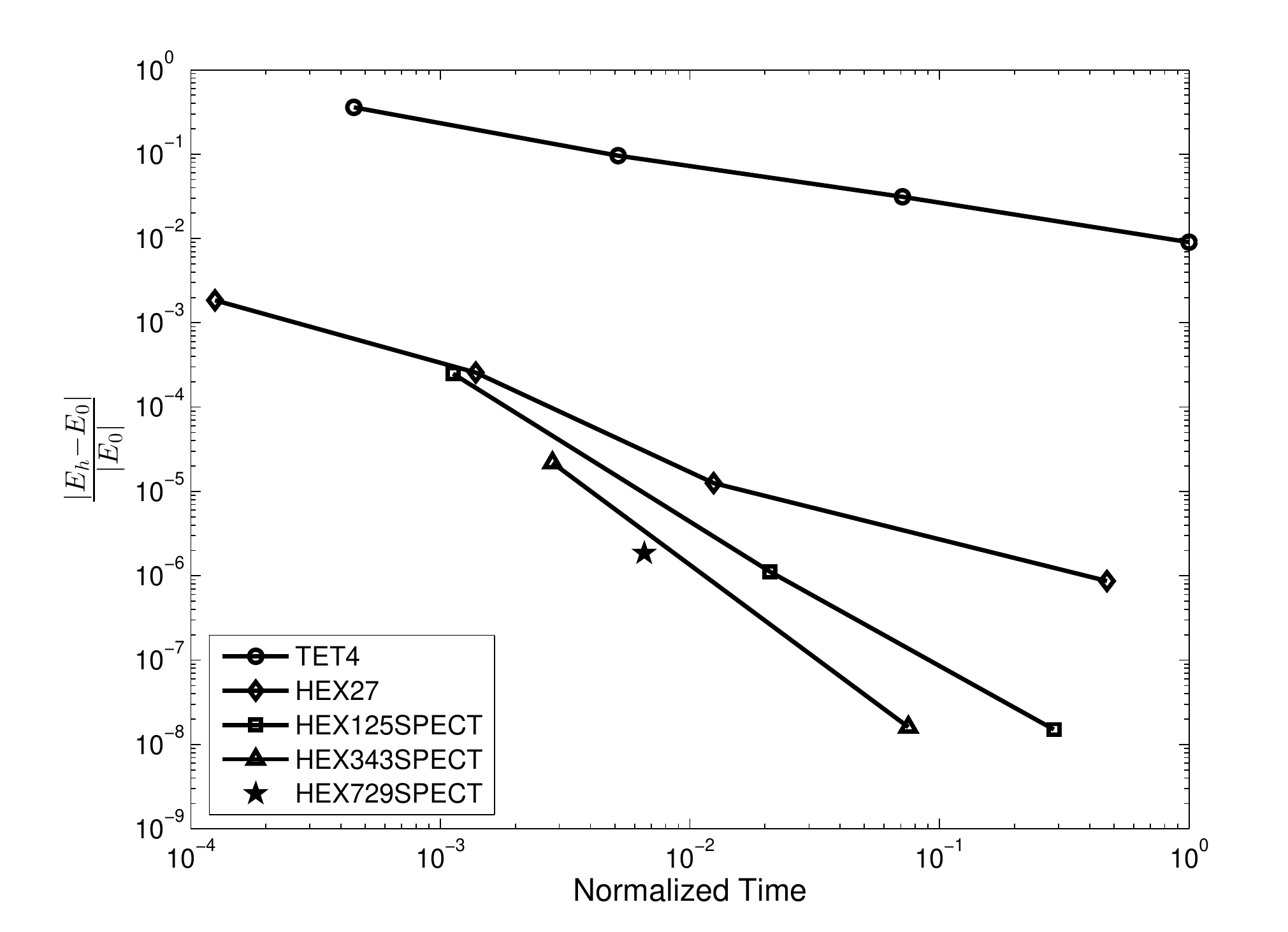}
\caption{\small{Computational efficiency of various orders of finite-element approximations.
Case study: boron atom.}}
\label{fig:SingleBoronTime}
\end{minipage}
\hfill
\begin{minipage}[t]{.48 \textwidth}
\centering \includegraphics[width = \textwidth]{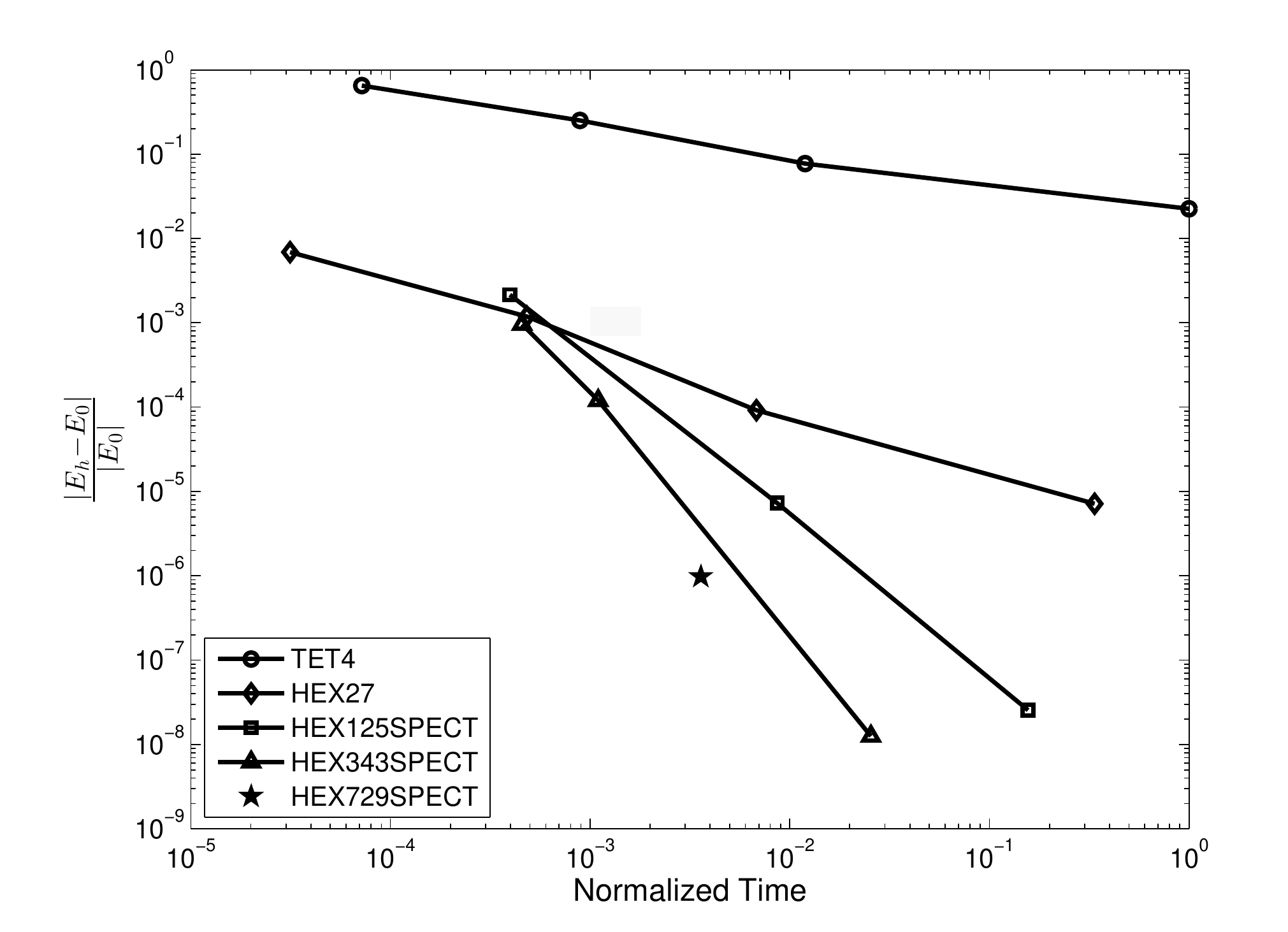}
\caption{\small{Computational efficiency of various orders of finite-element approximations. Case study: methane molecule.}}
\label{fig:MethaneTime}
\end{minipage}\\[0.3in]
\hfill
\begin{minipage}[b]{.48 \textwidth}
\centering \includegraphics[width = \textwidth]{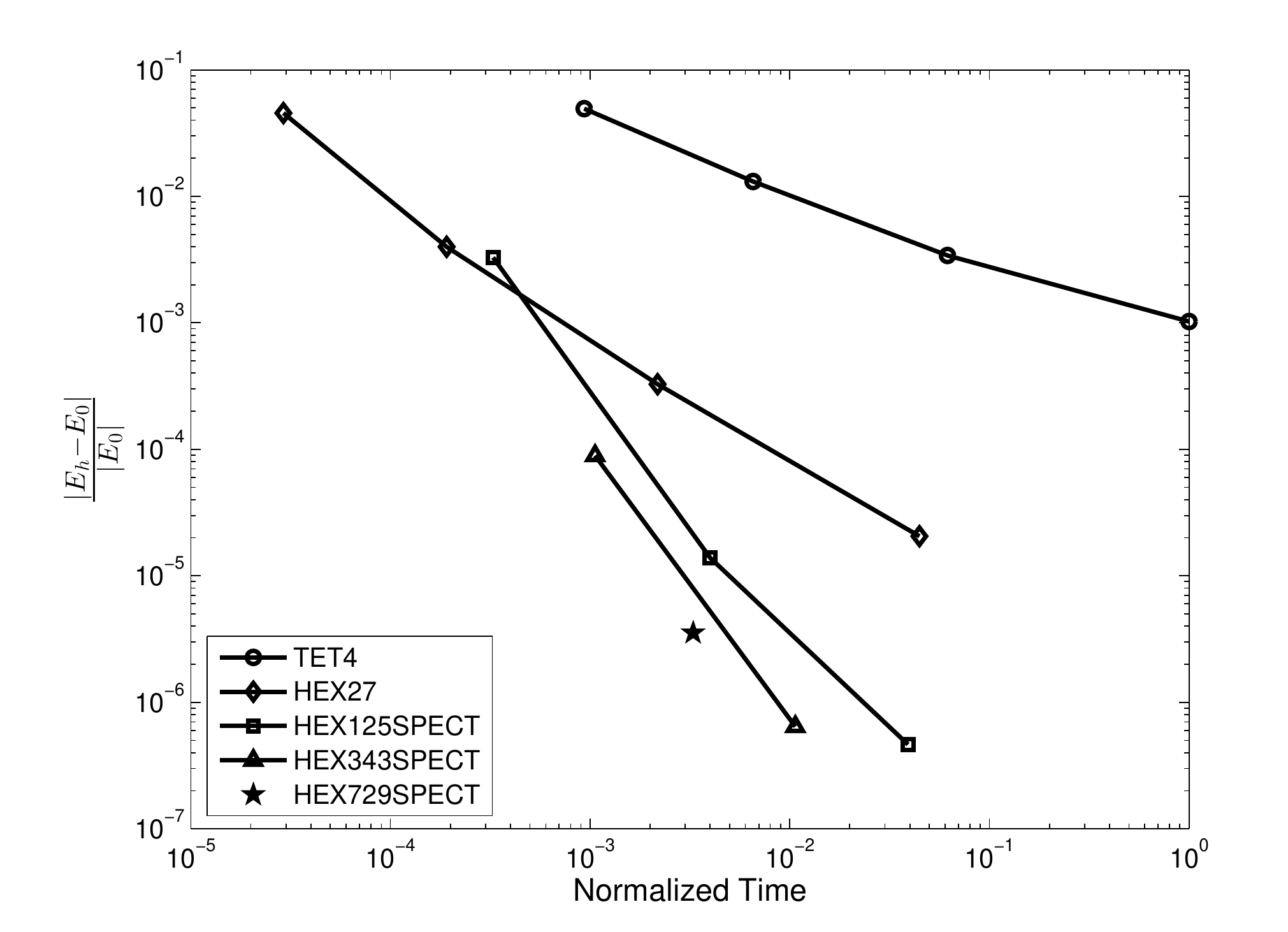}
\caption{\small{Computational efficiency of various orders of finite-element approximations. Case study: barium $2\times 2\times 2$ BCC cluster.}}
\label{fig:BariumClusterTime}
\end{minipage}
\hfill
\begin{minipage}[b]{.48 \textwidth}
\centering \includegraphics[width = \textwidth]{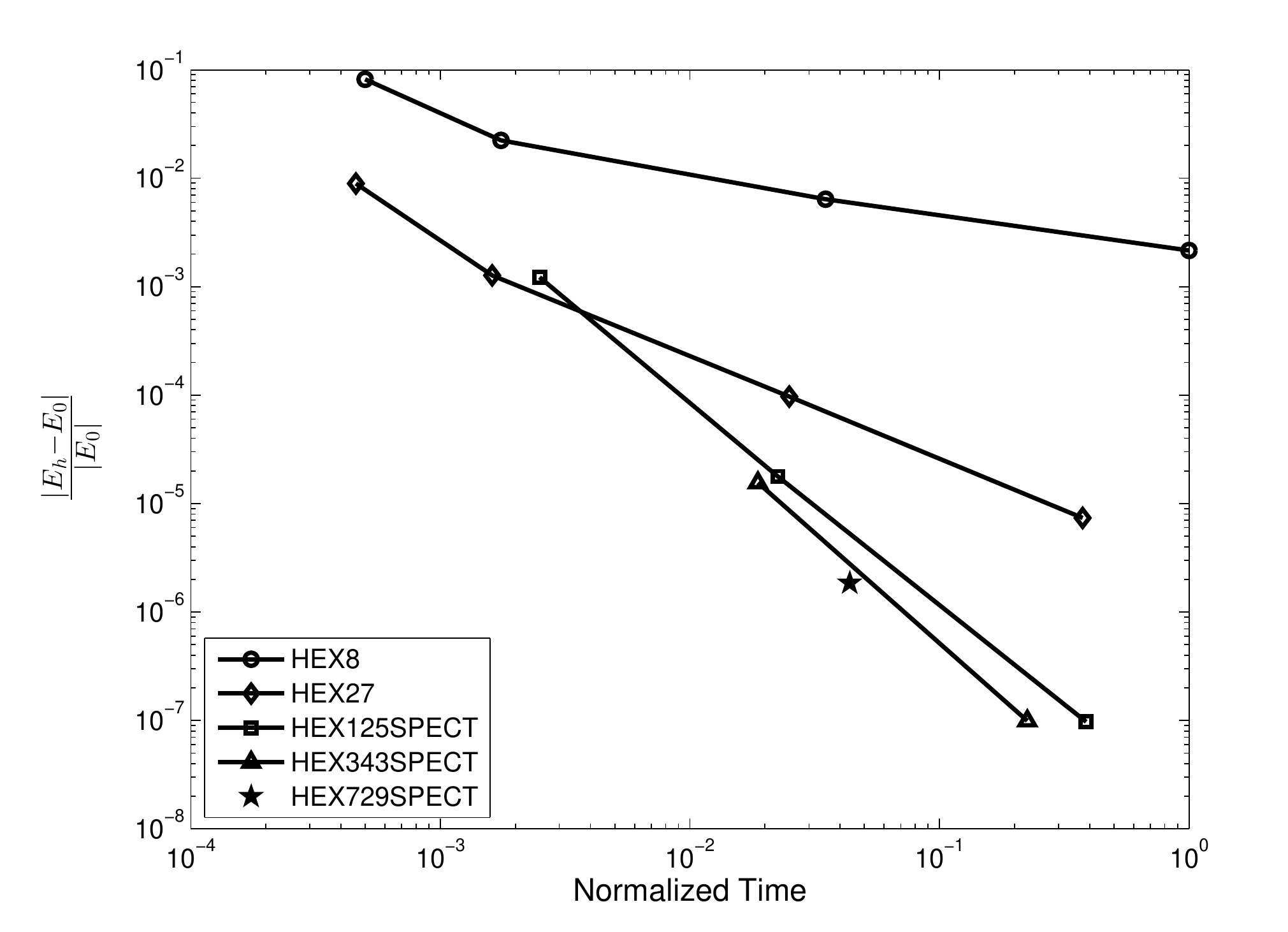}
\caption{\small{Computational efficiency of various orders of finite-element approximations. Case study: bulk calcium FCC crystal.}}
\label{fig:CalciumClusterTime}
\end{minipage}
\end{figure}

\subsubsection{Large materials systems}
In this section, we further investigate the computational efficiency afforded by higher-order finite-elements by considering larger
material systems involving both local pseudopotential and all-electron calculations. As a part of this investigation, we demonstrate the effectiveness of the higher-order finite-elements by comparing the solution times of calculations with local pseudopotentials against plane-wave basis set and solution times of all-electron calculations against a Gaussian basis set providing similar relative accuracy in the ground-state energy. The systems chosen as a part of this study are aluminium clusters containing $3\times 3\times 3$, $5\times 5\times 5$, $7\times 7\times 7$ FCC unit cells for the case of pseudopotential calculations. A graphene sheet containing 100 atoms and a coordination complex, tris (bipyridine) ruthenium, containing 61 atoms are chosen in the case of all-electron calculations.
\paragraph{Local pseudopotential calculations:}
The pseudopotential calculations on aluminum clusters are conducted using the evanescent core pseudopotential~\cite{fiolhais}. All the simulations in the case studies involving local pseudopotentials use superposition of single-atom electron densities as the initial guess for the electron density in the first SCF iteration. We used the Krylov-Schur iteration for solving the eigenvalue problem in the first SCF iteration and used Chebyshev filtered subspace iteration for the subsequent SCF iterations. The order of Chebyshev filters used for the $3\times 3\times 3$, $5\times 5\times 5$ and $7\times 7\times 7$ aluminum clusters are $12$, $30$ and $50$ respectively. All simulations are conducted using a finite temperature Fermi-Dirac smearing parameter of $0.0003168~Ha$ (T=100K). In order to conduct a one-to-one comparison, the plane-wave simulations are also conducted using the same pseudopotential and finite temperature Fermi-Dirac smearing used in the finite-element simulations.
\subparagraph{Aluminium $\mathbf{3\times 3\times 3}$ cluster:}\label{sec:Al3x3x3}
We first consider an aluminium cluster containing $3\times 3\times 3$ FCC unit cells with a lattice spacing of $7.45~a.u.$. The system comprises of 172 atoms with 516 electrons. The finite-element mesh for this calculation is chosen to be uniform in the cluster region containing aluminium atoms, while coarse-graining away. The mesh coarsening rate in the vacuum is determined numerically by employing the asymptotic solution of the far-field electronic fields, estimated as a superposition of single atom far-field asymptotic fields, in equation~\eqref{optimmesh}. To this end, the asymptotic behavior of atomic wavefunctions in an aluminium atom ($\psibar(r)$) is chosen to be
\begin{equation}\label{mesh_Al_psi}
\psibar(r) = \sqrt{\frac{\xi^3}{\pi}} \exp\Bigl[-\xi \;r\Bigr] \;\;\;\;\text{where}\;\;\;\;\xi = \sqrt{3\,\tilde{\epsilon}}\,,
\end{equation}
where $\tilde{\epsilon}$ (negative of the eigenvalue of the highest occupied eigenstate) is determined from a single aluminum atom coarse mesh calculation. The corresponding total electrostatic potential, governed by the Poisson equation with total charge density being equal to the sum of $3\bar{\psi}^{2}(r)$ and $-3\delta(r)$, is given by
\begin{equation}\label{mesh_Al_phi}
\bar{\phi}(r) = -3 \exp{(-2\xi \,r)}\left(\xi + \frac{1}{r}\right)\,.
\end{equation}

\begin{table}[htbp]
\caption{\small{Convergence with finite-element basis for a $3\times 3\times 3$ FCC aluminum cluster using HEX125SPECT element.}}
 \begin{center}
 \begin{tabular}{|c|c|c|}
   \hline
 Degrees of freedom & Energy per atom (eV) & Relative error  \\ \hline\hline
$184,145$ &  -54.1076597  & 3.4 $\times 10^{-2}$   \\ \hline
$1,453,089$  & -56.0076146 & 1.8 $\times 10^{-4}$  \\ \hline
$11,546,177$ & -56.01788889 & 1.3 $\times 10^{-6}$ \\\hline \hline
\end{tabular}
\end{center}\label{tab:conv3x3x3Cluster}
\end{table}

We obtain the converged value of the ground-state energy by following the procedure outlined in Section~\ref{RatesOfConv}. We use a sequence of increasingly refined HEX125SPECT finite-element meshes on a cubic simulation domain of side $400~a.u.$, and compute the ground-state energy $E_h$ for these meshes which are tabulated in Table~\ref{tab:conv3x3x3Cluster}. Using the extrapolation procedure discussed in Section~\ref{RatesOfConv} (equation~\eqref{convfit}), we compute the reference ground-state energy (energy per atom) to be $E_0=-56.0179603~eV$. The relative errors reported in Table~\ref{tab:conv3x3x3Cluster} are with respect to this reference energy, and this reference energy will be used to compute the relative errors for all subsequent finite-element and plane-wave basis simulations for this material system.

\begin{table}[htbp]
\caption{\small{Comparison of higher-order finite-element (FE) basis with plane-wave basis for a $3\times 3\times 3$ FCC aluminum cluster.}}
\begin{center}
\begin{tabular}{|p{5cm}|p{2.5cm}|p{2.8cm}|c|c|}
   \hline
  Type of basis set & Energy (eV) per atom & Abs. error (eV) per atom & Rel. error  & Time (CPU-hrs) \\ \hline\hline
Plane-wave basis (cut-off $30~Ha$; cell-size of $60~a.u.$; $847,348$ plane waves)&  -56.0181429& 0.00018 &3.3 $\times 10^{-6}$  & $646$ \\ \hline
FE basis (HEX343SPECT; $2,808,385$ nodes; domain size: $200~a.u.$) & -56.0177597& 0.0002& 3.6 $\times 10^{-6}$  & $371$\\
 \hline
\end{tabular}
\end{center}\label{tab:eff3x3x3cluster}
\end{table}

In order to assess the performance of higher-order finite-elements on this material system, we conduct the finite-element simulation with a mesh containing HEX343SPECT elements and compare the computational CPU-time against a plane-wave basis code ABINIT~\cite{ABINIT} solved to a similar relative accuracy in the ground-state energy with respect to reference value $E_0$ obtained above. The finite-element simulation has been performed on a cubic domain size of $200~a.u.$ with a mesh coarsening rate away from the cluster of atoms as determined using equations~\eqref{optimmesh},~\eqref{mesh_Al_psi},~\eqref{mesh_Al_phi}. The resulting mesh contains $12,800$ HEX343SPECT elements with $2,808,385$ nodes. The plane-wave basis simulation has been performed by using a cell-size of $60~a.u.$ and a cut-off energy of $30~Ha$ with one k-point to obtain the ground-state energy of similar relative accuracy(~0.0002eV/atom) as the finite-element simulation. The computational times for the finite-element basis and the plane-wave basis for the full self-consistent solution are tabulated in Table~\ref{tab:eff3x3x3cluster}. These results demonstrate that the performance of higher-order finite-element discretization is comparable, in fact better by a two-fold factor, to the plane-wave basis for this material system. Figure~\ref{fig:al3x3x3Contour} shows the electron density contours on the mid-plane of the $3\times 3\times 3$ aluminum cluster from the finite-element simulation.
\begin{figure}[h]
\hfill
\begin{minipage}[t]{0.48 \textwidth}
\centering
\includegraphics[width=\textwidth]{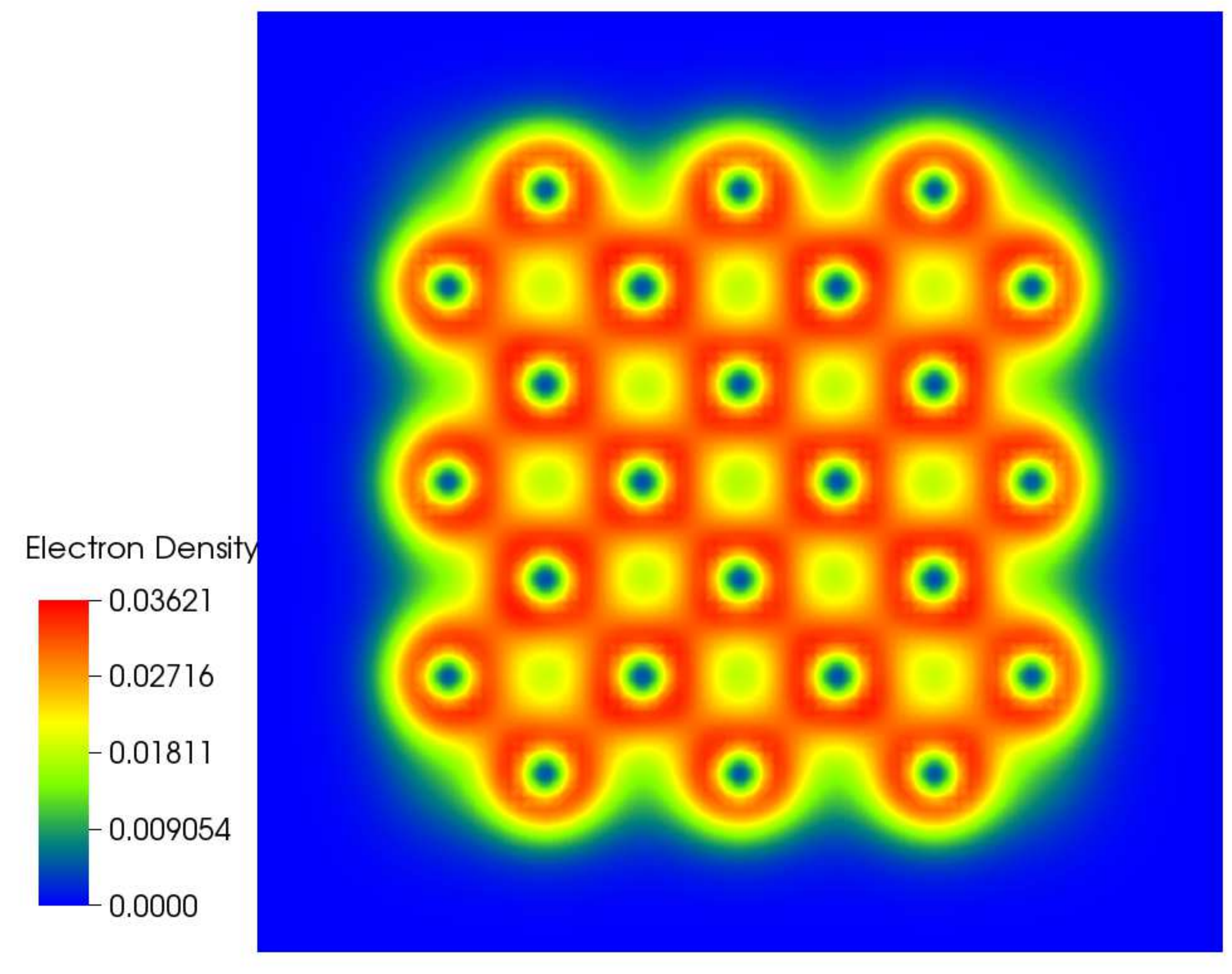}
\caption{\small{Electron density contours of $3\times 3\times 3$ FCC aluminium cluster.}}
\label{fig:al3x3x3Contour}
\end{minipage}
\hfill
\begin{minipage}[t]{0.48 \textwidth}
\centering
\includegraphics[width=\textwidth]{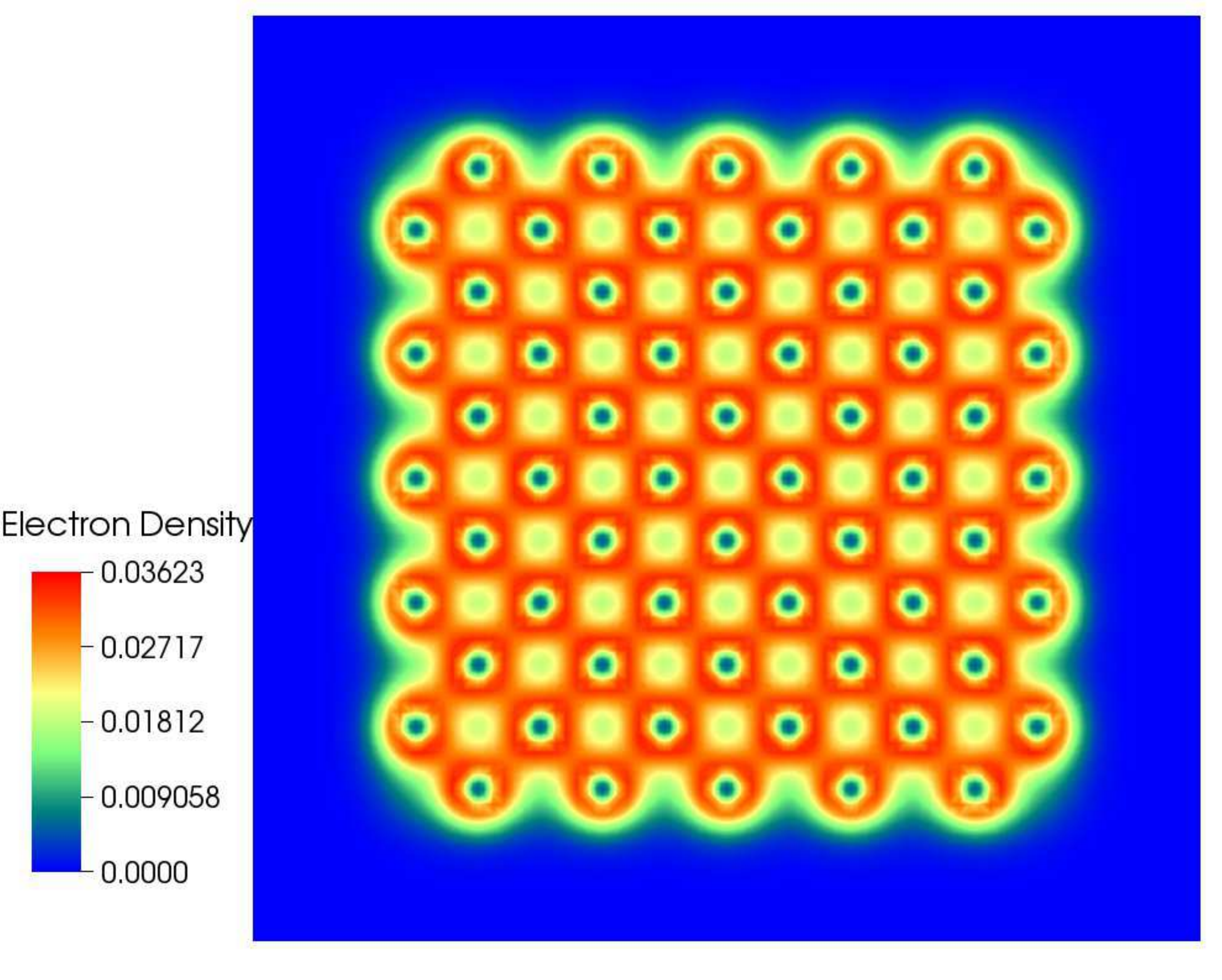}
\caption{\small{Electron density contours of $5\times 5\times 5$ FCC aluminum cluster.}}
\label{fig:al5x5x5Contour}
\end{minipage}
\end{figure}
\subparagraph{Aluminium $\mathbf{5\times 5\times 5}$ cluster:}
We next consider an aluminium cluster containing $5\times 5\times 5$ FCC unit cells with a lattice spacing of $7.45~a.u.$. This material system comprises of $666$ atoms with $1998$ electrons. The finite-element mesh is constructed along similar lines as the $3\times 3\times 3$ cluster, where a uniform mesh resolution is chosen in the cluster region containing aluminium atoms and coarse-graining away into the vacuum with a numerically determined coarsening rate as discussed earlier. As before, we first obtain the reference ground-state energy by using a sequence of increasingly refined HEX125SPECT finite-element meshes with a cubic simulation domain of side $800~a.u.$ and extrapolating the computed ground-state energies on these meshes (cf. Table~\ref{tab:conv5x5x5Cluster}). The reference ground-state energy (energy per atom), thus determined, is $E_0=-56.0495071~eV$.
\begin{table}[htbp]
\caption{\small{Convergence with finite-element basis for a $5\times 5\times 5$ FCC cluster using HEX125SPECT element.}}
\begin{center}
 \begin{tabular}{|c|c|c|}
   \hline
Degrees of freedom & Energy per atom(eV) &Relative error  \\ \hline\hline
$394,169$ & -54.8536312 & 2.1 $\times 10^{-2}$   \\ \hline
$3,124,593$  & -56.0425334 & 1.2 $\times 10^{-4}$  \\ \hline
$24,883,937$ & -56.0494500 & 1.01 $\times 10^{-6}$ \\\hline \hline
\end{tabular}
\end{center} \label{tab:conv5x5x5Cluster}
\end{table}

We now assess the performance of higher-order finite-elements on this material system in comparison to a plane-wave basis. The finite-element simulation in this case has been performed on a simulation domain size of $400~a.u.$ containing $36,064$ HEX343SPECT elements with $7,875,037$ nodes. The plane-wave basis simulation conducted using the ABINIT package has been performed on a cell-size of $80~a.u.$ and a cut off energy of $30~Ha$ with one k-point to sample the Brillouin zone to obtain the ground-state energy of similar accuracy with respect to the reference value $E_0$ obtained above. The solution time for the finite-element basis and the plane-wave basis are tabulated in Table~\ref{tab:eff5x5x5Cluster}, which shows that using higher-order finite-elements one can achieve similar computational efficiencies as afforded by a plane-wave basis, at least in the case of non-periodic calculations. Figure~\ref{fig:al5x5x5Contour} shows the electron density contours on the mid-plane of the $5\times 5\times 5$ FCC cluster from the finite-element simulation.

\begin{table}[htbp]
\caption{\small{Comparison of higher-order finite-element (FE) basis with plane-wave basis sets for a $5\times 5\times 5$ FCC aluminum cluster.}}
\begin{center}
 \begin{tabular}{|p{5cm}|p{2.5cm}|p{2.8cm}|c|c|}
   \hline
Type of basis set & Energy (eV) per atom & Abs. error (eV) per atom & Rel. error  & Time (CPU-hrs) \\ \hline\hline
Plane-wave basis (cut-off $30~Ha$; cell-size of 80 a.u; $2,009,661$ plane waves)& -56.0506841& 0.0012 & 2.1 $\times 10^{-5}$  & 7307  \\ \hline
FE basis ( HEX343SPECT; $7,875,037$ nodes; domain size: $400~a.u.$) & -56.04906430&0.00044&7.9 $\times 10^{-6}$ & 6619 \\
 \hline
\end{tabular}
\end{center}\label{tab:eff5x5x5Cluster}
\end{table}

\subparagraph{Aluminium $\mathbf{7\times 7\times 7}$ cluster:}
As a final example in our case study with local pseudopotential calculations, we study an aluminium cluster containing $7\times 7\times 7$ FCC unit cells with a lattice spacing of $7.45~a.u.$ This material system comprises of 1688 atoms with 5064 electrons. We only use the finite-element basis to simulate this system as the plane-wave basis calculation was beyond reach for this material system with the computational resources at our disposal. The finite-element simulation has been performed on a cubic simulation domain with a side of $800~a.u.$. The finite-element mesh was constructed as described in the simulation of other aluminum clusters, and comprised of $69,984$ HEX343SPECT elements with $15,257,197$ nodes. The computed energy per atom for this aluminum cluster is $-56.06826762~eV$, and figure~\ref{fig:al7x7x7Contour} shows the electron density contours on the mid-plane of the cluster.
\begin{figure}[htbp]
\begin{center}
\includegraphics[width=0.6\textwidth]{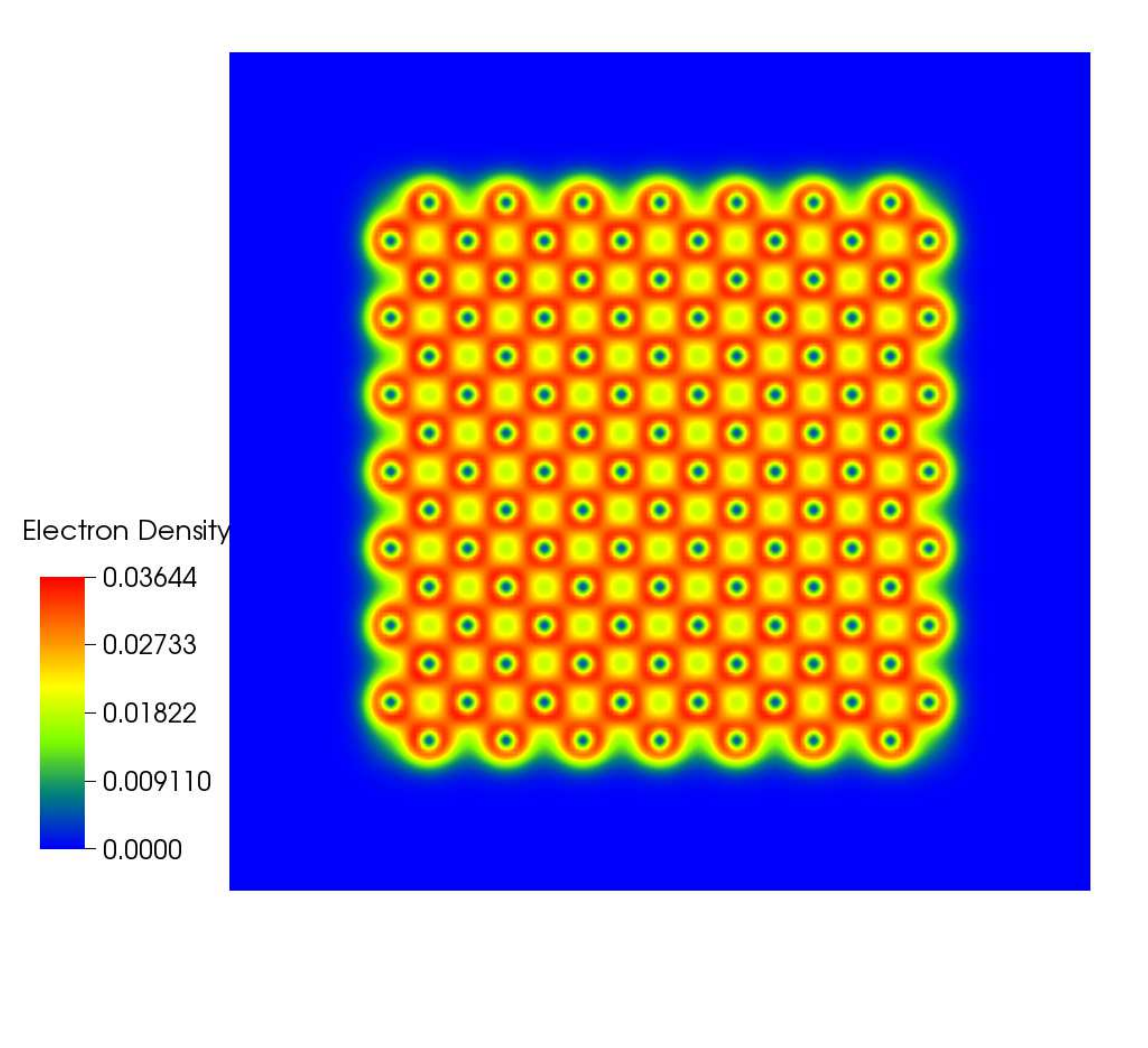}
\caption{\small{Electron density contours of $7\times 7\times 7$ FCC aluminium cluster.}}
\label{fig:al7x7x7Contour}
\end{center}
\end{figure}
\paragraph{All-electron calculations:}
We now demonstrate the performance of higher-order finite-element discretization in the case of all-electron calculations, by considering a graphene sheet and a transition metal complex, namely, tris (bipyridine) ruthenium as our benchmark problems. In these calculations which employ HEX125SPECT finite-element, the initial guess for electron-density in the first SCF iteration is computed by interpolating the self-consistent solution obtained from a lower-order HEX27 finite-element mesh. The computational times reported for these calculations include the time taken to generate the initial guess.
\subparagraph{Graphene sheet:}
We begin with a graphene sheet containing 100 atoms (600 electrons) with a C-C bond length of $2.683412~a.u.$. We first obtain a converged value of the ground-state energy by conducting simulations using the GAUSSIAN package~\cite{GAUSSIAN} using the polarization consistent DFT basis sets~(pc-n), which have been demonstrated to provide a systematic convergence in Kohn-Sham DFT calculations~\cite{jensen}. Since these basis sets are not directly available in the GAUSSIAN package, we introduce them as an external basis set for conducting these simulations. The ground-state energy value obtained for triple-zeta pc-3 basis set is taken as the reference value ($E_0$) in this study, which is computed to be $E_0=-37.7619141~Ha$ per atom. We note that we did not use the extrapolation procedure outlined in Section~\ref{RatesOfConv}, as it was computationally beyond reach with our resources.

We assess the performance of higher-order finite-elements on this material system by comparing the computational CPU-time against the pc-2 basis set, which provides similar relative accuracy in the ground-state energy with respect to the $E_0$ determined above. The finite-element mesh for this problem is chosen to be uniform in the region containing carbon atoms with local refinement around each atom while coarse-graining away into vacuum. The mesh coarsening rate in the vacuum is determined numerically by employing the asymptotic solution of the far-field electronic fields, estimated as a superposition of single atom far-field asymptotic fields, in equation~\eqref{optimmesh}. To this end, the asymptotic behavior of the atomic wavefunctions in carbon atom ($\psibar(r)$) is chosen to be as in equation~\eqref{carbon}. Since the GAUSSIAN package does not account for partial occupancy of energy levels, we suppress the Fermi-Dirac smearing in the finite-element simulations for the present case in order to conduct a one-to-one comparison. A Chebyshev filter of order 500 is used in this simulation. The simulation domain used is a cubical domain of side $300$ $a.u.$ with Dirichlet boundary conditions employed on electronic wavefunctions and total electrostatic potential. Table~\ref{tab:grap100} shows the relevant results of the simulation with figure~\ref{fig:grap100} showing the electron density contours of the graphene sheet. We remark that the finite-element simulation with HEX125SPECT elements is around ten-fold slower than the GAUSSIAN simulation with pc basis set.
\begin{table}[htbp]
\caption{\small{100 atom graphene sheet (600 electrons).}}
 \begin{center}
 \begin{tabular}{|p{5cm}|p{2.5cm}|p{2.8cm}|c|c|}
   \hline
 Type of basis set &  Energy (Ha) per atom & Abs. error (Ha) per atom & Rel. error & Time (CPU-hrs) \\ \hline\hline
pc2 (Gaussian; $3,000$ basis functions)  & -37.757954 & 0.00396 & 1.06 $\times 10^{-4}$  & 666 \\ \hline
FE basis (HEX125SPECT; $8,004,003$ nodes) &  -37.757382 &0.00452&1.2 $\times 10^{-4}$  & 7461\\
 \hline
\end{tabular}
\end{center}\label{tab:grap100}
\end{table}
\begin{figure}[htbp]
\begin{center}
\includegraphics[width=0.75\textwidth]{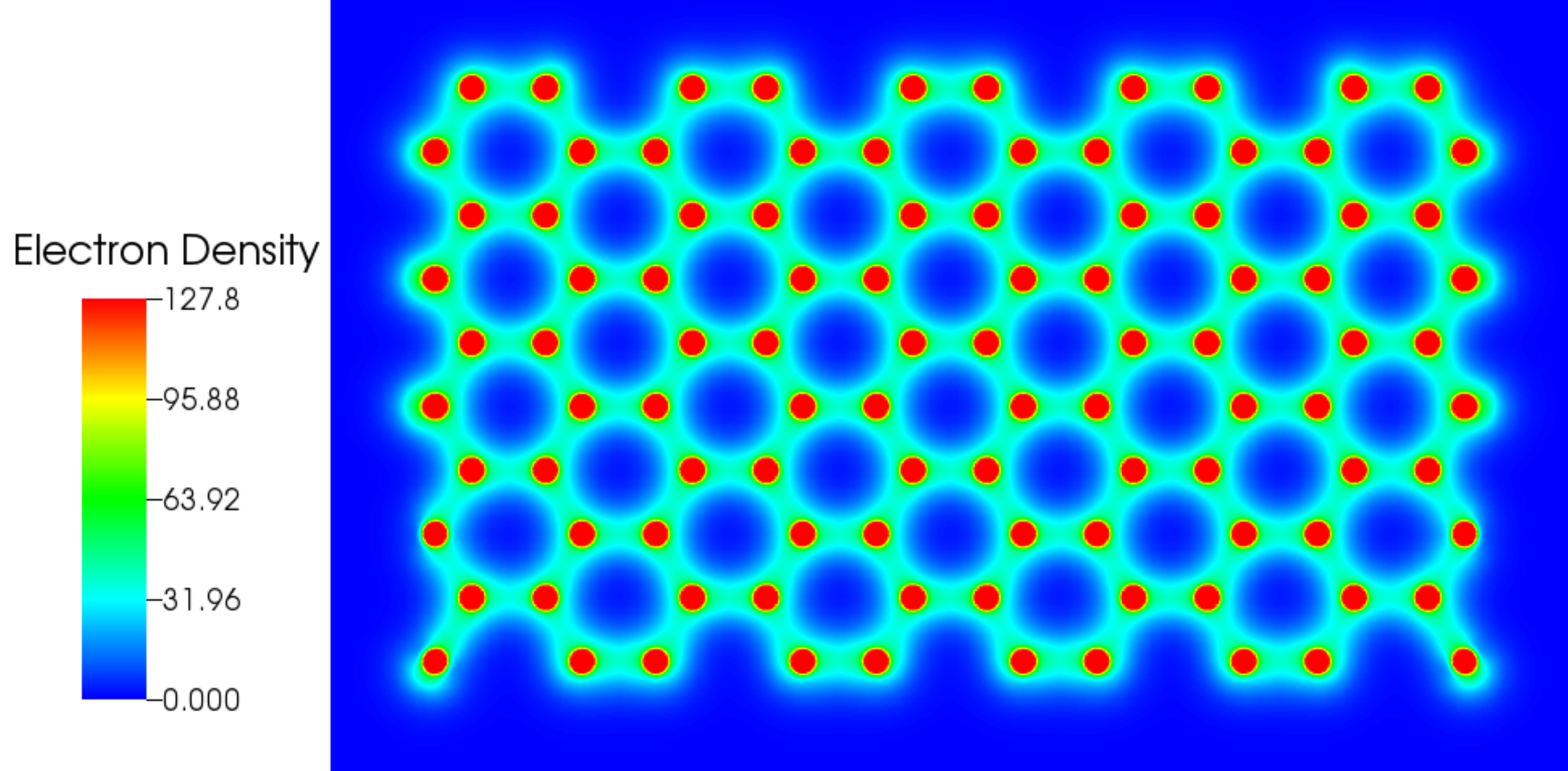}
\caption{\small{Electron density contours of a graphene sheet containing 100 atoms.}}
\label{fig:grap100}
\end{center}
\end{figure}

\subparagraph{Tris (bipyridine) ruthenium:}
We now demonstrate the performance of our numerical algorithms on a compound involving a moderately heavy metal. We choose Tris (bipyridine) ruthenium complex (TBR) as our benchmark problem, which belongs to the class of transition metal complexes that has attracted significant attention because of its distinctive optical properties~\cite{TBR}. Though the prototype complex TBR is extensively studied as a dication, we consider the charge neutral complex in this case study. The geometric structure of this compound was determined using the geometry optimization option in the GAUSSIAN package. The resulting compound consists of a central ruthenium atom bonded to nitrogen atoms (ligands) of the three bipyridine rings as shown in Figure~\ref{fig:TBRGeom}. The compound contains a total of 61 atoms (290 electrons) comprising of 30 carbon atoms, 24 hydrogen atoms, 6 nitrogen atoms and 1 ruthenium atom.

We conducted the simulations using GAUSSIAN package in the following manner. Polarized consistent DFT basis sets provide a series of improved basis sets for carbon, hydrogen and nitrogen atoms, but are not suited for the ruthenium atom. Hence, we first conducted the GAUSSIAN simulation using the most refined polarized consistent basis set namely pc-4 basis functions for carbon, hydrogen and nitrogen, and used the polarized valence double zeta basis function designed for DFT (DZDFTO) for ruthenium. However, the self-consistent iteration did not converge for this choice of basis sets. Hence, we conducted the GAUSSIAN simulation with a coarser polarized consistent quadruple zeta basis (pc-3) functions for carbon, hydrogen and nitrogen atoms, and used the DZDFTO basis for ruthenium. The self-consistent iteration did converge for this choice, and the ground-state energy and solution time for this case are tabulated in Table~\ref{tab:TBR}.

We now assess the performance of higher-order finite-elements on this material system.  The finite-element mesh for this problem is chosen to be uniform in the region where the molecule is present with local refinement around each atom while coarse-graining away into the vacuum. The mesh coarsening rate in the vacuum is determined numerically by employing the asymptotic solution of far-field asymptotic fields, estimated as a superposition of single atom far-field asymptotic fields in equation~\ref{optimmesh}. The finite-element simulation is conducted using HEX125SPECT elements with a Chebyshev filter of order 500 and a fermi-dirac smearing parameter of $0.00158~Ha$~(T=500K). The simulation domain used is a cubical domain of side $200$ $a.u.$ with Dirichlet boundary conditions employed on electronic wavefunctions and total electrostatic potential. Table~\ref{tab:TBR} demonstrates the relevant results for the finite-element simulation with figure~\ref{fig:TBR} showing the electron-density contours. We remark that the ground-state energy per atom obtained from finite-element simulation differs by $0.00146~Ha$ in comparison with aforementioned GAUSSIAN simulation. As observed in the case of graphene sheet, the finite-element simulation is ten-fold slower than the GAUSSIAN simulation.

 \begin{table}[htbp]
\caption{\small{Tris(bipyridine)ruthenium (290 electrons).}}
 \begin{center}
 \begin{tabular}{|p{5cm}|c|c|}
   \hline
  Type of basis set &  Energy (Ha) per atom &  Time (CPU-hrs) \\ \hline\hline
pc3 (Gaussian; $3,156$ basis functions)  & -96.923328 & 311   \\ \hline
FE basis (HEX125SPECT; $10,054,041$ nodes) &  -96.924636 & 3927\\
 \hline
\end{tabular}
\end{center}\label{tab:TBR}
\end{table}

\begin{figure}[htbp]
\hfill
\begin{minipage}[t]{.48 \textwidth}
\centering
\includegraphics[width = \textwidth]{TBRII_geometry.eps}
\caption{\small{Schematic of Tris(bipyridine) ruthenium complex.}}
\label{fig:TBRGeom}
\end{minipage}
\hfill
\begin{minipage}[t]{.48 \textwidth}
\centering
\includegraphics[width =\textwidth]{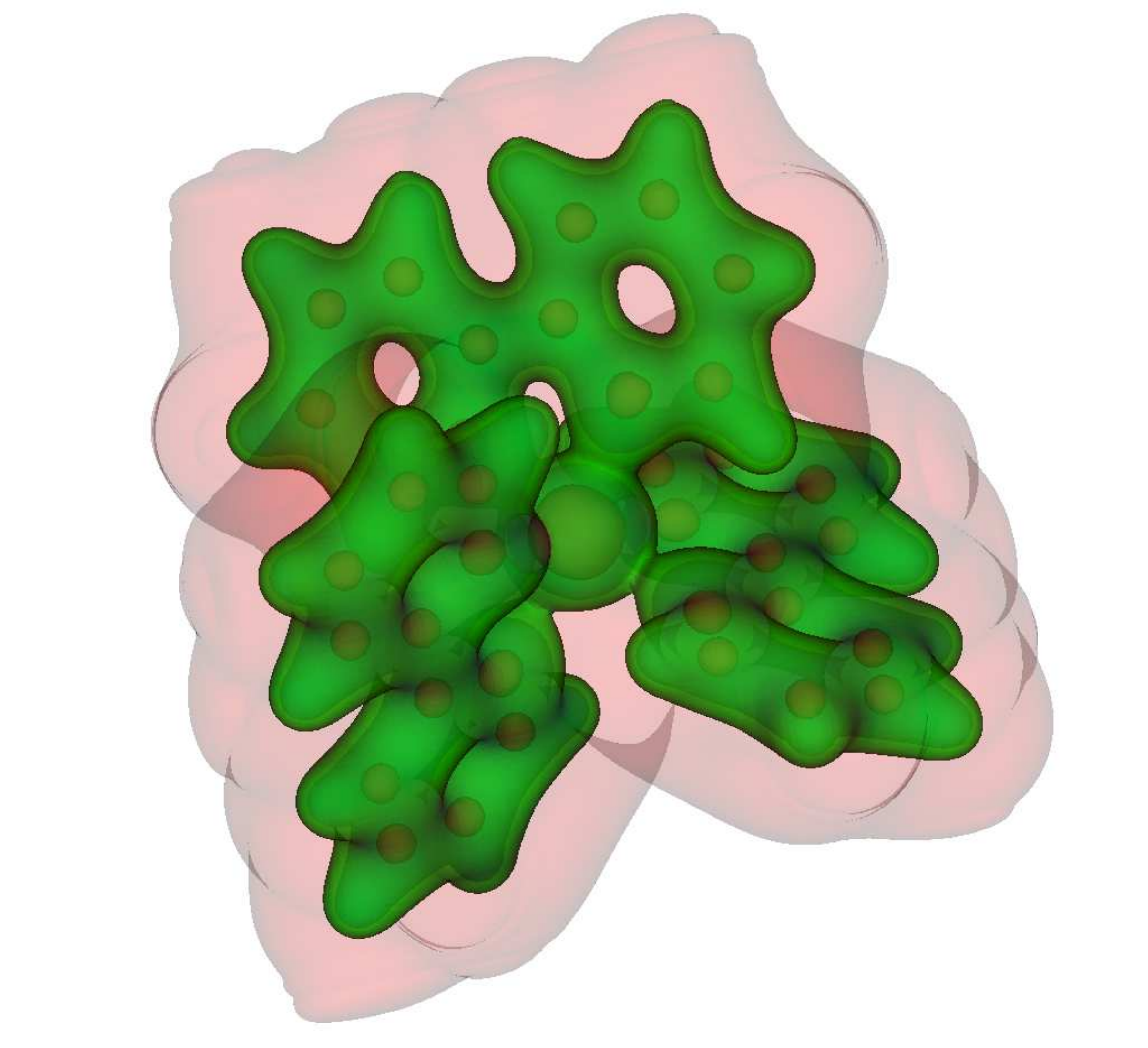}
\caption{\small{Electron density contours of Tris (bipyridine) ruthenium complex.}}
\label{fig:TBR}
\end{minipage}
\end{figure}

We note that the Gaussian basis sets employed in these benchmark studies are highly optimized for specific material systems, which is reflected in the far fewer basis functions required for the above calculations. We believe this is the main reason for the superior performance of Gaussian basis. We also note that the computational time using finite-element basis functions can possibly be reduced significantly by enriching the finite-element shape functions with single atom wavefunctions using the partitions-of-unity approach~\cite{PUFEM1,PUFEM2}. The degree of freedom advantage of the partitions-of-unity approach for Kohn-Sham calculations has been first demonstrated in~\cite{Suku}, and presents a promising future direction for all-electron Kohn-Sham DFT calculations.

\subsection{Scalability of finite-element basis:}\label{sec:scalability}
The parallel scalability of our numerical implementation is demonstrated in Figure~\ref{fig:scalability}. We study the strong scaling behavior by measuring the relative speedup with increasing number of processors on a fixed problem of constant size, which is chosen to be the aluminum $3\times 3\times 3$ cluster discretized with HEX125SPECT elements containing 3.91 million degrees of freedom. The speedup is measured relative to the computational CPU-time taken on 2 processors, as a single processor run was beyond reach due to memory limitations. It is evident from the figure, that the scaling is almost linear. The relative speedup corresponding to 96-fold increase in the number of processors is 87.82, which translates into an efficiency of 91.4\%.
\begin{figure}[htbp]
\begin{center}
\includegraphics[width=0.5\textwidth]{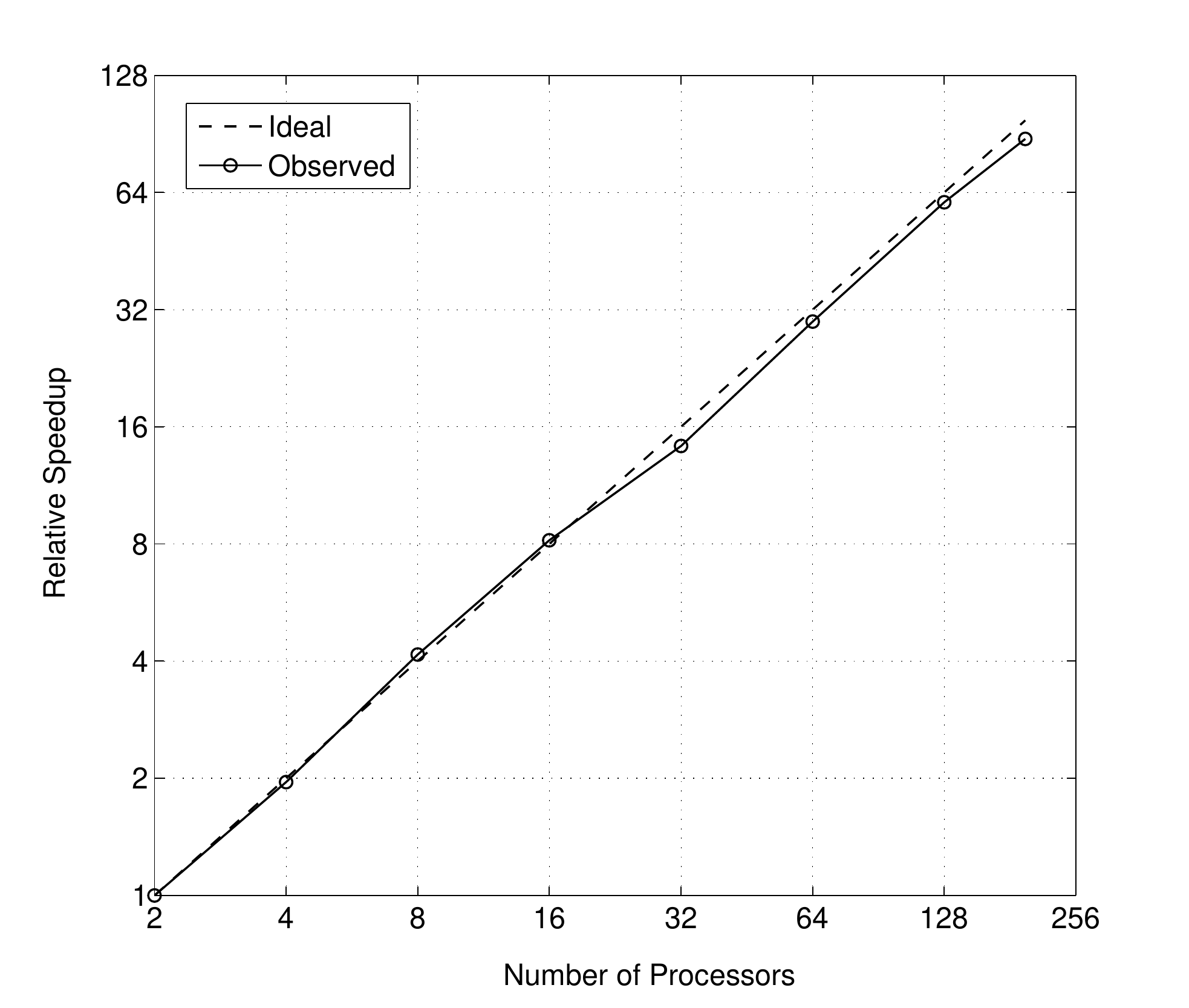}
\captionof{figure}{Relative speedup as a function of the number of processors.}
\label{fig:scalability}
\end{center}
\end{figure} 
\section{Conclusions}\label{concl}
In the present study, we have analyzed numerically the higher-order adaptive finite-element discretization of the Kohn-Sham DFT problem. The present work is focussed towards demonstrating the significant computational efficiency in electronic structure calculations that is afforded by using an adaptive higher-order spectral finite-element discretization in conjunction with appropriate solution strategies. We use the self-consistent field formulation of the Kohn-Sham DFT problem as our starting point. In order to aid our investigation, we first developed estimates for the discretization error in the ground-state energy in terms of the ground-state electronic fields (wavefunctions and electrostatic potential) and characteristic mesh-size. These error estimates and the {\it a priori} knowledge of the asymptotic solutions of far-field electronic fields were used to construct mesh coarsening rates for the various benchmark problems considered in this work. Since the finite-element discretization of the Kohn-Sham problem results in a generalized eigenvalue problem, which is computationally expensive to solve, we presented an approach to trivially transform this into a standard eigenvalue problem by using spectral finite-elements in conjunction with the Gauss-Lobatto quadrature rules that results in a diagonal overlap matrix. We subsequently examined two different strategies to solve the Kohn-Sham problem: (i) explicit computation of eigenvectors at every self-consistent field iteration; (ii) a Chebyshev filtering approach that directly computes the occupied eigenspace. Our investigations suggest that the use of spectral finite-elements and Gauss-Lobatto rules in conjunction with Chebyshev acceleration techniques to compute the eigenspace gives a $10-20$ fold computational advantage, even for modest materials system sizes, in comparison to traditional methods of solving the standard eigenvalue problem where the eigenvectors are computed explicitly. Further, the proposed approach has been observed to provide a staggering $100-200$ fold computational advantage over the solution of a generalized eigenvalue problem that does not take advantage of the spectral finite-element discretization and Gauss-Lobatto quadrature rules.

Using the derived error estimates and the {\it a priori} knowledge of the asymptotic solutions of far-field electronic fields, we constructed close to optimal finite-element meshes for the various benchmark problems, which include all-electron calculations on systems comprising of boron atom and methane molecule, and local pseudopotential calculations on barium cluster and bulk calcium crystal. We employed the Chebyshev filtering approach on the transformed standard eigenvalue problem in our numerical investigations to study the computational efficiency of higher-order finite-element discretizations. To this end, we first investigated the performance of higher-order elements by studying the convergence rates of linear tetrahedral element and hexahedral spectral-elements up to sixth-order. In all the benchmark problems considered, we observed close to optimal rates of convergence for the finite-element approximation in the ground-state energy. Importantly, we note that optimal rates of convergence were obtained for all orders of finite-element approximations, considered in this work, even for all-electron Kohn-Sham DFT calculations with Coulomb-singular potentials, the mathematical analysis of which, to the best of our knowledge, is an open question to date and has not been numerically demonstrated elsewhere.

We further investigated the computational efficiency afforded by the use of higher-order finite-elements up to eighth-order spectral-elements. To this end, we used the mesh coarsening rates determined from the proposed mesh adaption scheme and studied the CPU time required to solve the benchmark problems. Our results demonstrate that significant computational savings can be realized by using higher-order elements. For instance, a staggering $1000-$fold savings in terms of CPU-time are realized by using sixth-order hexahedral spectral-element in comparison to linear tetrahedral element. We also note that the point of diminishing returns in terms of computational efficiency was determined to be around sixth-order for the benchmark systems we examined. The degree of freedom advantage of higher-order finite-elements is nullified by the increasing per basis-function costs beyond this point. To further assess the performance of higher-order finite-elements, we extended our investigations to study large materials systems and compared the computational CPU-time with commercially available plane-wave and Gaussian basis codes. We first conducted simulations on aluminium clusters with local pseudopotential containing 172 atoms and 666 atoms using sixth-order spectral-element in our implementation, as well as, the plane-wave basis in ABINIT package solved to a similar relative accuracy in the ground-state energy. These studies showed that the computational CPU-time required for the finite-element simulations is lesser in comparison to plane-wave basis sets underscoring the fact that higher-order finite-elements can compete with plane-waves, at least in non-periodic settings, when employed in conjunction with efficient solution strategies. Furthermore, we were able to compute the electronic structure of an aluminium cluster containing $1,688$ atoms by employing the sixth-order spectral-element, which was not possible using ABINIT due to large memory requirements. Next, we examined the computational efficiency in the case of all-electron calculations on a graphene sheet containing 100 atoms and tris (bipyridine) ruthenium complex containing 61 atoms. The all-electron calculations were conducted using Gaussian DFT basis sets and the fourth-order spectral-element basis, and we observed that the computational time for the finite-element basis was $10-$fold greater than the Gaussian basis.

The prospect of using higher-order spectral finite-elements as basis functions, in conjunction with the proposed solution strategies, for Kohn-Sham DFT electronic structure calculations is indeed very promising. While finite-elements have the advantages of handling complex geometries and boundary conditions and exhibit good scalability on massively parallel computing platforms, their use has been limited in electronic structure calculations as their computational efficiency compared unfavorably to plane-wave and Gaussian basis functions. The present study shows that the use of higher-order discretizations can alleviate this problem, and presents a useful direction for electronic structure calculations using finite-element discretization. Further, the computational cost in the case of all-electron calculations can be further reduced by enriching the finite-element shape functions with single-atom wavefunctions, and this is currently being studied. Last, but not the least, the implications of using higher-order spectral finite-element approximations in the development of a linear scaling Kohn-Sham DFT formulation is a worthwhile subject for future investigation. 
\section*{Acknowledgements}
We thank Dr. John Pask at Lawrence Livermore National Laboratory for the many useful discussions on this work. We gratefully acknowledge the support of National Science Foundation under Grant No. 1053145 and the Army Research Office under Grant No. W911NF-09-0292. V.~G. also gratefully acknowledges the support of Air Force Office of Scientific  Research under Grant No. FA9550-09-1-0240 and FA9550-13-1-0113 under the auspices of which part of the computational framework was developed. V.~G. also acknowledges the Alexander von Humboldt Foundation through a research fellowship, and is grateful to the hospitality of the Max-Planck Institute for Mathematics in Sciences and the Institute of Applied Mechanics at University of Stuttgart while completing this work. This work used the Extreme Science and Engineering Discovery Environment (XSEDE), which is supported by National Science Foundation grant number OCI-1053575.

\appendix
\section{Discrete formulation of electrostatic interactions in all-electron calculations}\label{self_energy}
The electrostatic interaction energy in the discrete formulation is given by
\begin{align}\label{eqn:self_energy1}
E^h_{\text{electrostatic}} &= \left[-\frac{1}{8\pi}\intomega |\del \phi_{h}(\bx)|^2 \dx + \intomega (\rho_{h}(\bx) + b(\bx))\phi_{h}(\bx)\dx\right] - E^h_{\text{self}} \quad \mbox{where}\\
E^h_{\text{self}} &= \frac{1}{2}\sum_{I=1}^{M}\intomega Z_{I} \delta(\bx-\bR_{I}){V_{h_{I}}}(\bx) \dx\,\,,
\end{align}
where $\phi_h$ denotes the total electrostatic potential field, corresponding to the electron-density $\rho_h$ and nuclear charge distribution $b(\bx)$, computed in the finite-element basis.  The nuclear potential corresponding to the $I^{th}$ nuclear charge, i.e $Z_{I}\delta(\bx-\bR_{I})$, computed in the finite-element basis is denoted by $V_{h_{I}}$. The nuclear charges located on the nodes of the finite-element triangulation are treated as point charges and the discreteness of the finite-element triangulation provides a regularization of the potential fields. However, the self-energy of the nuclei in this case is mesh-dependent and diverges upon mesh refinement. Thus, care must be taken to evaluate the total electrostatic potential $\phi_h$ and the nuclear potentials $V_{h_{I}}$, $I=1,\ldots M$ on the same finite-element mesh. The electrostatic interaction energy in equation~\eqref{eqn:self_energy1} can be simplified to
\begin{equation}\label{eqn:self_energy2}
E^h_{\text{electrostatic}} =  \frac{1}{2} \intomega \rho(\bx) \phi_{h}(\bx)\dx + \underbrace{\frac{1}{2}\intomega b(\bx)\phi_{h}(\bx) \dx}_{(a)}  - \underbrace{\frac{1}{2}\sum_{I=1}^{M}\intomega \delta(\bx-\bR_{I}){V_{h_{I}}}(\bx) \dx}_{(b)}\,.
\end{equation}
In the above expression, the first term on the righthand side contains the contribution of  electron-electron interaction energy and half contribution of the electron-nuclear interaction energy. The term ($a$) contains the other half of the electron-nuclear interaction energy, nuclear-nuclear repulsion energy, and the self energy of the nuclei. The term ($b$) represents the self energy of the nuclei. By evaluating all the electrostatic potentials on the same finite-element mesh, the divergent self energy contribution in term ($a$) equals the sum of separately evaluated divergent self-energy terms in ($b$) owing to the linearity of the Poisson problem. The boundary conditions used for the computation of the discrete potential fields are homogeneous Dirichlet boundary conditions for total electrostatic potential ($\phi_h$) and Dirichlet boundary conditions with the prescribed Coulomb potential for nuclear potentials ($V_{h_{I}}$), applied on a large enough domain where the boundary conditions become realistic. The numerical results we present below show that the diverging components of self energy in terms ($a$) and ($b$) indeed cancel. To this end, we present the case study of the electrostatic interaction energy computed for a methane molecule with the geometry as described in section 5.1.1. The electron-density $\rho(\bx)$  is chosen to be the superposition of the distributions computed from equation~\eqref{carbon} with $\xi$ equal to $0.83235$ and equation~\eqref{hydrogen}, and normalized to the number of electrons in the system. We choose  a sequence of refined meshes obtained by uniform subdivision of initial coarse mesh with HEX27 and HEX125SPECTRAL elements. The results in tables~(\ref{tab:self_energy1}) and ~(\ref{tab:self_energy2}) show that while terms ($a$) and ($b$) diverge upon mesh refinement, the electrostatic energy nevertheless converges suggesting the cancelation of divergent self energy terms. Figure~(\ref{fig:ElectrostaticsConvergence}) shows the convergence rates for the electrostatic energy which are close to optimal. The value of $E_0$ in the above plot is obtained using the extrapolation procedure as discussed in section 5.1 and is found to be $-23.79671760794$.

\begin{table}[h]
\caption{Convergence of $E^h_{\text{electrostatic}}$ for ``HEX27'' element}
\centering
\begin{tabular}{|c|c|c|c|}
  \hline
  Degrees of Freedom & Term ($a$)  & Term ($b$) & $E^h_{electrostatic}$ \\ \hline\hline
  13059 &1637.011830893&1665.4003185717& -22.77175242597 \\ \hline
  96633 & 3641.626361382&3657.972341778&-23.7628285436  \\ \hline
  765,041 &7299.84650294 &7316.0488206578&-23.7936738766  \\ \hline
  6,090,465 &14534.01973132&14550.219757955&-23.7964615239\\ \hline
  48,608,705 &29248.01834776&29264.218261189&-23.7966991925\\ \hline \hline
\end{tabular}
\label{tab:self_energy1}
\end{table}

\begin{table}[h]
\caption{Convergence of $E^h_{\text{electrostatic}}$ for ``HEX125SPECTRAL'' element}
\centering
\begin{tabular}{|c|c|c|c|}
  \hline
  Degrees of Freedom & Term ($a$)  & Term ($b$) & $E^h_{electrostatic}$ \\ \hline\hline
  64841 & 1995.473107 & 2011.773736  & -23.5282153593 \\ \hline
  510,993 & 4003.728396 & 4019.928861 & -23.7965792134 \\ \hline
  4,058,657 & 8023.635544 & 8039.835447 & -23.79670460619 \\ \hline
  32,355,393 & 16063.4709881 & 16079.670897 & -23.79671759075 \\ \hline \hline
\end{tabular}
\label{tab:self_energy2}
\end{table}

\begin{figure}[htbp]
\begin{center}
\includegraphics[width=0.5\textwidth]{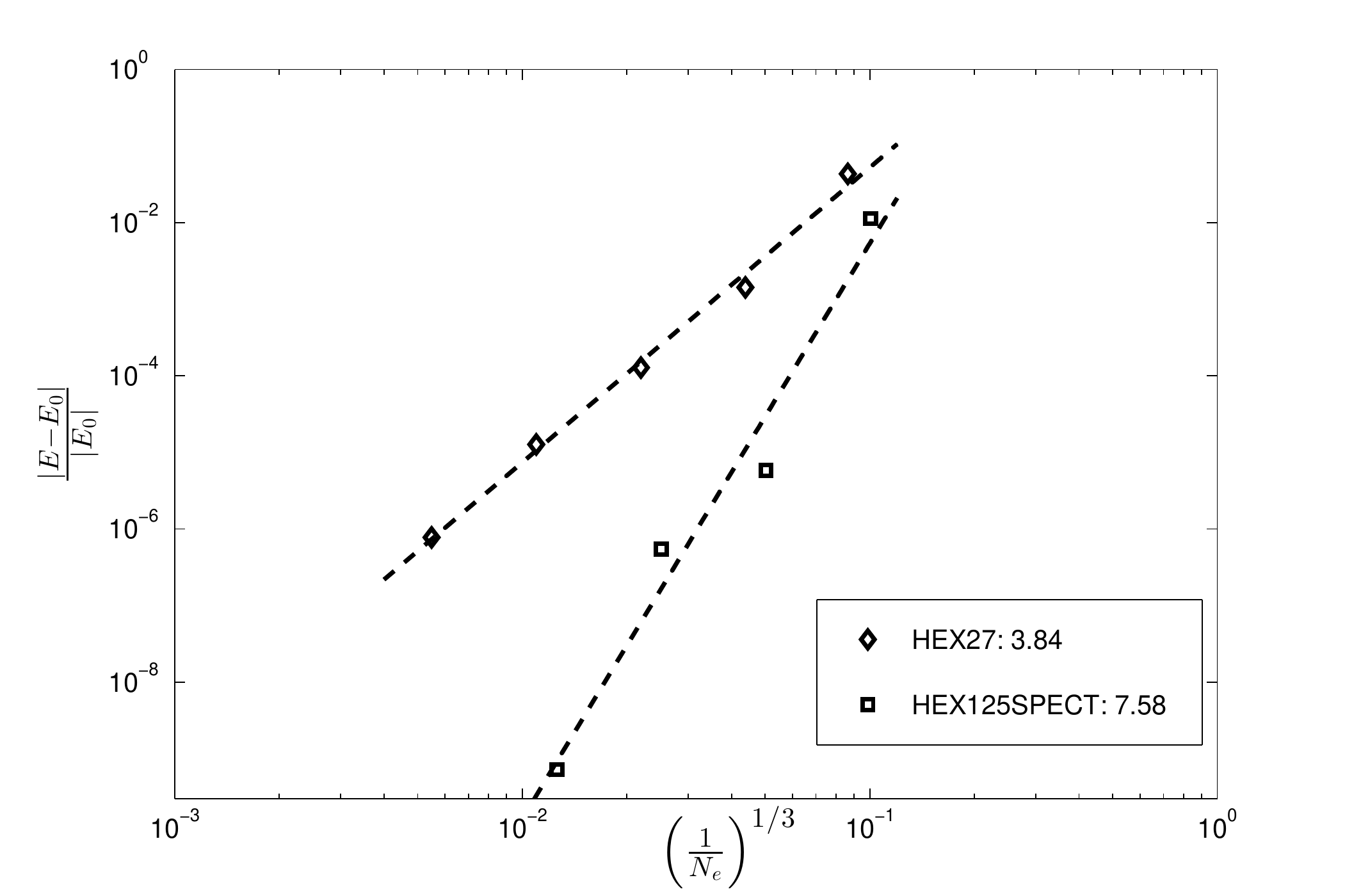}
\captionof{figure}{Convergence rates of electrostatic energy for the finite-element approximations. Case study: Methane molecule.}
\label{fig:ElectrostaticsConvergence}
\end{center}
\end{figure}

\section{Trade-offs in higher-order elements: Source of diminishing returns with increasing order}
In section 5.2.1, we observed that for relative errors commensurate with chemical accuracies, the computational efficiency improves significantly with the order of element up to sixth order, but with diminishing returns beyond this for the benchmark problems considered.  Here we identify the source of diminishing returns on $2\times2\times2$ barium cluster. To this end, we choose three finite-element meshes containing HEX125SPECT, HEX343SPECT and HEX729SPECT elements which give a relative discretization error in the ground-state energy of the order of $10^{-5}$. Table ~(\ref{tab:comp_times}) shows the computational cost (measured in terms of CPU-mins) incurred in building the Hamiltonian matrix and the matrix-vector multiplications involved in a single SCF iteration.
\begin{table}[h]
\caption{Computational cost per SCF iteration. Case study: Barium cluster }
\centering
\begin{tabular}{|l|l|p{4cm}|p{3cm}|p{2.5cm}|}
  \hline
Type of element & Degrees of freedom & Hamiltonian matrix construction ($t_1$ mins)  & Matrix vector multiplication ($t_2$ mins)  & Total time ($t_1+t_2$ mins)\\ \hline\hline
  HEX125SPECT & 667,873 & 18.83 &  15.03 & 33.86\\ \hline
  HEX343SPECT & 143,989 & 21.91 &  3.79  &25.70 \\ \hline
  HEX729SPECT & 41,825 & 25.99 & 1.64 &27.63 \\ \hline\hline
\end{tabular}
\label{tab:comp_times}
\end{table}
We observe from the table~(\ref{tab:comp_times}) that while there is significant reduction in the number of basis functions with increasing polynomial degree to achieve the same relative accuracy, there is no computational savings obtained by using octic element over hexic element due to the increase in the computational cost involved in building the Hamiltonian matrix which increasingly dominates the total time with increasing order of the element. The cost of computing the Hamiltonian matrix depends on the number of basis functions per element and the order of the quadrature rule, both of which increase with increasing order and cannot be mitigated. However, we remark that, for large enough systems (in terms of number of electrons) the orthogonalization of the Chebyshev filtered vectors will become the dominant cost in a SCF iteration, at which point the order of the polynomial beyond which diminishing returns will be observed can move to a polynomial order beyond the sixth order. But, in the present study, for the range of systems considered, this point has not been reached where the orthogonalization step is the dominant cost in a SCF iteration.

\section{Accuracy of Gauss-Lobatto-Legendre quadrature}
As presented in section 4, employing the Gauss-Lobatto-Legendre (GLL) quadrature---a reduced-order quadrature rule---for the overlap matrix corresponding to spectral elements results in a standard eigenvalue problem which can be very effectively solved using the Chebyshev filtering technique. Here, we present the results from benchmark problems to establish the accuracy of this reduced-order quadrature employed in the computation of the overlap matrix. To begin with, we consider the Hydrogen atom which represents the simplest example in the all-electron Kohn-Sham DFT problem. We consider a sequence of finite-element meshes on a spherical domain of radius~ $20~a.u.$ employing HEX125SPECTRAL elements that are uniform subdivisions of the coarsest mesh. The ground-state energies obtained by employing the GLL quadrature rule for the overlap matrix are presented in table~\ref{tab:hyd}, which demonstrates the convergence of the ground-state energies. Further, figure~\ref{fig:hydrogenConvgRate} demonstrates the (near) optimal rate of convergence of the ground-state energies computed employing equation~(\ref{convfit}). The obtained value of $E_0$ is $-0.50000000000926~Ha$ and differs from the theoretical value in the 12th significant digit. This difference can be attributed to the finite size of the domain and the finite precision tolerances used in the solution of the eigenvalue problem.

We subsequently used two benchmark problems---methane molecule (all-electron calculation) and $2\times2\times2$ barium cluster (local pseudopotential calculation)---to compare the ground-state energies obtained by employing GLL quadrature rules for the overlap matrix with those obtained by employing Gauss quadrature rules. For each of these benchmark systems, we considered a coarse and a relatively fine mesh discretization for different orders of discretization. These results are tabulated in table~\ref{tab:GLL}. We note that the absolute difference in ground-state energies per atom between GLL and Gauss quadrature rules, for both the systems and for the different meshes considered, is about an order of magnitude smaller than the discretization error (reference values in sections 5.1.1.2 and 5.1.2.1). These results demonstrate the accuracy and sufficiency of GLL quadrate rules for the computation of the overlap matrix.

We further note a recent numerical analysis~\cite{hughes2010} which investigates the error in the eigenspectrum of second-order linear differential operators due to discretization and reduced-order quadrature. While this analysis was not the main objective of this work, it comprises of results that presents a qualitative understanding of the sufficiency of reduced-order quadrature rules for the Kohn-Sham DFT problem. The results in figure A2 in ~\cite{hughes2010} demonstrate that reduced-order quadratures introduce errors in the higher-end of the spectrum, where $C^0$ finite-elements are anyway no longer accurate even with exact integration and result in spurious optical modes, but are accurate for the lower-end of the eigenspectrum. The ground-state properties in the Kohn-Sham DFT are solely governed by the lower-end of the eigenspectrum, which provides a qualitative explanation for the observed accuracy of the GLL quadrature.

\begin{table}[h]
\caption{Computed ground-state energies of Hydrogen atom by employing GLL quadrature rules for overlap matrix.}
\centering
\begin{tabular}{|c|c|c|c|c|}
  \hline
 Degrees of freedom (DoF) & Ground state  energy ($Ha$)  \\ \hline\hline
  17,713 &  -0.499894312878\\ \hline
  140,257  &  -0.499999964823\\ \hline
  1,117,634& -0.4999999999076\\ \hline
  8,926,245 & -0.50000000000912\\ \hline \hline
  \end{tabular}
\label{tab:hyd}
\end{table}

\begin{figure}[htbp]
\begin{center}
\includegraphics[width=0.5\textwidth]{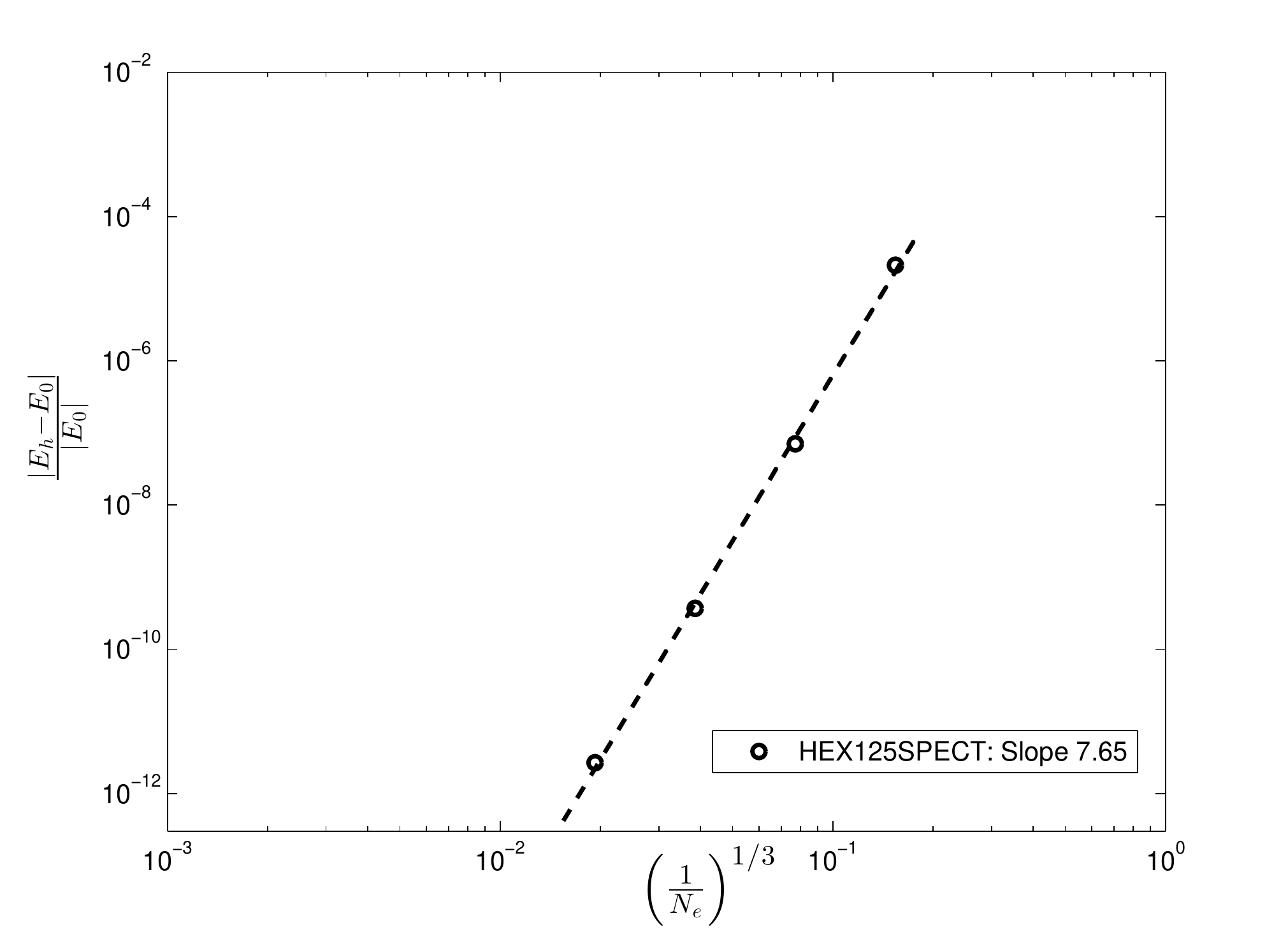}
\captionof{figure}{Convergence of the finite-element approximation for Hydrogen atom using GLL quadrature rule for overlap matrix.}
\label{fig:hydrogenConvgRate}
\end{center}
\end{figure}

\begin{table}[htbp]
\caption{Comparison between Gauss-Lobatto-Legendre (GLL) quadrature rule and Gauss quadrature (GQ) rule }
\centering
\begin{tabular}{|c|c|c|c|c|}
  \hline
Type of System & Element Type & DoF &  Energy/atom (GLL)  &  Energy/atom (GQ) \\ \hline\hline
  Methane & HEX27 & 18,509 & -7.9989600253$Ha$ & -8.0030290831 $Ha$ \\ \hline
  Methane & HEX27 & 317,065  & -8.0215895393$Ha$ & -8.0220114249 $Ha$ \\ \hline
  Methane & HEX125SPECT & 43,289 &-8.0065360952 $Ha$&-8.0044564654$Ha$  \\ \hline
  Methane & HEX125SPECT  & 289,401  & -8.0239636665 $Ha$ & -8.0239659379 $Ha$\\ \hline
  Barium 2x2x2 cluster  & HEX27 & 175,101 & -0.64013198302 $Ha$ & -0.6403673302 $Ha$ \\ \hline
  Barium 2x2x2 cluster  & HEX27 & 2,379,801 & -0.63858359722 $Ha$& -0.6385901453 $Ha$ \\ \hline
  Barium 2x2x2 cluster  & HEX343SPECT & 57,121 & -0.6373331092$Ha$ & -0.6374840072$Ha$  \\ \hline
  Barium 2x2x2 cluster  & HEX343SPECT & 449,473  & -0.6386270069 $Ha$ & -0.6386263119 $Ha$\\ \hline\hline
  \end{tabular}
\label{tab:GLL}
\end{table}

\end{document}